%% file: main.tex
\documentclass[11pt,letterpaper]{article}
\usepackage{amsmath,amssymb,amsfonts,amsthm}
\usepackage{mathrsfs}
\usepackage{mathtools}
\usepackage{thm-restate}
\usepackage{tikz}
\usepackage{enumitem}
\usepackage{bbm}
\usepackage{xcolor}
\usepackage{extpfeil}
\usepackage[ruled,vlined,linesnumbered]{algorithm2e}
\usepackage[margin=1in]{geometry}
\usepackage[colorlinks, citecolor = blue, linkcolor = magenta]{hyperref}
\usepackage{cleveref}
\Crefname{subappendix}{Section}{Sections}
\usepackage[most]{tcolorbox}   
\tcbuselibrary{skins, breakable}
\numberwithin{equation}{section}
\input{macros}

\title{Near-Optimal Space Lower Bounds for Streaming CSPs}
\author{Yumou Fei\thanks{Department of EECS, Massachusetts Institute of Technology.}\and Dor Minzer\thanks{Department of Mathematics, Massachusetts Institute of Technology. Supported by NSF CCF award 2227876 and NSF CAREER award 2239160.}\and Shuo Wang\thanks{Department of Mathematics, Massachusetts Institute of Technology. Supported by NSF award 2239160.}}
\date{\vspace{-5ex}}

\begin{document}
\maketitle
\begin{abstract}
    In a streaming constraint satisfaction problem (streaming CSP), a $p$-pass algorithm receives the constraints of an instance sequentially, making $p$ passes over the input in a fixed order, with the goal of approximating the maximum fraction of satisfiable constraints. We show near optimal space lower bounds for streaming CSPs, improving upon prior works:
    \begin{enumerate}
        \item Fei, Minzer and Wang (\textit{STOC 2026}) showed that for any CSP, the basic linear program defines a threshold $\alpha_{\mathrm{LP}}\in [0,1]$ such that, for any $\varepsilon > 0$, an $(\alpha_{\mathrm{LP}} - \varepsilon)$-approximation can be achieved using constant passes and polylogarithmic space, whereas achieving $(\alpha_{\mathrm{LP}} + \varepsilon)$-approximation requires $\Omega(n^{1/3}/p)$ space. We improve this lower bound to $\Omega(\sqrt{n}/p)$, which is nearly tight for a gap version of the problem.
        
        \item  For $p=o(\log n)$, we further strengthen the lower bound to $\Omega(n\cdot2^{-O_{\varepsilon}(p)})$. Combined with existing algorithmic results, this shows that $\alpha_{\mathrm{LP}}$ is not only the limit of multi-pass polylogarithmic-space algorithms, but also the limit of single-pass sublinear-space algorithms on bounded-degree instances.

        \item For certain CSPs, we show that there exists $\alpha < 1$ such that achieving an $\alpha$-approximation requires $\Omega(n/p)$ space. 
    \end{enumerate}
    Our proofs are Fourier analytic, building on the techniques of Fei, Minzer and Wang (\textit{STOC 2026}) and the Fourier-$\ell_1$-based lower bound method of Kapralov and Krachun (\textit{STOC 2019}). 
\end{abstract}

\thispagestyle{empty}
\setcounter{page}{0}
\newpage
{
\setcounter{tocdepth}{3} 
\tableofcontents
}
\newpage
\setcounter{page}{1}

\input{introduction}

\input{prelim}

\input{communication_game}

\input{communication_hardness}

\input{decomposition}

\input{discrepancy}

\input{induction_lemma}

\input{twowise_independence}

\addcontentsline{toc}{section}{References}
\bibliography{reference}
\bibliographystyle{alpha}
\appendix
\input{appendices}

\end{document}

%% file: macros.tex
\newtheorem{theorem}{Theorem}[section]

\newtheorem{proposition}[theorem]{Proposition}
\newtheorem{corollary}[theorem]{Corollary}
\newtheorem{lemma}[theorem]{Lemma}

\theoremstyle{definition}
\newtheorem{definition}[theorem]{Definition}

\newtheorem{remark}[theorem]{Remark}
\newtheorem{notation}[theorem]{Notation}

\allowdisplaybreaks

\newcommand{\bad}{\mathrm{bad}}
\newcommand{\leaf}{\mathrm{leaf}}

\renewcommand{\leq}{\leqslant}
\renewcommand{\le}{\leqslant}
\renewcommand{\geq}{\geqslant}
\renewcommand{\ge}{\geqslant}

\renewcommand{\d}{\,\mathrm{d}}

\newcommand{\calB}{\mathcal{B}}
\newcommand{\calC}{\mathcal{C}}
\newcommand{\calD}{\mathcal{D}}
\newcommand{\calE}{\mathcal{E}}
\newcommand{\calF}{\mathcal{F}}

\newcommand{\calM}{\mathcal{M}}

\newcommand{\calR}{\mathcal{R}}
\newcommand{\calS}{\mathcal{S}}

\newcommand{\calU}{\mathcal{U}}
\newcommand{\calV}{\mathcal{V}}
\newcommand{\calI}{\mathcal{I}}
\newcommand{\calX}{\mathcal{X}}

\newcommand{\bC}{\mathbb{C}}
\newcommand{\bE}{\mathbb{E}}
\newcommand{\bF}{\mathbb{F}}
\newcommand{\bN}{\mathbb{N}}
\newcommand{\bP}{\mathbb{P}}
\newcommand{\bR}{\mathbb{R}}
\newcommand{\bZ}{\mathbb{Z}}

\newcommand{\bfa}{\mathbf{a}}

\newcommand{\bfy}{\mathbf{y}}
\newcommand{\bfz}{\mathbf{z}}
\newcommand{\bfP}{\mathbf{P}}
\newcommand{\bfR}{\mathbf{R}}
\newcommand{\bfY}{\mathbf{Y}}
\newcommand{\bfZ}{\mathbf{Z}}

\newcommand{\bdP}{\boldsymbol{P}}
\newcommand{\bdR}{\boldsymbol{R}}
\newcommand{\bdzeta}{\boldsymbol{\zeta}}
\newcommand{\bdxi}{\boldsymbol{\xi}}

\newcommand{\sfe}{\mathsf{e}}
\newcommand{\sfv}{\mathsf{v}}
\newcommand{\sfH}{\mathsf{H}}
\newcommand{\sfW}{\mathsf{W}}

\newcommand{\fraki}{\mathfrak{i}}

\newcommand{\yes}{\mathrm{yes}}
\newcommand{\no}{\mathrm{no}}

\newcommand{\nil}{\textup{\texttt{nil}}}

\newcommand{\proj}{\mathrm{proj}}

\newcommand{\ZmodN}{\bZ_{N}}
\newcommand{\ZNk}{\bZ_{N}^{k}}

\newcommand{\bd}{\mathsf{bd}}
\newcommand{\itn}{\mathsf{in}}

\newcommand{\supp}{\mathrm{supp}}   

\newcommand{\lp}{{\textsc{BasicLP}}}
\newcommand{\lpval}{\mathrm{val}^{\mathrm{LP}}}
\newcommand{\val}{\mathrm{val}}

\newcommand{\cspF}{\mathrm{CSP}(\mathcal{F})}

\newcommand{\McspF}[2]{\mathsf{MaxCSP}(\calF)[#1, #2]}

\newcommand{\Map}[2]{\mathrm{Map}\left(#1, #2\right)}

\newcommand{\ind}[1]{\mathbbm{1}\left\{#1\right\}}

\newcommand{\tcup}{\textstyle\bigcup}
\newcommand{\tprod}{\textstyle\prod}

\newcommand{\DIHP}{\mathsf{DIHP}}
\newcommand{\CC}{\mathsf{CC}}

\newcommand{\reff}{\mathrm{ref}}
\newcommand{\adv}{\mathrm{adv}}

\newcommand{\Ex}[1]{\bE \left[ #1 \right]}
\newcommand{\Exu}[2]{\underset{#1} \bE \left[ #2 \right] }

\newcommand{\Exs}[2]{\bE_{#1}\left[ #2 \right]}
\renewcommand{\Pr}[1]{\bP \left[ #1 \right]} 
\newcommand{\Prs}[2]{\bP_{#1} \left[ #2 \right]}
\newcommand{\Pru}[2]{\underset{ #1 }\bP \left[ #2 \right]}

\newcommand{\inp}[2]{\left\langle #1, #2 \right\rangle}

%% file: introduction.tex
\section{Introduction}
Constraint satisfaction problems (CSPs) are among the most well studied problems in theoretical computer science. Traditionally, they have been studied in the context of polynomial time computation, and by now the tractability of their decision version~\cite{Bulatov,Zhuk}, as well as of their approximation version~\cite{Raghavendra}, are (for the most part) well understood. Over the last decade CSPs have gained significant interest in the streaming community, with the goal of understanding what approximation ratios efficient streaming algorithms can achieve. This is the subject of the current paper, and our main result is a quantitatively stronger (and nearly optimal in some sense) space lower bound for the classes of single pass and multi-pass streaming algorithms for CSPs.

\subsection{Constraint Satisfaction Problems in the Streaming Model}
A single-pass, space $S$ streaming algorithm receives its input as a stream, and upon viewing each element it is allowed to modify its memory arbitrarily. At the end of the stream, the algorithm is supposed to output an answer.
The study of CSPs in the streaming setting started in~\cite{KKS15}, who considered the max-cut problem. In this problem, the stream consists of the edges of some graph $G$, and the algorithm is supposed to estimate the size of the largest cut in $G$. There is a trivial $O(\log n)$ space algorithm achieving approximation ratio $1/2$, and the work~\cite{KKS15} shows that for all $\varepsilon>0$, achieving approximation ratio $1/2+\varepsilon$ requires $\Omega_{\varepsilon}(\sqrt{n})$ memory.

The work of~\cite{KKS15} inspired a flourish of activity in the study of the performance of streaming algorithms on optimization problems; see~\cite{guruswami2017streaming,KK19,chou2020optimal,AKSY20,AN21,chou2022linear,saxena2023improved,hwang2024oblivious,CGSV24,saxena2025streaming,FMWa,FMW25b,STV,Velus,ABFS} and~\cite{SudanSurvey,Assadi,Singer} for surveys. All of these works consider constraint satisfaction problems, which is a rich class of problems extending the max-cut problem, defined as follows.

\begin{definition}\label{def:CSPs}
Let $k\in\mathbb{N}$ be a parameter, $\Sigma$ be a finite alphabet and ${\cal{F}}\subseteq \{ f: \Sigma^k\rightarrow \{0,1\}\}$ be a family of predicates. An instance $\calI=(\calV, \calC)$ of $\cspF$ (we write $\calI\in\cspF$) consists of a set of variables $\calV$ and a collection of constraints $\calC=(C_{1},\dots,C_{m})$. Each constraint $C_{i}$ is specified as $((\sfv_{i,1},\dots,\sfv_{i,k}),f_{i})$, where $\sfv_{i,j}\in \calV$ is a variable for $j=1,\ldots,k$, and $f_i\in \calF$ is a predicate. 
\end{definition}
Given an instance $\calI$ of $\cspF$, the goal is to find an assignment $\tau\colon \calV\to\Sigma$ satisfying as many of the constraints as possible, i.e., maximizing 
\[
\val_{\calI}(\tau) = \sum\limits_{i=1}^{m}f_i(\tau(\sfv_{i,1}),\ldots,\tau(\sfv_{i,k})).
\]
The value of the instance $\calI$ is defined as $\val_{\calI} = \max_{\tau}\val_{\calI}(\tau)$. We will sometimes discuss the degree of an instance, which is the maximum number of constraints a variable appears in.

All prior works on streaming CSPs either consider a specific problem (such as max-cut or max-directed-cut) or a class of constraint satisfaction problems. They then study the performance of a class of streaming algorithms on them (such as sketching algorithms, single-pass algorithms, mult-pass algorithms), both on the algorithmic front and the hardness front. For example, the work of~\cite{CGSV24} gave an exact characterization of the power of the class of sketching algorithms, which is a subclass of single-pass streaming algorithms. 
Another example, which is most relevant to our paper, is the work of~\cite{KK19}, who improved the result of~\cite{KKS15} and showed that any single-pass, $(1/2+\varepsilon)$-approximation streaming algorithm for max-cut requires $\Omega_{\varepsilon}(n)$ memory. In that work 
the authors show how to replace the $L_2$-based lower bound approach of~\cite{KKS15} by a Fourier-$\ell_1$-based lower bound approach, which gives better quantitative bounds.

\paragraph{Multi-pass algorithms:} the main topic of this paper is multi-pass streaming algorithms. Given an instance $\calI$ of a CSP, a $p$-pass streaming algorithm gets to view the constraints of $\calI$ in the same predetermined order for $p$ times, and it then must output an approximation of $\val_{\calI}$. A multi-pass streaming algorithm is considered efficient if both the pass complexity $p$ and the space complexity are small (by which one typically means poly-logarithmic). Building on~\cite{yoshida2011optimal,saxena2025streaming}, the work~\cite{FMW25b} proposed a general approximation algorithm based on a linear-programming relaxation of $\calI$ due to~\cite{yoshida2011optimal}, called $\lp_{\calI}$, defined as follows. The program has variables $(x_{\sfv,\sigma})_{\sfv\in \calV,\,\sigma\in \Sigma}$, thought of as specifying a probability distribution over the possible assignments to each variable $v\in\calV$, and variables $(z_{i,b})_{i\in [m],\, b\in \Sigma^{k}}$, thought of as specifying a probability distribution over the possible assignments to each clause.
\begin{tcolorbox}[
  enhanced,
  breakable,
  title={$\lp_{\calI}$ for $\calI=(\calV, \calC)$},
  fonttitle=\bfseries,
  colframe=black,
  colback=white,
  boxrule=0.6pt,
  arc=6pt,
  left=6pt,
  right=6pt,
  top=4pt,
  bottom=4pt
]
\begin{alignat*}{3}
\textup{maximize} \quad     &&  \frac{1}{m}\sum_{i=1}^{m} 
            \sum_{b\in\Sigma^{k}} f_i(b)\,z_{i,b} \\
\textup{subject to} \quad && \sum_{\sigma\in\Sigma} x_{\sfv,\sigma}&=1
            \quad &\forall\,\sfv\in \calV \\
           && \sum_{b\in\Sigma^{k}}\mathbbm{1}\{b_{j}=\sigma\}\cdot
             z_{i,b}&=x_{\sfv_{i,j},\,\sigma}
            \quad &\forall\,i\in [m],\,j\in [k],\,\sigma\in \Sigma\\
           && x_{\sfv,\sigma}&\ge 0
            \quad &\forall\,\sfv\in \calV,\,\sigma\in\Sigma \\
           && z_{i,b}&\ge 0
            \quad &\forall\,i\in [m],\,b\in\Sigma^{k}
\end{alignat*}
\end{tcolorbox}
The objective value counts the total mass put on assignment that satisfies the constraints of $\calI$, and the constraints enforce local consistency between the distributions described by the variables of the program. We denote by $\lpval_{\calI}$ the value of $\lp_{\calI}$.

The work~\cite{FMW25b} argued that for 
\begin{equation}\label{eq:def_alpha_LP}
\alpha_{\mathrm{LP}}(\calF):=\inf_{\calI\in \cspF} \frac{\val_{\calI}}{\lpval_{\calI}}
\end{equation}
and any $\varepsilon>0$, a streaming algorithm with $O_{\varepsilon}(1)$ many passes and $O_{\varepsilon}(\log n)$ memory can $\alpha_{\mathrm{LP}}$-approximate the values of $\cspF$ instances.\footnote{Here $n$ stands for the number of variables in the input instance. We additionally assume that the number of constraints is at most polynomial in $n$.} To get a matching hardness result, they showed that for any $\varepsilon>0$, a $p$-pass, space-$S$ streaming algorithm achieving approximation ratio $\alpha_{\mathrm{LP}}(\calF) + \varepsilon$ must satisfy $pS = \Omega_{\varepsilon}(n^{1/3})$. Thus, a sharp transition in the complexity of approximation $\cspF$ occurs at $\alpha_{\mathrm{LP}}(\calF)$, roughly from constant-pass/logarithmic-space, to polynomial. While this gap is very large, it is not optimal. 

\subsection{Main Results}

The main result of this paper is a near-optimal, quantitative improvement of the result of~\cite{FMW25b}. We state it in the form of gap problems: given parameters $1\geq c> s\geq 0$,  $\McspF{c}{s}$ is the promise problem wherein the input instance $\calI\in \cspF$ is promised to either have $\val_{\calI}\geq c$ or $\val_{\calI}\leq s$, and the goal is to distinguish between these two cases.
\begin{theorem}\label{thm:main}
Fix a nonempty instance $\calI\in \cspF$. Let $s:=\val_{\calI}$ and $c:=\lpval_{\calI}$. Then the following statements hold:
\begin{enumerate}[label = (\arabic*)]
\item If $c<1$, then for any fixed error parameter $\varepsilon\in (0,1)$, the memory size of any $p$-pass streaming algorithm for $\McspF{c-\varepsilon}{s+\varepsilon}$ is at least $\max\big\{\Omega_{\varepsilon}(\sqrt{n}/p),\;n\cdot \exp(-O_{\varepsilon}(p))\big\}$.
\item If $c=1$, then for any fixed error parameter $\varepsilon\in (0,1)$, the memory size of any $p$-pass streaming algorithm for $\McspF{1}{s+\varepsilon}$ is at least $\max\big\{\Omega_{\varepsilon}(\sqrt{n}/p),\;n\cdot \exp(-O_{\varepsilon}(p))\big\}$.
\end{enumerate}
\end{theorem}
\Cref{thm:main} immediately implies that for general pass complexity $p$, the lower bound of~\cite{FMW25b} can be improved to $pS = \Omega_{\varepsilon}(\sqrt{n})$:

\begin{corollary}\label{cor:alg_and_lb}
    For any predicate family $\calF\subseteq \{f:\Sigma^{k}\rightarrow \{0,1\}\}$ and any $\varepsilon>0$, any $p$-pass streaming algorithm with $S$ bits of memory achieving approximation ratio $\alpha_{\mathrm{LP}}(\calF)+\varepsilon$ (the threshold $\alpha_{\mathrm{LP}}(\calF)$ is defined in \eqref{eq:def_alpha_LP}) for $\cspF$ must satisfy $pS=\Omega_{\varepsilon}(\sqrt{n})$.
\end{corollary}

\begin{remark}
The threshold $\alpha_{\mathrm{LP}}(\calF)$ in Corollary~\ref{cor:alg_and_lb} is optimal due to the algorithm of \cite{FMW25b} mentioned earlier. We note that for the specific problems of max-cut and max-directed-cut, the threshold $\alpha_{\mathrm{LP}}$ equals $1/2$ (see \cite[Theorem 1.5]{FMW25b}). It is not clear whether the space-pass tradeoff $pS=\Omega_{\varepsilon}(\sqrt{n})$ in Corollary~\ref{cor:alg_and_lb} is optimal. However, we do know that the $\Omega(\sqrt{n} /p)$ lower bound in Theorem~\ref{thm:main}(2) cannot be improved by more than $\mathrm{polylog}(n)$ factors (see Section~\ref{subsubsec:sublinear_space_vs_time}). 
\end{remark}
For $p=o(\log n)$, the space lower bound in Theorem~\ref{thm:main} becomes near-linear in $n$. In particular, for single-pass streaming algorithms, we have the following corollary:

\begin{corollary}\label{cor:alg_and_lb_bd_deg}
    For any predicate family $\calF\subseteq \{f:\Sigma^{k}\rightarrow \{0,1\}\}$ and any $\varepsilon>0$, any single-pass streaming algorithm achieving approximation ratio $\alpha_{\mathrm{LP}}(\calF)+\varepsilon$ (the threshold $\alpha_{\mathrm{LP}}(\calF)$ is defined in \eqref{eq:def_alpha_LP}) for $\cspF$ must use at least $\Omega_{\varepsilon}(n)$ bits of memory.
\end{corollary}

Interestingly, the threshold $\alpha_{\mathrm{LP}}(\calF)$ in Corollary~\ref{cor:alg_and_lb_bd_deg} is also optimal for bounded-degree instances. Indeed, as observed by \cite{STV}, combining the algorithm of \cite[Section 5]{saxena2025streaming} with \cite{yoshida2011optimal} yields the following result:

\begin{theorem}[\cite{yoshida2011optimal,saxena2025streaming}]\label{thm:SSSV25}
For any $\varepsilon>0$ and any degree bound $\Delta\geq 1$, there is a single-pass $n^{1-\Omega_{\varepsilon,\Delta}(1)}$-space streaming algorithm that achieves $(\alpha_{\mathrm{LP}}(\calF)-\varepsilon)$-approximation on $\cspF$ instances with maximum degree at most $\Delta$.
\end{theorem}

\subsubsection{Discussion: Approximation Resistance}

A specific question that has guided prior research in streaming approximability of CSPs is \emph{approximation resistance}. To explain this notion, we define for any predicate family $\calF\subseteq \{f:\Sigma^{k}\rightarrow \{0,1\}\}$ the threshold 
\begin{equation}\label{eq:def_rho_F}
\rho(\calF) = \inf_{\calI\in\cspF} \val_{\cal I}.
\end{equation}
Clearly, the algorithm that simply outputs $\rho(\calF)$ (regardless of the input) is an $\rho(\calF)$-approximation of the values of all $\cspF$ instances. We say $\cspF$ is approximation-resistant against a certain class of algorithms if for any $\varepsilon>0$, no algorithm in that class can solve the gap problem $\McspF{1-\varepsilon}{\rho(\calF)+\varepsilon}$. 

The following approximation-resistance result was proved by \cite{CGSV24}.
\begin{theorem}[{\cite[Theorem 1.2]{CGSV24}}]\label{thm:CGSV_approximation_resistance}
Suppose $\calF\subseteq \{f:\Sigma^{k}\rightarrow\{0,1\}\}$ is a predicate family such that every predicate $f\in \calF$ supports one-wise independence (as formally defined in Section~\ref{subsec:general_notations}). Then for any $\varepsilon>0$, any single-pass streaming algorithm solving $\McspF{1}{\rho(\calF)+\varepsilon}$ must use $\Omega_{\varepsilon}(\sqrt{n})$ bits of memory.
\end{theorem}

Before this work, it was not clear whether the lower bound in Theorem~\ref{thm:CGSV_approximation_resistance} can be improved. For the specific problem of max-cut,\footnote{Max-cut corresponds to the predicate family $\calF\subseteq \{f:\{0,1\}^{2}\rightarrow \{0,1\}\}$ that consists of the single predicate $f$ defined by $f(x_{1},x_{2})=1$ if and only if $x_{1}\neq x_{2}$ (which clearly supports one-wise independence). For this family $\calF$, we have $\rho(\calF)=\alpha_{\mathrm{LP}}(\calF)=1/2$.} \cite{KK19} proved that any single-pass streaming algorithm solving $\mathsf{MaxCut}[1,1/2+\varepsilon]$ must use $\Omega_{\varepsilon}(n)$ bits of memeory. Later, \cite{chou2022linear} generalized \cite{KK19}'s result to predicate families that satisfy a certain condition (stronger than supporting one-wise independence). However, they were not able to generalize their lower bound to all predicate families that support one-wise independence, due to technical complications (as explained in \cite[Section 1.4]{chou2022linear}).

One of the technical components (see the discussion at the end of Section~\ref{subsubsec:improve_discrepancy}) in the proof of Theorem~\ref{thm:main} is to address these complications and generalize \cite{KK19}'s result to all predicate families supporting one-wise independence. Indeed, Theorem~\ref{thm:main} implies the following strengthening of Theorem~\ref{thm:CGSV_approximation_resistance}:

\begin{theorem}\label{thm:one_wise_approximation_resistance}
Suppose $\calF\subseteq \{f:\Sigma^{k}\rightarrow\{0,1\}\}$ is a predicate family such that every predicate $f\in \calF$ supports one-wise independence (as formally defined in Section~\ref{subsec:general_notations}). Then for any $\varepsilon>0$, any single-pass streaming algorithm solving $\McspF{1}{\rho(\calF)+\varepsilon}$ must use $\Omega_{\varepsilon}(n)$ bits of memory.
\end{theorem}

\begin{remark}
Theorem~\ref{thm:main} actually implies that this approximation resistance also holds against, say, $o(\log n)$-pass $n^{0.99}$-space, or $n^{o(1)}$-pass $n^{0.49}$-space streaming algorithms.
\end{remark}
\subsubsection{Discussion: Query-to-Communication Lifting}\label{subsubsec:sublinear_space_vs_time}

The dichotomy results of \cite{yoshida2011optimal,FMW25b} revealed a phenomenological connection between the multi-pass streaming model and the ``bounded-degree query model'' from the property testing literature. We do not attempt to fully elaborate on this connection here.\footnote{We refer the interested reader to \cite[Section 1.2.2]{FMW25b} for a more detailed discussion.} Instead, we provide an informal account of how it relates to the present work.

A standard approach to proving multi-pass streaming lower bounds for $\McspF{c}{s}$ is to consider a corresponding communication problem. In this formulation, a constant number of players each hold a sequence of constraints over a common set of variables, and their goal is to distinguish whether the value of the combined CSP instance is at least $c$ or at most $s$. If this communication problem requires $f(n)$ bits of communication, then a standard reduction implies that any $p$-pass streaming algorithm for $\McspF{c}{s}$ must use $\Omega(f(n)/p)$ bits of space.

In communication complexity, a successful paradigm (e.g., \cite{raz1999separation,goos2017query,Gooslifting}) is to first establish a \emph{query} lower bound for a suitably defined ``query version'' of the problem,\footnote{Typically, the query model is strictly weaker than the corresponding communication model.} and then ``lift'' this bound to the communication setting.

Interestingly, the bounded-degree query model studied in \cite{yoshida2011optimal} serves as an appropriate query analogue of the communication problem for $\McspF{c}{s}$. In this model, \cite{yoshida2011optimal} shows that for any $\varepsilon > 0$, achieving an $(\alpha_{\mathrm{LP}} + \varepsilon)$-approximation for $\cspF$ requires $\Omega_{\varepsilon}(\sqrt{n})$ queries. The work of \cite{FMW25b} lifts this query lower bound to a communication lower bound of $\Omega_{\varepsilon}(n^{1/3})$, thereby yielding an $\Omega(n^{1/3}/p)$ space lower bound for $p$-pass streaming algorithms. However, their lifting incurs a loss in the exponent, reducing it from $1/2$ to $1/3$. Our Theorem~\ref{thm:main} can be viewed as a refinement of their query-to-communication lifting that avoids this polynomial loss.

Conversely, algorithms in the query model can usually be simulated by multi-pass streaming algorithms. For example, the bipartiteness tester of \cite{goldreich1999sublinear,kaufman2004tight} implies the following result:\footnote{The algorithm of \cite{goldreich1999sublinear} works not only for bounded-degree graphs, but also for all regular graphs. To extend the algorithm to general graphs, one can use either the reduction in \cite[Section 4]{kaufman2004tight} or (a straightforward adaptation of) the one in \cite[Section 4]{FMW25b}.} 
\begin{theorem}[\cite{goldreich1999sublinear,kaufman2004tight}]\label{thm:bipartiteness_tester}
For any fixed constant $\varepsilon\in (0,1)$, there is a $\mathrm{polylog}(n)$-pass, $\sqrt{n}\cdot\mathrm{polylog}(n)$-space streaming algorithm for $\mathsf{MaxCut}[1,1-\varepsilon]$.
\end{theorem}

Theorem~\ref{thm:bipartiteness_tester} implies that the $\Omega(\sqrt{n} /p)$ lower bound in Theorem~\ref{thm:main}(2) cannot be improved by more than $\mathrm{polylog}(n)$ factors; for the other lower bound $n\cdot 2^{-O(p)}$ featured in Theorem~\ref{thm:main}(2), the exponent $O(p)$ cannot be improved to sub-polynomial in $p$.

\subsubsection{Beyond $\sqrt{n}$ Queries}

There are also some MaxCSP problems for which the query complexity in the bounded-degree query model is known to be $\Omega(n)$. For example, for the predicate family $\mathsf{E3Lin}$\footnote{The family $\mathsf{E3Lin}$ consists of two predicates $f_{0},f_{1}:\bF_{2}^{3}\rightarrow \{0,1\}$ defined by $f_{b}(x_{1},x_{2},x_{3})=1$ if and only if $x_{1}+x_{2}+x_{3}=b$, for each $b\in \bF_{2}$.}, \cite{BOT02} showed that $\mathsf{MaxE3Lin}[1,1/2+\varepsilon]$ requires $\Omega_{\varepsilon}(n)$ queries. By slightly extending our techniques for proving Theorem~\ref{thm:main}, we are also able to lift \cite{BOT02}'s query lower bound to a communication lower bound, yielding the following result:

\begin{theorem}\label{thm:multi_pass_lin}
Suppose $\calF\subseteq \{f:\Sigma^{k}\rightarrow\{0,1\}\}$ is a predicate family such that every predicate $f\in \calF$ supports two-wise independence (as formally defined in Section~\ref{subsec:general_notations}). Then for any $\varepsilon>0$, any $p$-pass streaming algorithm solving $\McspF{1}{\rho(\calF)+\varepsilon}$ must use $\Omega_{\varepsilon}(n/p)$ bits of memory.
\end{theorem}

A recent work \cite{fei2025unbounded} showed that for any predicate family that has \emph{unbounded width},\footnote{The definition of unbounded width is slightly complicated (see \cite{feder1998computational} or \cite[Appendix A]{fei2025unbounded}); a nice point of reference is that if deciding satisfiability of $\cspF$ instances is (known to be) $\textbf{NP}$-hard, then $\calF$ has unbounded width.} there exists $\varepsilon>0$ such that $\McspF{1}{1-\varepsilon}$ requires $\Omega_{\varepsilon}(n)$ queries in the bounded-degree query model. Using the same techniques, it is not hard to lift this lower bound also to the multi-pass streaming model (answering \cite[Question 1.11]{fei2025unbounded}):

\begin{theorem}\label{thm:unbounded_width}
For any predicate family $\calF\subseteq \{f:\Sigma^{k}\rightarrow\{0,1\}\}$ that has unbounded width, there exists a constant $\varepsilon>0$ such that any $p$-pass streaming algorithm for $\McspF{1}{1-\varepsilon}$ must use at least $\Omega(n/p)$ bits of memory.
\end{theorem}

\subsection{Technical Overview}
In this section we discuss our techniques, which heavily build on~\cite{KK19} and~\cite{FMW25b}. The heart of the proof of~\Cref{thm:main} is a communication lower bound for the distributional hidden partition problem (DIHP) from~\cite{FMW25b}, which we describe next.

 \subsubsection{The DIHP Problem}
 Given an instance $\calI$ of $\cspF$ and a solution $\{x_{\sfv,\sigma}\}_{\sfv\in\calV, \sigma\in\Sigma}, \{z_{i,b}\}_{i\in[m],b\in\Sigma^{k}}$ to $\lp_{\calI}$, we first consider an intermediate object called a distribution-labeled graph. It is convenient to turn the LP solution into distributions over the Abelian groups $\mathbb{Z}_{N}$ and $\mathbb{Z}_{N}^k$. To do that, one picks a large enough $N$ so that $x_{\sfv,\sigma} N$ are all integers, and for each $\sfv\in\calV$ of $\calI$ picks a partition $\{I_{\sfv,\sigma}\}_{\sigma\in \Sigma}$ of $\mathbb{Z}_N$ where $I_{\sfv,\sigma}$ has size $x_{\sfv,\sigma} N$. Now for each constraint $C_i$, the distribution $\mu_{C_i}$ is defined as sampling $b = (b_1,\ldots,b_k)\in\Sigma^{k}$ according to $z_{i,b}$, and then for each variable $\{\sfv_{i,j}\}_{j=1,\ldots,k}$ in $C_i$, replace $b_j$ with a random element from $I_{\sfv_{i,j},b_j}$. We denote the underlying $k$-uniform hypergraph by $G$, and the distribution over $\mathbb{Z}_N^k$ on each edge ${\sf e}$ by $\mu_{\sf e}$.

With this setting in mind, the DIHP problem is defined as follows. The number of players is $|E(G)|K$, where $K$ is a large constant.
Let $n$ be thought of tending to infinity, and consider the blow-up $H$ of $G$ resulting from replacing each vertex in it by $n$ new vertices. For each hyperedge ${\sf e}\in E(G)$ we have $K$ corresponding players, each one of them receiving a randomly chosen matching of size $\alpha n$ between the clouds in $H$ corresponding to the vertices of ${\sf e}$. The edges in $H$ are additionally labeled, in a way that depends on whether we generate a yes or a no instance:
\begin{enumerate}
    \item In $\mathcal{D}_{\yes}$, we first sample $x\in \mathbb{Z}_N^{\calV\times [n]}$ uniformly. Then, for each edge $e$ sent to a player, we consider ${\sf e}$ the edge of $G$ that generated it, sample $w\sim \mu_{\sf e}$, and label $e$ by $x_{|e}-w$.
    \item In $\mathcal{D}_{\no}$, the label of each edge $e$ of $H$ is chosen uniformly from $\mathbb{Z}_N^k$.
\end{enumerate}
An instance of DIHP is generated either from $\mathcal{D}_{\yes}$
or $\mathcal{D}_{\no}$, and the goal of the players is to distinguish between the two cases. As proved in~\cite{FMW25b}, a communication lower bound of $f(n)$ on this problem implies a memory lower bound of $\Omega_{\varepsilon}(f(n)/p$) for any $p$-pass algorithm solving 
$\McspF{c-\varepsilon}{s+\varepsilon}$ where $c = \lpval_{\calI}$ and $s=\val_{\calI}$, and the rest of the discussion is focused on the DIHP problem.

\subsubsection{Previous Approaches and Challenges}\label{subsubsec:previous_challenges}

Consider any protocol $\Pi$ for DIHP. A key idea in~\cite{FMWa,FMW25b} is that, while a naive application of the discrepancy method fails to give decent lower bounds for $\Pi$, it is successful assuming the protocol $\Pi$ is \emph{global}, a notion we introduce next. 

Note that if we choose a random $k$-uniform matching of size $\alpha n$ on $kn$ vertices, then the probability a particular edge will appear in it is roughly $\frac{\alpha}{n}$. More generally, the probability $i$ specific edges will appear in it is roughly $\left(\frac{\alpha}{n}\right)^i$.
Informally, a protocol $\Pi$ is called global if each one of its induced rectangles $R = A^{(1)}\times\ldots\times A^{(|E(G)|K)}$ has a similar behavior. Also, we call this type of rectangles \emph{global} rectangles. Namely, if we sample a labeled matching from any one of the $A^{(j)}$'s, each subset of $i$ specific edges will  appear there with probability $O((\alpha/n)^i)$. One might expect the following two hopes to be true: 
\begin{itemize}
    \item general communication protocols can be reduced to global protocols; 
    \item the discrepancy of global protocols can be bounded.
\end{itemize}
First, notice that there are a lot examples of non-global communication protocols. To see this, we consider the following protocol: after receiving its input, the first player writes one of the labeled edges in its matching on the blackboard. Then, for every induced rectangle $R= A^{(1)}\times \dots \times A^{(|E(G)|K)}$, there exists one certain edge that appears in every labeled matching in $A^{(1)}$. So, it seems impossible to directly reduce general protocols to global ones. To eliminate this issue, \cite{FMW25b} proves that every general $\Pi$ can actually be transformed into a protocol of a certain type, and every induced rectangle $R =A^{(1)}\times \dots\times A^{(|E(G)|K)} $ of the transformed protocol is in the following form: there exists tuples of labeled partial matchings $\bdzeta= \left(\bfz^{(1)},\dots,\bfz^{(|E(G)|K)}\right)$,  such that (1) for every $i\in [|E(G)|K]$, every labeled matching in $A^{(i)}$ subsumes $\bfz^{(i)}$;  and (2) the rest parts of $A^{(i)}$ look global. In this case, we simply refer to $(\bdzeta,R)$ as a structured rectangle.

The main technical result in~\cite{FMW25b} is that given a communication protocol $\Pi$ with $O(n^{1/3})$ bits of communication: 
\begin{itemize}
    \item one can transform $\Pi$ into a protocol such that most of its induced rectangle are (1) structured, and thus can be written as $(\bdzeta,R)$; (2) large, which means $\calD_{\no}\geq 2^{-O(n^{1/3})}$; (3) the total number of labeled edges in $\bdzeta$ is small (say, at most $O(n^{1/3})$); (4) edges in $\bdzeta$ do not form a cycle, i.e., there is no collection of $r$ edges that cover at most $r(k-1)$ vertices. 
    \item every rectangle with the above four conditions has $\calD_{\no}(R)\approx \calD_{\yes}(R)$. 
\end{itemize}
Indeed, the first part can be extended to protocols with $O(\sqrt{n})$ communication with corresponding weaker guarantees. The barrier to an $\Omega(\sqrt{n})$ lower bound lies in the second part, namely, proving a structured rectangle $(\bdzeta,R)$, which satisfies (1) $\calD_{\no}(R)\geq 2^{-O(\sqrt{n})}$; (2) the total number of labeled edges in $\bdzeta$ is at most $O(\sqrt{n})$; (3) edges in $\bdzeta$ do not form a cycle, has small discrepancy under the two measures. 

Furthermore, when one tries to prove a near-linear lower bound against protocols with a bounded number of rounds of communication, the first part also breaks. Briefly speaking, our previous argument is round-insensitive, and there actually exists a communication protocol using $\widetilde{O}(\sqrt{n})$ bits to find a cycle in the union of $|E|K$ labeled matchings. This means the previous argument cannot  show that most structured rectangles $(\bdzeta,R)$ do not contain any cycle when the communication cost exceeds $\sqrt{n}$.

With this description in mind, the argument in this paper has to overcome the following two challenges:
\begin{enumerate}
    \item Better discrepancy bound for structured rectangles: as mentioned before, the discrepancy  bound for structured rectangles in~\cite{FMW25b} is the bottleneck in the above argument, and improving it (up to $\sqrt{n}$) would immediately translate to better lower bounds for multi-pass algorithms. Having said that, the Fourier analytic arguments there seem tight and no specific component could be improved in an obvious way.
    \item Better characterization of structured rectangles with low discrepancy: the above argument is incapable of getting lower bounds better than $\sqrt{n}$, even for a small number of passes $p$. The issue is that if we only consider the number of exposed edges in the structured part, we are unable to tell the difference between a collection of $t$ disjoint edges being exposed, and a collection of connected $t$ edges (in which further addition of edges is more likely to create cycles)\footnote{Please see Section \ref{subsec:good_rectangles} for a more detailed explanation.}. Thus, it seems necessary to find another characterization of structured rectangles with small discrepancy instead of considering the number of edges exposed. 
\end{enumerate}

\subsubsection{The Weight of a Restriction}\label{subsubsec:KK19_weight}

To solve the second issue we use an idea from~\cite{KK19}. Denoting the graph of exposed edges by $H_{\bdzeta}$ and letting its connected components be $V_1,\ldots,V_t$, the weight of $H_{\bdzeta}$ is defined as 
$\|\bdzeta\| = \sum\limits_{i=1}^{t}|V_i|^2$. The work~\cite{KK19} analyzed the behavior of $\|\bdzeta\|$ under edge exposure when each player only speaks once, and showed that it is a good progress measure towards cyclicity for up to linearly many bits of communication. 

The parameter $\|\bdzeta\|$ plays an important role in our analysis as well: we show that it is a good progress measure toward cyclicity for general protocols, as a replacement of the number of edges exposed by $\bdzeta$. In particular, we give an analog of the decomposition lemma of~\cite{FMW25b} with respect to it (instead the progress measure $|\bdzeta| = \sum_{i}|\supp(\bfz^{(i)})|$ used previously).
For our purposes, we have to consider protocols consisting of several rounds of communication, and we show that $\|\bdzeta\|$ typically grows by at most a constant factor in each round of communication. Thus, in the end it is typically at most as large as $2^{O(p)}
|\Pi|$.\footnote{This is the source of the exponential dependence on $p$ in the $n\cdot 2^{-O(p)}$ space lower bound featured in Theorem~\ref{thm:main}.}

\subsubsection{Improving the Discrepancy Bound}\label{subsubsec:improve_discrepancy}

To address the first issue raised at the end of Section~\ref{subsubsec:previous_challenges}, we again draw inspiration from~\cite{KK19}. The central argument of~\cite{KK19} analyzes how the distribution of the hidden partition $x \in \ZmodN^{\calV \times [n]}$,\footnote{In \cite{KK19}, the hidden partition lies in $\bF_{2}^{n}$.} conditioned on the transcript of the communication so far, evolves over the course of the protocol. By controlling the Fourier-$\ell_{1}$ mass of the probability density function of $x$, \cite{KK19} improved the single-pass space lower bound for $\mathsf{MaxCut}[1,1/2+\varepsilon]$ from $\Omega(\sqrt{n})$ (obtained in \cite{KKS15} via Fourier-$\ell_{2}$ methods) to $\Omega(n)$.

Our situation appears analogous to the state of affairs prior to the breakthrough of~\cite{KK19}. In particular, we seek to strengthen the \emph{multi-pass} space lower bound for $\mathsf{MaxCut}[1,1/2+\varepsilon]$ from $\Omega(n^{1/3}/p)$, as established in \cite{FMWa} using Fourier-$\ell_{2}$ techniques, to the near-optimal $\Omega(\sqrt{n}/p)$. It is thus natural to suspect the Fourier-$\ell_{1}$-based techniques of \cite{KK19} could be used to overcome the first issue mentioned in Section~\ref{subsubsec:previous_challenges}, and yield a near-optimal multi-pass space lower bound.

However, it is not immediately clear how the techniques of \cite{KK19} can be extended to establish multi-pass streaming lower bounds. At a high level, two main obstacles arise:

\begin{enumerate}
\item To control the evolution of the Fourier-$\ell_{1}$ mass of the density function of $x$, \cite{KK19} applies an induction on the number of rounds in the protocol. At low Fourier levels, the $\ell_{1}$-mass bound deteriorates by a constant factor with each additional message. While this degradation is tolerable in the single-pass setting --- where there are only a constant number of players, each speaking at most once --- it becomes prohibitive in the multi-pass setting, as it leads to an exponential dependence on the number of rounds (and hence on the number of passes). Consequently, this approach cannot yield an $\Omega(\sqrt{n}/p)$-type lower bound.

\item The induction argument in \cite{KK19} also critically relies on the fact that each player's message is generated from a fresh random labeled matching that is independent of the prior communication transcript. This independence assumption breaks down in the multi-pass setting, where players may speak multiple times.
\end{enumerate}
The strength of the structure-vs.-randomness framework of \cite{FMWa,FMW25b} lies precisely in its ability to handle such obstacles. From the perspective of this framework, however, the argument of \cite{KK19} has an intriguing feature: it does not rely on any explicit decomposition of the players' messages into ``structured'' and ``pseudorandom'' components. As noted in Section~\ref{subsubsec:KK19_weight}, \cite{KK19} does carry out a combinatorial analysis of the (highly structured) edge-exposure protocol using the notion of weight $\|\bdzeta\|$. Importantly, this analysis is \emph{not} invoked in their treatment of general single-pass protocols. Instead, the Fourier-$\ell_{1}$ method developed in \cite{KK19} can be viewed as a \emph{generalization} of this combinatorial analysis: it provides a unified approach that simultaneously captures both structured information and pseudorandom noise, which may coexist in general protocols.

The key idea enabling us to incorporate the techniques of \cite{KK19} into the framework of \cite{FMW25b} is to apply the Fourier-$\ell_{1}$-based induction argument of \cite{KK19} only to the pseudorandom component of each player's message. We begin by using the decomposition technique from \cite[Section~7]{FMW25b} to cleanly separate the structured information in the messages from the pseudorandom noise. The structured components of all players' messages are then aggregated and analyzed via the combinatorial method described in Section~\ref{subsubsec:KK19_weight}, while the pseudorandom components are handled using the Fourier-$\ell_{1}$-based induction argument.

This approach overcomes the two high-level obstacles discussed above. Nevertheless, several technical issues must still be addressed:

\begin{enumerate}
    \item In the induction argument of \cite{KK19}, each player's message is derived from a labeled matching that is not only independent of prior communication, but also \emph{uniformly random}. While the independence lacked in the multi-pass setting can be restored by the idea described above, the resulting labeled matchings are no longer perfectly uniform, but only ``pseudorandom.'' Fortunately, such pseudorandom matchings can be shown to be sufficiently similar to uniform for the argument to carry through.

    \item While the above discussion suffices to obtain the desired $\Omega(\sqrt{n}/p)$ space lower bound for $\mathsf{MaxCut}[1,1/2+\varepsilon]$, extending the result to general CSPs introduces additional challenges that are largely orthogonal to the issues discussed so far. In particular, we must first establish the \emph{single-pass} result (Theorem~\ref{thm:one_wise_approximation_resistance}), which was previously unknown.
\end{enumerate}

\subsection{Open Problems}\label{sec:open}
We finish this introductory section with a few open directions for future research.
\begin{enumerate}
    \item While the multi-pass lower bound provided in~\Cref{thm:main} is tight in general, it can be improved for some problems as in~\Cref{thm:multi_pass_lin,thm:unbounded_width}. It would be interesting to know if there are a natural criteria dictating whether the space-pass tradeoff is $\widetilde{\Theta}(\sqrt{n})$ or $\widetilde{\Theta}(n)$, and even more interesting if there are intermediate behaviors. Perhaps the ``query version'' of this question has to be answered first; see \cite[Question 1.13]{fei2025unbounded}.
    
    \item The algorithm (Theorem~\ref{thm:SSSV25}) of \cite{saxena2025streaming} that complements Corollary~\ref{cor:alg_and_lb_bd_deg} only works for bounded degree instances, because we cannot afford the degree-reduction step in~\cite{FMW25b} as it requires more than a single pass. The recent work \cite{ABFS} shows that the bounded-degree assumption in Theorem~\ref{thm:SSSV25} can be removed for the specific problem of max-directed-cut. For general CSPs, it remains open whether there is a single-pass $n^{1-\Omega_{\varepsilon}(1)}$-space algorithm for $(\alpha_{\mathrm{LP}}(\calF)-\varepsilon)$-approximation of $\cspF$.
    
    \item As noted in the first item above, it is unknown whether the space-pass tradeoff grows from $\widetilde{\Theta}(\sqrt{n})$ to $\widetilde{\Theta}(n)$ continuously or discretely as a function of the approximation ratio $\alpha\in [\alpha_{\mathrm{LP}}(\calF),1]$. In the single-pass setting, the mystery lies \emph{below} the LP threshold $\alpha_{\mathrm{LP}}(\calF)$: for the specific problem of max-directed-cut (whose LP threshold equals $1/2$), \cite{chou2020optimal} shows that $(4/9-\varepsilon)$-approximation can be achieved by single-pass $O_{\varepsilon}(\log n)$-space algorithms, while $(4/9+\varepsilon)$-approximation requires $\Omega_{\varepsilon}(\sqrt{n})$-space. As shown by \cite{KK19}, achieving $(1/2+\varepsilon)$-approximation of max-directed-cut in the single-pass setting requires $\Omega_{\varepsilon}(n)$ space. It is unknown whether the space complexity grows continuously as a function of the approximation ratio $\alpha \in [4/9, 1/2]$. For general CSPs, it is also unknown whether there is a critical approximation ratio at which the single-pass space complexity jumps from $\mathrm{polylog}(n)$ to $n^{\Omega(1)}$. 
\end{enumerate}

%% file: prelim.tex
\section{Preliminaries}

\subsection{General Notations}\label{subsec:general_notations}

In this subsection we summarize general notational conventions used throughout the paper. Additional notation will be introduced as needed, typically within dedicated ``Notation'' environments.

\paragraph{Arithmetic.} We use the convention $0^{0}=1$. For a real number $x$, we denote $x^{+}=\max\{x,0\}$.

\paragraph{Probability.}  
For a finite set $\Lambda$, we write $\Exs{x \in \Lambda}{\cdot}$ and $\Prs{x \in \Lambda}{\cdot}$ to denote expectation and probability, respectively, when $x$ is drawn uniformly at random from $\Lambda$.  
If $x$ is sampled according to a specific distribution $\calD$ over $\Lambda$, we write $x \sim \calD$ in place of $x \in \Lambda$. A \emph{probability mass function} on $\Lambda$ is a function $p:\Lambda\rightarrow [0,\infty)$ such that $\sum_{x\in \Lambda}p(x)=1$, while a \emph{probability density function} is a function $f:\Lambda\rightarrow [0,\infty)$ such that $\Exs{x\in \Lambda}{f(x)}=1$. A \emph{right stochastic matrix}, or a \emph{Markov kernel}, is a matrix in which each row is a probability mass function on the set of columns.  

\paragraph{Hilbert space.}  
For a finite set $\Lambda$, we denote by $L^2(\Lambda)$ the (finite-dimensional) Hilbert space of complex-valued functions on $\Lambda$, equipped with the inner product
\[
\langle f, g \rangle := \Exu{x \in \Lambda}{f(x)\, \overline{g(x)}}.
\]

\paragraph{Fourier analysis.}  
We denote the finite cyclic group $\bZ / N\bZ$ by $\bZ_N$, where $N\geq 2$ is an integer. Throughout the paper, the capital letter $N$ is reserved exclusively for this notation. For any finite index set $\Lambda$, the collection of Fourier characters on the product group $\bZ_N^{\Lambda}$ is indexed by $\bZ_{N}^{\Lambda}$ itself. More precisely, for $b \in \bZ_N^{\Lambda}$, the associated character function $\chi_b : \bZ_N^{\Lambda} \to \bC$ is defined by
\[
\chi_b(x) := \exp\left( \frac{2\pi \mathrm{i}}{N} \sum_{v \in \Lambda} b_v x_v \right),
\]
where $\mathrm{i}$ denotes the imaginary unit. For a function $f\in L^{2}(\ZmodN^{\Lambda})$, we write $\widehat{f}(b):=\langle f,\chi_{b}\rangle$, and the Fourier $\ell^{1}$-norm of $f$ is denoted by
\[
\left\|f\right\|_{\sfW}:=\sum_{b\in \ZmodN^{\Lambda}}\left|\widehat{f}(b)\right|.
\]

\paragraph{Vectors and maps.}  
For a vector $x \in \bZ_N^{\Lambda}$ or $x \in [0,1]^{\Lambda}$, we denote its coordinates by subscripts: $x_v$ for each $v \in \Lambda$.  
A related notion is that of a \emph{map} $\bfy : \Lambda \to \Gamma$. We use boldface symbols for maps, especially when their images are themselves vectors, to distinguish them from ordinary vectors.  
For $v \in \Lambda$, the value of the map at $v$ is denoted by $\bfy(v)$.  
The collection of all such maps is denoted by $\Map{\Lambda}{\Gamma}$.

\paragraph{Support sets.}
Let $\Gamma$ be a domain containing a distinguished \emph{nullity element}. For either a vector $x \in \Gamma^{\Lambda}$ or a map $\bfy \in \Map{\Lambda}{\Gamma}$, the \emph{support} of $x$ or $\bfy$ --- denoted $\supp(x)$ or $\supp(\bfy)$ --- is the set of elements $v \in \Lambda$ such that $x_v$ or $\bfy(v)$ is not equal to the nullity element. For example, when $\Gamma = \bZ_N$, the nullity element is the additive identity $0$. In some cases, the domain $\Gamma$ is taken to be a disjoint union of an Abelian group and a special symbol --- such as $\bZ_N^k \cup \{\nil\}$ --- in which case the nullity element is the special symbol $\nil$, rather than the identity of the group. The \emph{Hamming weight} of a vector $x \in \Gamma^{\Lambda}$ or a map $\bfy \in \Map{\Lambda}{\Gamma}$, denoted by $\|x\|_{\sfH}$ or $\|\bfy\|_{\sfH}$, respectively, is the cardinality of its support set.

\paragraph{Degree decomposition.} Let $\Lambda$ and $\Gamma$ be finite sets. A function $f:\Gamma^{\Lambda}\rightarrow\bC$ is said to be a $d$-junta, where $d$ is a nonnegative integer no larger than $|\Lambda|$, if $f(x)$ depends on at most $d$ coordinates of $x$. A function $f$ is said to have \emph{degree} at most $d$ if $f$ is a linear combination of $d$-juntas. For any function $f\in L^{2}(\Gamma^{\Lambda})$, we let $f^{\leq d}$ be the orthogonal projection of $f$ to the linear subspace of $L^{2}(\Gamma^{\Lambda})$ consisting of functions of degree at most $d$. In the case $\Gamma=\ZmodN$, there is a convenient explicit formula for $f^{\leq d}$: if $f\in L^{2}(\ZmodN^{\Lambda})$ we will write the degree-$d$ part of $f$ as
\[
f^{=d}:=\sum_{b\in \ZmodN^{\Lambda},\, |\supp(b)|=d}\left\langle f,\chi_{b}\right\rangle\cdot \chi_{b},
\]
and we then have $f^{\leq d}=\sum_{i=0}^{d}f^{=i}$.

\paragraph{CSPs and hypergraphs.}  
Throughout the paper, $\Sigma$ denotes the CSP alphabet, and the lowercase letter $k$ always refers to the arity of predicates. The calligraphic letter $\calF$ always denotes a nonempty finite set of predicates mapping from $\Sigma^{k}$ to $\{0,1\}$. Correspondingly, all hypergraphs in this paper are assumed to be $k$-uniform. When the context is clear, hyperedges are sometimes simply referred to as edges. A set of hyperedges in a \(k\)-uniform hypergraph is said to contain a cycle if there exist \(\ell\) hyperedges in the set that cover at most \(\ell(k-1)\) vertices, for some \(\ell \geq 1\). A connected component of a hypergraph is said to be nontrivial if it contains at least two vertices.

\paragraph{Blow-up of hypergraphs.} Sometimes a hypergraph undergoes a \emph{blow-up}, in which each original vertex is replaced by $n$ copies. We adopt the following notational convention: pre-blowup vertices and hyperedges are denoted using sans-serif font (e.g., $\mathsf{v}$ and $\mathsf{e}$), while post-blowup vertices and hyperedges are written in standard math font (e.g., $v$ and $e$).

\paragraph{One-wise/two-wise independence.} Let $\Gamma$ be a finite set. A probability distribution over $\Gamma^{k}$ is called \emph{one-wise independent} if its marginal on each of the $k$ coordinates is the uniform distribution on $\Gamma$. It is called \emph{two-wise independent} if its marginal on any two distinct coordinates is the uniform distribution over $\Gamma^{2}$. For any one-wise independent distribution $\mu$, we denote by $\mu(\cdot)$ the probability mass function of $\mu$. A predicate $f:\Sigma^{k}\rightarrow\{0,1\}$ is said to support one-wise (respectively, two-wise) independence if there exists a one-wise (respectively, two-wise) independent distribution supported on $f^{-1}(1)$.

\subsection{Hypercontractivity}
Hypercontractive inequalities on product spaces have been crucial tools in establishing streaming lower bounds for approximating CSPs. We need the following version in this paper:
\begin{proposition}
    [{\cite[Theorem 10.21]{o2014analysis}}]\label{prop:classical_hypercontractivity}
    Let $\Lambda$ and $\Gamma$ be finite sets. For any function $f:\Gamma^\Lambda \rightarrow \bR$ with degree at most $d$ and any real number $q\geq 2$, we have
    \begin{align*}
        \lVert f \rVert_q \leq \left(\sqrt{(q-1)|\Gamma|} \right)^d \lVert f \rVert_2
    \end{align*}
\end{proposition}
As is standard in many applications, the above hypercontractivity result is used to obtain the following level-$d$ inequality.
\begin{proposition}\label{prop:level_d_classical}
    Let $\Lambda$ and $\Gamma$ be finite sets. For any function $f: \Gamma^{\Lambda} \rightarrow \bR$ and any positive integer $d\leq 2\log_{2}(\|f\|_{2}/\|f\|_{1})$, we have 
    \begin{align*}
        \left\| f^{\leq d}\right\|_2^2\leq \|f\|_{1}^{2}\cdot\left( \frac{8|\Gamma|}{d}\log_{2}\left(\frac{\|f\|_{2}}{\|f\|_{1}}\right)\right)^d . 
    \end{align*}
\end{proposition}
\begin{proof}
    For any $q\geq 2$, we have
    \begin{align*}
        \left\|f^{\leq d} \right\|_2^2 = \left\langle f,f^{\leq d}\right\rangle &\leq \left\| f^{\leq d}\right\|_q \cdot\lVert f\rVert_{q/(q-1)}\leq  \left\| f^{\leq d}\right\|_q \cdot \lVert f\rVert_{1}^{(q-2)/q}\lVert f\rVert_{2}^{2/q}\\
        &\leq \left(\sqrt{(q-1)|\Gamma|}\right)^d\left\|  f^{\leq d}\right\|_2 \cdot \lVert f\rVert_{1}^{(q-2)/q}\lVert f\rVert_{2}^{2/q},
    \end{align*}
    where the second and third transitions are by H\"{o}lder's inequality, and the fourth transition is by \Cref{prop:classical_hypercontractivity}. Thus, we have 
    \begin{align}\label{ineq:level_d_classical}
         \left\| f^{\leq d} \right\|_2^2 \leq \|f\|_{1}^{2}\cdot \left((q-1)|\Gamma|\right)^d \left(\frac{\|f\|_{2}}{\|f\|_{1}}\right)^{4/q}. 
    \end{align}
    Taking $q = 4d^{-1}\log_{2} \left(\lVert f\rVert_{2}/\|f\|_{1}\right)$ yields the conclusion.
\end{proof}

\subsection{The Basic Linear Program}

We show that Theorem~\ref{thm:one_wise_approximation_resistance} follows from Theorem~\ref{thm:main} using the following fact about $\lp_{\calI}$.

\begin{proposition}\label{prop:one_wise_LP_val_1}
Suppose $\calF\subseteq \{f:\Sigma^{k}\rightarrow\{0,1\}\}$ is a predicate family such that every predicate $f\in \calF$ supports one-wise independence. Then for any instance $\calI=(\calV,\calC)\in \cspF$, we have $\lpval_{\calI}=1$. 
\end{proposition}

\begin{proof}
Notice that the solution where 
\begin{enumerate}[label=(\arabic*)]
\item $\{x_{\sfv,\sigma}\}_{\sigma\in \Sigma}$ represents the uniform distribution over $\Sigma$ for any variable $\sfv\in \calV$, and 
\item $\{z_{i,b}\}_{b\in \Sigma^{k}}$ represents a one-wise independent distribution supported on $f_{i}^{-1}(1)$ for any constraint $C_{i}=((\sfv_{i,1},\dots,\sfv_{i,k}),f_{i})$ in $\calC$.
\end{enumerate}
achieves objective value 1 in $\lp_{\calI}$.
\end{proof}

\begin{proof}[Proof of Theorem~\ref{thm:one_wise_approximation_resistance} assuming Theorem~\ref{thm:main}]
For any $\varepsilon>0$, we can pick an instance $\calI\in\cspF$ with $\val_{\calI}\leq \rho(\calF)+\varepsilon/2$ (by the definition~\eqref{eq:def_rho_F}). By Proposition~\ref{prop:one_wise_LP_val_1} we know that $\lpval_{\calI}=1$. We can then apply Theorem~\ref{thm:main}(2) to the instance $\calI$ to yield the hardness of $\McspF{1}{\rho(\calF)+\varepsilon}$.
\end{proof}

Similarly to the proof of Proposition~\ref{prop:one_wise_LP_val_1}, it is easy to derive the following proposition, which will be used in Section~\ref{subsec:implication_twowise} to prove Theorems~\ref{thm:multi_pass_lin} and~\ref{thm:unbounded_width}.

\begin{proposition}\label{prop:two_wise_LP_val_1}
Suppose $\calF\subseteq \{f:\Sigma^{k}\rightarrow \{0,1\}\}$ is a predicate family such that every predicate $f\in\calF$ supports \emph{two-wise} independence. Then for any instance $\calI=(\calV,(C_{1},\dots,C_{m}))\in\cspF$, there exists a solution $\left((x_{\sfv,\sigma})_{\sfv \in \calV,\, \sigma \in \Sigma}, (z_{i,b})_{i \in [m],\, b \in \Sigma^k}\right)$ to $\lp_{\calI}$ achieving objective value 1 such that $\{z_{i,b}\}_{b\in \Sigma^{k}}$ represents a two-wise independent distribution over $\Sigma^{k}$ for any $i\in [m]$.
\end{proposition}

%% file: communication_game.tex
\section{Streaming Lower Bound from Communication Complexity}

The purpose of this section is to introduce (as a black box) a key lemma in \cite[Section 5]{FMW25b} showing that communication lower bounds for a certain communication game imply space lower bounds for streaming approximation of CSPs. In Sections~\ref{subsec:labeled_matchings} to~\ref{subsec:communication_game}, we present the necessary definitions for formalizing the communication game, and we refer the reader to \cite{FMW25b} for more detailed discussion of their motivations. 

\subsection{Labeled Matchings}\label{subsec:labeled_matchings}

The first ingredient in the communication game is labeled matchings, which are combinatorial objects that have been widely used in establishing streaming lower bounds for approximating CSPs. While several different sets of notation have been used in the literature to represent labeled matchings, in this paper we adopt the formalism introduced in \cite[Section 5]{FMW25b}. We list the relevant definitions below.

\begin{definition}
For finite sets $U_{1},\dots,U_{k}$ of equal cardinality, we call the tuple $\calU=(U_{1},\dots,U_{k})$ a \emph{$k$-universe}. The \emph{cardinality} of $\calU$, denoted by $|\calU|$, is defined to be the common cardinality of the sets $U_{i}$. For convenience, we use the shorthand $\tcup \calU$ for the union $\bigcup_{i\in [k]}U_{i}$ and $\tprod \calU$ for the Cartesian product $\prod_{i=1}^{k}U_{i}$.
\end{definition}

\begin{definition}
For a $k$-universe $\calU$ and a nonnegative integer $m\leq |\calU|$, we let $\calM_{\calU,m}$ denote the collection of all matchings (without labels) in the complete $k$-partite hypergraph $(\tcup \calU,\tprod\calU)$ (the hypergraph with vertex set $\tcup \calU$ and edge set $\tprod\calU$) with $m$ edges. We also write $\calM_{\calU,\leq m}:=\bigcup_{d=0}^{m}\calM_{\calU,d}$.
\end{definition}

\begin{definition}
For a $k$-universe $\calU$ and a nonnegative integer $m\leq |\calU|$, we define the following space of labeled matchings:
$$\Omega^{\calU,m}:=\left\{\bfy\in\Map{\tprod \calU}{\ZNk\cup\{\nil\}}:\supp(\bfy)\text{ is a matching with }m\text{ edges}\right\}.$$
Here, $\supp(\bfy)$ denotes the support of $\bfy$, i.e., the edges in $\prod \calU$ mapped to $\bZ_N^k$ (see \Cref{subsec:general_notations}). 
\end{definition}

\subsection{The Markov Kernel}\label{subsec:Markov_kernel}

The notion of Markov kernels provides significant convenience in both the formalization and the analysis of the communication game. The following notation will be helpful in defining the relevant Markov transitions.

\begin{notation}\label{notation:vector_subscript}
Suppose $\Lambda$ is a ground set and \(x \in \mathbb{Z}_N^{\Lambda}\) is a $\ZmodN$-vector indexed by $\Lambda$. If \(e = (v_1, \dots, v_k)\) is a tuple of elements with each \(v_i \in \Lambda\) for \(i \in [k]\), we denote by \(x_{|e}\) the vector \((x_{v_1}, \dots, x_{v_k}) \in \mathbb{Z}_N^k\). 
\end{notation}


\begin{definition}\label{def:Markov_kernel}
Fix a $k$-universe $\calU$, a positive integer $m\leq |\calU|$, and a one-wise independent distribution $\mu$ over $\ZNk$. We define a right stochastic matrix $\bfP^{\calU,m}_{\mu}:\ZmodN^{\bigcup \calU}\times \Omega^{\calU,m}\rightarrow[0,+\infty)$ as follows.  For each $x\in \bZ_{N}^{\bigcup \calU}$ and $\bfy\in \Omega^{\calU,m}$, the entry $\bfP^{\calU,m}_{\mu}(x,\bfy)$ is the probability that the output of the following process equals $\bfy$:
\begin{enumerate}
\item sample a matching $M$ uniformly at random from $\calM_{\calU,m}$; 
\item let $\bfz\in \Omega^{\calU,m}$ have support $\supp(\bfz)=M$, and
\item for each edge $e\in M$, draw $w_{e}\in \ZNk$ independently from $\mu$ and set $\bfz(e)=x_{|e}-w_{e}$, where the subtraction is performed in the Abelian group $\ZNk$;
\item output $\bfz$.
\end{enumerate}
\end{definition}

\subsection{Distribution-Labeled $k$-Graphs}\label{subsec:distribution_labeled_graph}

The communication game $\DIHP(G,n,\alpha,K)$ (to be defined in \Cref{subsec:communication_game}) is based on an abstract structure $G$ called a distribution-labeled $k$-graph, which we now define as follows.

\begin{definition}\label{def:distribution_k_graph}
A \emph{distribution-labeled $k$-graph} $G$ consists of the following data: a vertex set $\calV$; a multi-set $\calE$ of hyperedges, each an ordered $k$-tuple of distinct vertices in $\calV$; a positive integer $N$; and a collection of probability distributions $(\mu_{\sfe})_{\sfe\in \calE}$, where each $\mu_{\sfe}$ is a one-wise independent probability distribution on the Abelian group $\ZNk$. 
\end{definition}

The second parameter $n$ in $\DIHP(G,n,\alpha,K)$ is the blow-up factor of the distribution-labeled $k$-graph $G$. The set-theoretic structure of the blow-up is captured by the following definition.

\begin{definition}\label{def:associated_combinatorial}
Given a distribution-labeled $k$-graph $G=(\calV,\calE,N,(\mu_{\sfe})_{\sfe\in \calE})$ and a positive integer $n$, we define the following associated combinatorial objects.
\begin{enumerate}
\item The set $\calV\times [n]$, i.e. the $n$-blow-up of the vertex set $\calV$, will be referred to as the ground set.
\item For each $\sfv\in \calV$, let $U_{\sfv}:=\{\sfv\}\times [n]$ be the subset of $\calV\times [n]$ consisting of the $n$ copies of $\sfv$.
\item We associate with each hyperedge $\sfe=(\sfv_{1},\dots,\sfv_{k})\in \calE$ the $k$-universe $\calU_{\sfe}:=(U_{\sfv_{1}},\dots,U_{\sfv_{k}})$.
\end{enumerate}
\end{definition}

Note that in the above definition, for each $\sfe\in \calE$, the $k$-universe $\calU_{\sfe}$ is \emph{embedded} in the set $\calV\times [n]$, in the sense that $\tcup\calU_{\sfe}$ is a subset of $\calV\times [n]$. In general, we record the following notational convention.

\begin{notation}\label{notation:embed_universe}
A $k$-universe $\calU$ is said to be \emph{embedded} in a finite set $\Lambda$ if $\tcup\calU$ is a subset of $\Lambda$.
\end{notation}
\subsection{The Communication Game}\label{subsec:communication_game}

The following notation will be helpful in defining the communication game, as well as in later parts of the paper.

\begin{definition}
Fix a distribution-labeled \(k\)-graph \(G = (\calV, \calE, N, (\mu_{\sfe})_{\sfe \in \calE})\). The Abelian group \(\mathbb{Z}_N^{\calV \times [n]}\) will play a central role throughout the paper. For each edge \(\sfe \in \calE\), recall from \Cref{def:associated_combinatorial} that \(\bigcup \calU_{\sfe} \subseteq \calV \times [n]\). We denote by \(\proj_{\sfe}\) the canonical projection from \(\mathbb{Z}_N^{\calV \times [n]}\) onto \(\mathbb{Z}_N^{\bigcup \calU_{\sfe}}\).
\end{definition}

We are now ready to define the communication game $\DIHP(G,n,\alpha,K)$.

\begin{definition}\label{def:communication_game}
Given a distribution-labeled $k$-graph $G=(\calV,\calE,N,(\mu_{\sfe})_{\sfe\in \calE})$, parameters $n,K\in \bN$ and $\alpha\in (0,1)$, we define the communication game $\DIHP(G,n,\alpha,K)$ as follows:

\begin{enumerate}
\item There are $|\calE|\cdot K$ players, each indexed by a pair $(\sfe,j)$, where $\sfe\in \calE$ and $j\in [K]$.

\item Each player $(\sfe,j)$ receives as input a labeled matching in $\Omega^{\calU_{\sfe},\alpha n}$. 

\item \textbf{The no distribution:} define $\calD_{\no}$ to be the uniform distribution on the Cartesian product $\prod_{(\sfe,j)\in \calE\times [K]}\Omega^{\calU_{\sfe},\alpha n}$, i.e. each player gets an independent uniformly random input.

\item \textbf{The yes distribution:} define $\calD_{\yes}$ to be the joint distribution of $(\bfy^{(\sfe,j)})_{(\sfe,j)\in \calE\times [K]}$ obtained by the following procedure:
\begin{itemize}
\item Sample a uniformly random vector $x\in \ZmodN^{\calV\times [n]}$. 
\item For each player $(\sfe,j)\in \calE\times [K]$, independently draw a labeled matching $\bfy^{(\sfe,j)}\in \Omega^{\calU_{\sfe},\alpha n}$ according to the distribution given by the probability mass function $\bfP^{\calU_{\sfe},\alpha n}_{\mu_{\sfe}}\big(\proj_{\sfe}(x),\cdot\big)$.
\end{itemize}
\end{enumerate}
The goal of the players is to decide whether their inputs $(\bfy^{(\sfe,j)})_{(\sfe,j)\in \calE\times [K]}$ comes from $\calD_{\yes}$ or $\calD_{\no}$.
\end{definition}

\begin{remark}
Throughout this paper, whenever we refer to the communication game $\DIHP(G,n,\alpha,K)$, we treat $G,\alpha$ and $K$ as fixed parameters, and consider the asymptotic regime $n\rightarrow\infty$. 
\end{remark}

As is standard in distributional communication complexity, we measure the performance of a communication protocol by its ``advantage'', defined as follows.

\begin{definition}\label{def:advantage}
A deterministic communication protocol $\Pi$ for $\DIHP(G,n,\alpha ,K)$ computes a function $\Pi:\prod_{(\sfe,j)\in \calE\times [K]}\Omega^{\calU_{\sfe},\alpha n}\rightarrow\{0,1\}$. We define its \emph{advantage} in the communication game as 
\begin{align*}
    \mathrm{adv}(\Pi) : = \left| \Pru{\bfY\sim \mathcal{D}_\yes}{\Pi(\bfY)  =1 } - \Pru{\bfY\sim \mathcal{D}_\no}{\Pi(\bfY) =1 }\right|,
\end{align*}
where $\bfY$ denotes a joint input $\bfY=(\bfy^{(\sfe,j)})_{(\sfe,j)\in \calE\times [K]}$.
\end{definition}

The communication complexity of $\DIHP$ is then defined as follows. 

\begin{definition}
The \emph{communication cost} of a protocol~$\Pi$, denoted by~$|\Pi|$, is the total number of bits broadcasted by all players across all rounds during its execution.  
The \emph{communication complexity} of the game $\DIHP(G, n, \alpha, K)$, denoted by~$\CC(G, n, \alpha, K)$, is the minimum communication cost over all protocols~$\Pi$ that satisfy $\adv(\Pi) \geq 0.1$.
\end{definition}

The first result of our paper is the following improved communication lower bound of the $\DIHP(G,n,\alpha,K)$ game. 
\begin{theorem}\label{thm:sqrt_lower_bound}
    Fix a distribution-labeled $k$-graph $G=(\calV,\calE,N,(\mu_{\sfe})_{\sfe\in \calE})$, an integer $K>0$ and a parameter $\alpha\in \big(0,2^{-20}N^{-10k}|\calV|^{-2}\big]$. There exists a constant $\gamma=\gamma(G,\alpha,K)>0$ such that $\CC(G,n,\alpha,K)\geq \gamma \sqrt{n}$.
\end{theorem}
\begin{remark}\label{rem:polylog_protocol}
Previously, \cite[Theorem 5.13]{FMW25b} proved the lower bound $\CC(G,n,\alpha,K)\geq \Omega(n^{1/3})$. The $\Omega(\sqrt{n})$ lower bound in theorem~\ref{thm:sqrt_lower_bound} cannot be improved by more than $\mathrm{polylog}(n)$ factors, due to Theorem~\ref{thm:bipartiteness_tester}. In fact, the result of Theorem \ref{thm:bipartiteness_tester} implies that for some distribution-labeled $2$-graph $G$, there is a $\mathrm{polylog}(n)$-round protocol for $\DIHP(G,n,\alpha,K)$ with cost $\widetilde{O}(\sqrt{n})$ as long as $\alpha K$ is a sufficiently large constant.
\end{remark}

Every communication protocol can be viewed as having a certain number of \emph{communication rounds}, and in each round only one player broadcasts a message of certain length. Motivated by studying the space lower bounds of streaming algorithms with bounded number of passes, it is also of interest to consider how efficient protocols can be if we impose a limit on the number of communication rounds.

\begin{definition}
For any positive integer $r$, we define $\CC(G,n,\alpha,K,r)$ to be the minimum communication cost over all protocols $\Pi$ that use at most $r$ communication rounds and satisfy $\adv(\Pi)\geq 0.1$.
\end{definition}
Recall from Remark~\ref{rem:polylog_protocol} that the lower bound in Theorem~\ref{thm:sqrt_lower_bound} cannot be improved by much as long as we allow $\mathrm{polylog}(n)$ rounds of communication. It is not clear, however, whether we can prove stronger lower bound for communication protocols with fewer rounds of communication. In this direction, we prove the following communication lower bound against protocols with bounded rounds of communication. 
\begin{theorem}\label{thm:linear_lower_bound}
     Fix a distribution-labeled $k$-graph $G=(\calV,\calE,N,(\mu_{\sfe})_{\sfe\in \calE})$, an integer $K>0$ and a parameter $\alpha\in \big(0,2^{-20}N^{-10k}|\calV|^{-2}\big]$. There exists a constant $\beta= \beta(G,\alpha,K)>0$ such that $\CC(G,n,\alpha,K,r)\geq \beta^{r}n$ for any positive integer $r$.
\end{theorem}
\begin{remark}
    Note that our two main results for the $\DIHP$ games are incomparable. Theorem \ref{thm:sqrt_lower_bound} gives an $\Omega(\sqrt{n})$ lower bound regardless of rounds. Theorem \ref{thm:linear_lower_bound} gives an almost linear lower bound when $r=o(\log n)$ but only yields a trivial bound when $r\geq  \log_{1/\beta} n$. 
\end{remark}
The proofs of the two theorems will occupy Sections~\ref{sec:communication_lower_bound} to~\ref{sec:induction}.

\subsection{Streaming Lower Bounds}

We are now ready to state the following lemma from \cite{FMW25b} that transfers communication lower bounds to the streaming setting.\footnote{In \cite[Lemma 5.14]{FMW25b}, the stated memory lower bound for streaming is $(pmK)^{-1}\cdot \CC(G,n,\alpha,K)$, without round limits in the communication complexity. However, the proof of the lemma given in \cite{FMW25b} actually produces the lower bound $(pmK)^{-1}\cdot \CC(G,n,\alpha,K,pmK)$, because it transforms any $p$-pass streaming algorithm into a $pmK$ round communication protocol.} 

\begin{lemma}[{\cite[Lemma 5.14]{FMW25b}}]\label{lem:communication_reduction}
Fix a nonempty $\cspF$ instance $\calI = (\calV, (C_1, \dots, C_m))$, and let $s := \val_{\calI}$ and $c := \lpval_{\calI}$.  
Then there exists a distribution-labeled $k$-graph~$G=(\calV,\calE,N,(\mu_{\sfe})_{\sfe\in \calE})$ such that for any fixed error parameter $\varepsilon \in (0,1)$ and constants
\begin{equation}\label{eq:reduction_alpha_K}
0<\alpha \leq (100k)^{-1}\varepsilon
\qquad \text{and} \qquad 
K \geq 100 \alpha^{-1} \varepsilon^{-2} N^{2k} \cdot |\calV| \log |\Sigma|,
\end{equation}
the following holds for sufficiently large $n$:
\begin{enumerate}[label=(\arabic*)]
    \item If $c < 1$, then any $p$-pass algorithm for $\McspF{c - \varepsilon}{s + \varepsilon}$ requires at least $(pmK)^{-1} \cdot \CC(G, n, \alpha, K, pmK)$ bits of memory on input instances with $|\calV| \cdot n$ variables and at most $mK\cdot n$ constraints.
    
    \item If $c = 1$, then any $p$-pass algorithm for $\McspF{1}{s + \varepsilon}$ requires at least $(pmK)^{-1} \cdot \CC(G, n, \alpha,K, pmK)$ bits of memory on input instances with $|\calV| \cdot n$ variables and at most $mK\cdot n$ constraints.
\end{enumerate}
\end{lemma}
It is now clear that Theorems~\ref{thm:sqrt_lower_bound} and~\ref{thm:linear_lower_bound} together imply Theorem~\ref{thm:main}.

\begin{proof}[Proof of Theorem~\ref{thm:main} assuming Theorems~\ref{thm:sqrt_lower_bound} and~\ref{thm:linear_lower_bound}]

We pick constants $\alpha$ and $K$ such that
\[0<\alpha <\min \left\{(100k)^{-1}\varepsilon,\, 2^{-20}N^{-10k} |\calV|^{-2}\right\}\qquad\text{and}\qquad K  \geq 100\alpha^{-1}\varepsilon^{-2}N^{2k}\cdot |\calV|\log|\Sigma|.\]
By Theorem \ref{thm:sqrt_lower_bound} we have $$\mathsf{CC}(G,n,\alpha,K,pmK) \geq \mathsf{CC}(G,n,\alpha, K)= \Omega_{\alpha,K}(\sqrt{n}).$$
Similarly, by Theorem \ref{thm:linear_lower_bound}, we have 
\begin{align*}
    \mathsf{CC}(G,n,\alpha,K,pmK)\geq \beta(G,\alpha ,K)^{pmK}\cdot n,
\end{align*}
where $\beta(G,\alpha,K)$ is some constant depends on $G,\alpha ,K$. It then follows from Lemma~\ref{lem:communication_reduction}(1) that any $p$-pass algorithm for $\mathsf{MaxCSP}(\calF)[c-\varepsilon,s+\varepsilon]$ requires $$\max\left\{\Omega_{\varepsilon}(\sqrt{n}/p),\,n\cdot\exp(-O_{\varepsilon}(p))\right\}$$ bits of memory. The second statement in Theorem~\ref{thm:main} similarly follows from Lemma~\ref{lem:communication_reduction}(2).
\end{proof}

%% file: communication_hardness.tex
\section{Communication Lower Bound for DIHP}\label{sec:communication_lower_bound}

This section is devoted to the proofs of Theorems~\ref{thm:sqrt_lower_bound} and~\ref{thm:linear_lower_bound}. As in \cite{FMWa,FMW25b}, the argument follows the standard structure-vs.-randomness framework in communication complexity, and consists of the two main steps:  

\begin{enumerate}
    \item Given a protocol $\Pi$ with low communication cost, the first step is to decompose each rectangle induced by $\Pi$ into smaller subrectangles. 
Following \cite{FMWa,FMW25b}, we show that after this decomposition, most rectangles are ``good'', in the sense that each encodes a well-structured piece of information together with a controlled form of pseudorandom noise.

A key difference from \cite{FMWa,FMW25b} is that in our setting the goodness of these rectangles must be established under two distinct sets of assumptions. 
For the proof of Theorem~\ref{thm:sqrt_lower_bound}, we assume $|\Pi| = o(\sqrt{n})$. 
In contrast, for the proof of Theorem~\ref{thm:linear_lower_bound}, we assume $|\Pi| = o(n)$ and additionally require that $\Pi$ uses few communication rounds.

The ``decomposition lemma'' needed under the first assumption was already established in \cite[Lemma 6.11]{FMW25b}. 
Under the second set of assumptions, however, a new decomposition lemma is required; this is provided by Lemma~\ref{lem:linear_decomposition}.

    \item The second step is to show for each ``good'' rectangle $R$ a discrepancy bound of the form\footnote{Note that our discrepancy bound is \emph{one-sided}, unlike the two-sided bound $|\calD_{\yes}(R)-\calD_{\no}(R)|\leq 10^{-3}\cdot \calD_{\no}(R)$ proved in \cite{FMWa,FMW25b}. We are content with one-sided bounds because they are sufficient for proving communication lower bounds.}
    \[
\calD_{\yes}(R)\geq (1-10^{-3})\cdot \calD_{\no}(R).
    \]
    This is provided by our ``discrepancy lemma,'' Lemma~\ref{lem:discrepancy_bound}. The ``goodness'' assumption in our discrepancy lemma is weaker than its counterparts in \cite{FMWa,FMW25b}, which is the key reason we are able to obtain improved communication lower bounds.
\end{enumerate}
These two steps are then combined to complete the proof of the communication lower bounds.

The new decomposition lemma (Lemma~\ref{lem:linear_decomposition}) will be proved in Section~\ref{sec:decomposition}; its proof combines ideas in \cite[Section 5.2]{KK19} and \cite[Appendix A]{FMW25b}. The proof of  the discrepancy lemma (Lemma~\ref{lem:discrepancy_bound}) will take up Sections~\ref{sec:discrepancy} and~\ref{sec:induction}, where we combine ideas from \cite[Section 7]{KK19} and \cite[Section 7]{FMW25b}. This
section is devoted to laying out the overarching framework that connects these
components. In particular, we formalize the notion of “good” rectangles in Sections~\ref{subsec:pseudorandomness_notion} and~\ref{subsec:good_rectangles}.
Then, in Section~\ref{subsec:main_lemmas}, we lay out the main lemmas, from which we derive the desired communication
lower bound in Section~\ref{subsec:communication_result}. 

\subsection{Pseudorandomness Notions}\label{subsec:pseudorandomness_notion}

A “good” rectangle is one in which structural information and pseudorandom part are
cleanly separated, with the information contained in the structural part being of bounded size. In this subsection, we formalize the notions of pseudorandomness for sets of
labeled matchings. This will allow us, in Section~\ref{subsec:good_rectangles}, to control the pseudorandomness in rectangles.

Throughout this subsection, we fix a $k$-universe $\calU=(U_{1},\dots,U_{k})$ and a positive integer $m\leq |\calU|$. We will consider pseudorandomness notions for the space of labeled matchings $\Omega^{\calU,m}$.
Our notion is based on the following type of restriction on the space $\Omega^{\calU,m}$.

\begin{definition}\label{def:restrictions}
We define the set of \emph{restrictions} to be $\Omega^{\calU,\leq m}:=\bigcup_{0\leq d\leq m}\Omega^{\calU,d}$, i.e., the subset of $\Map{\tprod \calU}{\ZNk\cup\{\nil\}}$ that consists of all labeled matchings with at most $m$ edges. For each such labeled matching $\bfz\in\Omega^{\calU,\leq m}$, we let $\Omega^{\calU,m}_{\bfz}\subseteq \Omega^{\calU,m}$ be the restricted domain defined by
\[
\Omega^{\calU,m}_{\bfz}:=\left\{\bfy\in \Omega^{\calU,m}:\bfy(e)=\bfz(e)\text{ for all }e\in\supp(\bfz)\right\}.
\]
\end{definition}
For notational convenience, we use the following notion. 
\begin{notation}\label{notation:universe_minus_matching}
Fix a $k$-universe $\calU$, a nonnegative integer $m\leq |\calU|$, and a restriction $\bfz$ on the space $\Omega^{\calU,m}$. We denote by $\calU_{\setminus \supp(\bfz)}$ the $k$-universe $(U_{1}',U_{2}',\dots,U_{k}')$ defined by setting for each $i\in [k]$
$$U_{i}':=U_{i}\setminus\left\{u:\text{some edge in $\supp(\bfz)$ has $u$ as its $i$-th vertex}\right\}.$$
We use $\calM_{\calU,m,\bfz}$ as a shorthand for $\calM_{\calU_{\setminus\supp(\bfz)},\,m-|\supp(\bfz)|}$.
\end{notation}

Before formalizing the main notion of pseudorandomness, we define the following convenient concept of subsumption of restrictions.

\begin{definition}
    For two restrictions $\bfz,\bfz'\in \Omega^{\calU,\leq m}$, we say $\bfz'$ \emph{subsumes} $\bfz$ if $\supp(\bfz)\subseteq \supp(\bfz')$ and for all $e\in \supp(\bfz)$ we have $\bfz(e)=\bfz'(e)$. 
\end{definition}

We are now ready to define pseudorandomness for sets of labeled matchings:
\begin{definition}\label{def:global_set}
A subset $A\subseteq \Omega^{\calU,m}$ is said to be \emph{$\bfz$-global} if $A\subseteq \Omega^{\calU,m}_{\bfz}$, and for all restrictions $\bfz'$ that subsume $\bfz$ we have 
\[
\frac{\left|A\cap \Omega^{\calU,m}_{\bfz'}\right|}{\left|\Omega^{\calU,m}_{\bfz'}\right|}\leq 2^{|\supp(\bfz')|-|\supp(\bfz)|}\cdot \frac{\left|A\cap \Omega^{\calU,m}_{\bfz}\right|}{\left|\Omega^{\calU,m}_{\bfz}\right|}.
\]
When $\bfz = \boldsymbol{0}$ is the trivial restriction, we simply say 
that $A$ is global (omitting the $\bfz$).
\end{definition}

In words, for a set $A$ and a restriction $\bfz$, we say that $A$ is 
$\bfz$-global if any further restrictions $\bfz'$ that subsumes $\bfz$ increases
the relative density of $A$ by factor at most $2^{|\supp(\bfz')|-|\supp(\bfz)|}$.

Given the notion of globalness for labeled matchings defined in Definition~\ref{def:global_set}, it is natural to consider the following analogous pseudorandomness notion for \emph{unlabeled} matchings:

\begin{definition}\label{def:pseudo_uniform}
A distribution $\calD$ over $\calM_{\calU,m}$ is said to be \emph{pseudo-uniform} if for any nonnegative integer $d\leq m$ and any partial matching $S\in \calM_{\calU,d}$, we have
\[
\Pru{M\sim \calD}{S\subseteq M}\leq 2^{d}\cdot\Pru{M\in \calM_{\calU,m}}{S\subseteq M}.
\]
\end{definition}

Definitions~\ref{def:global_set} and~\ref{def:pseudo_uniform} are related by the following observation.

\begin{proposition}\label{prop:global_vs_pseudouniformity}
If $\bfy$ is a uniformly random element of a global set $A\subseteq \Omega^{\calU,m}$, then the distribution of $\supp(\bfy)$ is pseudo-uniform.
\end{proposition}

\begin{proof}
This follows from direct calculation:
    \begin{align*}
        \Pru{\bfy\in A}{M\subseteq \supp(\bfy)} &= \sum_{\substack{\bfz \in\Omega^{\calU,|M|} \\ \supp(\bfz) = M}} \frac{\left|A\cap \Omega^{\calU,m}_{\bfz}\right|}{\left|A\right|}
        \leq 2^{|M|} \sum_{\substack{\bfz \in\Omega^{\calU,|M|} \\ \supp(\bfz) = M}}\frac{\left|\Omega^{\calU,m}_{\bfz}\right|}{\left|\Omega^{\calU,m}\right|}   \\
        &=2^{|M|}\cdot\Pru{\bfy\in \Omega^{\calU,m}}{M\subseteq \supp(\bfy)} \qedhere
    \end{align*}
\end{proof}
The following corollary will be useful in Section~\ref{subsec:component_growing}.
\begin{corollary}\label{cor:edge_appears_in_globalset}
    Given a restriction $\bfz$ on $\Omega^{\calU,m}$ such that $|\supp(\bfz)|\leq \min\{m,|\calU|/2\}$ and a $\bfz$-global set $A\subseteq \Omega^{\calU,m}_{\bfz}$. For any matching $S\in\calM_{\calU,\leq m}$ containing $\supp(\bfz)$, we have
    \begin{align*}
        \Pru{\bfy\in A}{S\subseteq \supp(\bfy)}\leq \left(\frac{6^{k}m}{|\calU|^{k}}\right)^{|S|-|\supp(\bfz)|}. 
    \end{align*}
\end{corollary}
\begin{proof}
    When $\bfy$ is a uniformly random element of $A$, the induced distribution of $\supp(\bfy)\setminus \supp(\bfz)$ is a pseudo-uniform distribution over $\mathcal{M}_{\calU, m, \bfz}$ since $A$ is $\bfz$-global. Then, by Definition \ref{def:pseudo_uniform}, we have 
    \begin{align*}
        \Pru{\bfy\in A}{S\subseteq \supp(\bfy)} & =  \Pru{\bfy\in A}{S\setminus \supp(\bfz)\subseteq \supp(\bfy)\setminus \supp(\bfz)} \\
        &\leq 2^{|S| - |\supp(\bfz)|}\Pru{M\in \calM_{\calU,m,\bfz}}{S\setminus \supp(\bfz)\subseteq M}. 
    \end{align*}
    So, it suffices to prove 
    \begin{align*}
        \Pru{M\in \calM_{\calU,m,\bfz}}{S\setminus \supp(\bfz)\subseteq M}\leq \left(\frac{6^{k-1}m}{|\calU|^{k}}\right)^{|S|-|\supp(\bfz)|}. 
    \end{align*}
    The left-hand side of the above display is
    \[
    \prod_{d=|\supp(z)|}^{|S|-1}
    \frac{m-d}{(|\calU|-d)^{k}}\leq \left(\frac{m}{|\calU|}\right)^{|S|-|\supp(\bfz)|}\prod_{d=|\supp(z)|}^{|S|-1}
    \frac{1}{(|\calU|-d)^{k-1}}\leq \left(\frac{6^{k-1}m}{|\calU|^{k}}\right)^{|S|-|\supp(\bfz)|}.\qedhere\]
\end{proof}

\subsection{``Good'' Rectangles}\label{subsec:good_rectangles}

Now, we turn to pseudorandomness notions for \textit{rectangles}. In this subsection, we fix a distribution-labeled $k$-graph $G=(\calV,\calE,N,(\mu_{\sfe})_{\sfe\in \calE]})$ and a communication game $\DIHP(G,n,\alpha,K)$.

Recall from \Cref{def:communication_game} that in the communication game \(\DIHP(G, n, \alpha, K)\), the joint input to the $|\calE| \cdot K$ players is an element $\bfY$ in the product space \(\prod_{(\sfe, j) \in \calE \times [K]} \Omega^{\calU_{\sfe}, \alpha n}\). As is standard in communication complexity, a subset of this product space that is a Cartesian product is referred to as a \emph{rectangle}, formally defined below.

\begin{definition}
A subset \(R \subseteq \prod_{(\sfe, j) \in \calE \times [K]} \Omega^{\calU_{\sfe}, \alpha n}\) is called a \emph{rectangle} if it is a Cartesian product of sets \(A^{(\sfe, j)} \subseteq \Omega^{\calU_{\sfe}, \alpha n}\), one for each \((\sfe, j) \in \calE \times [K]\); that is,
\[
R = \prod_{(\sfe, j) \in \calE \times [K]} A^{(\sfe, j)}.
\]
\end{definition}
Then, it is natural to extend our definitions of global sets to rectangles, which requires each component $A^{(\sfe,j)}$ to be a global set. 
\begin{definition}\label{def:structured_rectangle}
Let \(\bdzeta = \left(\bfz^{(\sfe, j)}\right)_{(\sfe, j) \in \calE \times [K]}\) be a sequence where each \(\bfz^{(\sfe, j)}\) is a restriction on the space \(\Omega^{\calU_{\sfe}, \alpha n}\). A rectangle \(R = \prod_{(\sfe, j) \in \calE \times [K]} A^{(\sfe, j)}\) is called \emph{\(\bdzeta\)-global} if each set \(A^{(\sfe, j)}\) is \(\bfz^{(\sfe, j)}\)-global. When a rectangle $R$ is $\bdzeta$-global, we also say that the pair $(\bdzeta,R)$ is a \emph{structured rectangle}.
\end{definition}

In order for a structured rectangle $(\bdzeta,R)$ to be ``good,'' we need the restriction sequence $\bdzeta$ to satisfy additional properties. We would like to consider the following notion of \emph{weight} of a restriction sequence, inspired by \cite[Definition 5.1]{KK19}:

\begin{definition}\label{def:weight_of_restriction}
    Given a sequence of restriction \(\bdzeta = \left(\bfz^{(\sfe, j)}\right)_{(\sfe, j) \in \calE \times [K]}\), consider the hypergraph $H_{\bdzeta}$ with vertex set $\calV\times [n]$ and edge set
    \[
    \bigcup_{(\sfe,j)\in \calE\times [K]} \supp(\bfz^{(\sfe,j)}).
    \]
    Let $V_{1},\dots,V_{t}\subseteq \calV\times [n]$ be the list of nontrivial connected components of $H_{\bdzeta}$. We define the \emph{weight} of the restriction sequence $\bdzeta$ to be
    \[
    \|\bdzeta\|:=\sum_{i=1}^{t}|V_{i}|^{2}.
    \]
\end{definition}
We will also need the definition of ``cyclic'' restriction sequences stated as follows. 
\begin{definition}\label{def:cyclic_res_seq}
A restriction sequence \(\bdzeta = \left(\bfz^{(\sfe, j)}\right)_{(\sfe, j) \in \calE \times [K]}\) is said to be \emph{cyclic} if either the edge sets \(\left(\supp(\bfz^{(\mathsf{e},j)})\right)_{(\sfe,j)\in \calE\times [K]}\) are not pairwise disjoint, or the hypergraph $H_{\bdzeta}$ contains a cycle (for the definition of cycle-freeness in hypergraphs, see Section~\ref{subsec:general_notations}).
\end{definition}

The notion of weight defined above is crucial in the proof of the linear lower bound against protocols with bounded rounds of communication (Theorem \ref{thm:linear_lower_bound}). Intuitively, the quantity $\|\bdzeta\|$ controls the ``tendency'' of $\bdzeta$ to become cyclic, as will be formalized in Lemma~\ref{lem:bounded_growth}. To get some intuition, let $\bdzeta=(\bfz^{(\sfe,j)})_{(\sfe,j)\in \calE\times [K]}$ be an acyclic restriction, and consider the probability that $H_{\bdzeta}$ forms a cycle when we add a uniformly random $k$-edge $e$ to the graph $H_{\bdzeta}$. It is not hard to see that $H_{\bdzeta}$ forms a cycle with $e$ if and only if $e$ intersects some connected component $V$ of $H_{\bdzeta}$ on at least two vertices, which happens with probability $\approx \|\bdzeta\|/n^2$.\footnote{In our application we care about the case of adding a (pseudo-)uniformly random matching of size $\Theta(n)$ instead of adding a single edge. This is addressed formally in Lemma~\ref{lem:bounded_growth} by  careful calculations.} The work \cite{FMW25b} use the following cruder potential function to measure the tendency of $H_{\bdzeta}$ to become cyclic:
\begin{align*}
    |\bdzeta|:= \sum_{(\sfe,j)\in \calE\times[K]} \left|\supp(\bfz^{(\sfe,j)})\right|.
\end{align*}
The probability that $H_{\bdzeta}$ forms a cycle with a random edge is approximately at most $|\bdzeta|^2/n^2$, as $\|\bdzeta\|\lesssim |\bdzeta|^{2}$. To compare the two potential functions, consider the following two cases:
\begin{itemize}
    \item $H_{\bdzeta}$ is a graph of $t$ connected edges, forming a connected component of size $\approx kt$; in this case, $|\bdzeta|^2 \approx \|\bdzeta\|$, so the two bounds $|\bdzeta|^2/n^2$ and $\|\bdzeta\|/n^2$ are roughly the same.
    \item $H_{\bdzeta}$ is a graph of $t$ disjoint edges, forming $t$ connected components each of size $k$; in this case, $|\bdzeta| \approx \|\bdzeta\|$, so the bound $\|\bdzeta\|/n^2$ is much tighter than $|\bdzeta|^{2}/n^{2}$.
\end{itemize}
Provided that one expects $H_{\bdzeta}$ to have relatively small connected components, the above two examples suggest that $\|\bdzeta\|$ may be a more useful measure, and this is indeed the case.

We are now ready to give the formal definition of ``good'' rectangles. The rationale behind the three technical requirements in the following definition will become clear in \Cref{sec:discrepancy}.

\begin{definition}\label{def:good_rec_2}
    Let $W_{1},W_{2}$ be positive real numbers. We say a structured rectangle $(\boldsymbol{\zeta},R)$, where \(R=\prod_{(\sfe, j) \in \calE \times [K]} A^{(\sfe, j)}\) and \(\bdzeta = \left(\bfz^{(\sfe, j)}\right)_{(\sfe, j) \in \calE \times [K]}\), is $(W_{1},W_{2})$-good if the following conditions hold:
    \begin{enumerate}[label=(\arabic*)]
        \item The restriction sequence $\bdzeta$ is not cyclic.
        \item  $\left|A^{(\sfe,j)}\right| / \left|\Omega^{\calU_{\sfe},\alpha n}_{\bfz^{(\sfe,j)}}\right|\geq 2^{-W_{1}}$ for all $(\sfe,j)\in \calE\times [K]$.  
        \item The weight of $\bdzeta$ is at most $W_{2}$, i.e. $\lVert \bdzeta \rVert \leq W_{2}$. 
    \end{enumerate}
\end{definition}

\subsection{Three Main Lemmas}\label{subsec:main_lemmas}

We now present the two decomposition lemmas and the discrepancy lemma promised in the introductory text of Section~\ref{sec:communication_lower_bound}. The following decomposition lemma needed for proving Theorem~\ref{thm:sqrt_lower_bound} is already proved in \cite[Appendix A]{FMW25b}.

\begin{lemma}[First decomposition lemma,{~\cite[Lemma 6.11]{FMW25b}}]\label{lem:regularity_decomposition}
    Fix a distribution-labeled $k$-graph $G=(\calV,\calE,N,(\mu_{\sfe})_{\sfe\in \calE})$, an integer $K>0$ and a parameter $\alpha \in (0,1)$. There exists a constant $\eta>0$ such that given any communication protocol $\Pi$ for $\DIHP(G,n,\alpha,K)$ with $|\Pi|\leq \eta \sqrt{n}$, there exists a collection $\calR$ of pairwise-disjoint structured rectangles $(\boldsymbol{\zeta},R)$ in the space \(\prod_{(\sfe, j) \in \calE \times [K]} \Omega^{\calU_{\sfe}, \alpha n}\) such that the following conditions hold: 
    \begin{enumerate}[label=(\arabic*)]
        \item \(\calD_{\no}\left(\bigcup_{(\boldsymbol{\zeta},R)\in \calR}R\right)\geq 0.99\).
        \item Each $(\boldsymbol{\zeta},R)\in\calR$ is $\left(10^5|\Pi|,10^{10}k^{2}|\Pi|^{2}\right)$-good.\footnote{Note that the goodness condition here contains the requirement $\|\bdzeta\|\leq 10^{10}K^{2}|\Pi|^{2}$. In the original statement of \cite[Lemma 6.11]{FMW25b}, this particular requirement is replaced by the stronger assumption $\sum_{(\sfe,j)}|\supp(\bfz^{(\sfe,j)})|\leq 10^{5}|\Pi|$ (while all other components of the goodness condition are unchanged). Since the bound $\sum_{(\sfe,j)}|\supp(\bfz^{(\sfe,j)})|\leq 10^{5}|\Pi|$ clearly implies $\|\bdzeta\|\leq 10^{10}K^{2}|\Pi|^{2}$, the original result of \cite[Lemma~6.11]{FMW25b} directly implies our Lemma~\ref{lem:regularity_decomposition}.}
        \item For each $(\boldsymbol{\zeta},R)\in \calR $, there exists $a_R \in \{0,1\}$ such that $\Pi(\bfY) = a_R$ for every $\bfY \in R$. 
    \end{enumerate}
\end{lemma}

The next decomposition lemma, which is needed for Theorem~\ref{thm:linear_lower_bound}, will be proved in Section~\ref{sec:decomposition}.

\begin{restatable}[Second decomposition lemma]{lemma}{lineardecomposition}\label{lem:linear_decomposition}
    Fix a distribution-labeled $k$-graph $G=(\calV,\calE,N,(\mu_{\sfe})_{\sfe\in \calE})$, an integer $K>0$ and a parameter $\alpha \in (0,1)$. There exist constants $\eta_{1},\eta_{2}\in (0,1)$ (depending only on $G,\alpha$ and $K$) such that given any $r$-round communication protocol $\Pi$ for $\DIHP(G,n,\alpha,K)$ with $|\Pi|\leq \eta_{1}^{r} n$, there exists a collection $\calR$ of pairwise-disjoint structured rectangles $(\boldsymbol{\zeta},R)$ in the space \(\prod_{(\sfe, j) \in \calE \times [K]} \Omega^{\calU_{\sfe}, \alpha n}\) such that the following conditions hold: 
    \begin{enumerate}[label=(\arabic*)]
        \item \(\calD_{\no}\left(\bigcup_{(\boldsymbol{\zeta},R)\in \calR}R\right)\geq 0.99\).
        \item Each $(\boldsymbol{\zeta},R)\in\calR$ is $\left(\eta_{2}^{-r}|\Pi|,\eta_{2}^{-r}|\Pi|\right)$-good. 
        \item For each $(\boldsymbol{\zeta},R)\in \calR $, there exists $a_R \in \{0,1\}$ such that $\Pi(\bfY) = a_R$ for every $\bfY \in R$. 
    \end{enumerate}
\end{restatable}

The following discrepancy lemma is needed for proving both Theorem~\ref{thm:sqrt_lower_bound} and Theorem~\ref{thm:linear_lower_bound}.

\begin{restatable}[Discrepancy lemma]{lemma}{discrepancybound}\label{lem:discrepancy_bound}
    Fix a distribution-labeled $k$-graph $G=(\calV,\calE,N,(\mu_{\sfe})_{\sfe\in \calE})$, an integer $K>0$ and a parameter $\alpha\in \left(0,2^{-20}N^{-10k}|\calV|^{-2}\right]$. There exists a constant $\gamma\in (0,1)$ such that for any $(\gamma n, \gamma n)$-good structured rectangle $(\boldsymbol{\zeta},R)$, we have
    \begin{align*}
        \calD_{\yes}(R)\geq (1-10^{-3})\cdot \calD_{\no}(R).
    \end{align*}
\end{restatable}
The proof of Lemma~\ref{lem:discrepancy_bound} will take up Sections~\ref{sec:discrepancy} and~\ref{sec:induction}. 

\subsection{Communication Lower Bounds}\label{subsec:communication_result}

Assuming the lemmas in Section~\ref{subsec:main_lemmas}, we can now prove Theorems~\ref{thm:sqrt_lower_bound} and~\ref{thm:linear_lower_bound}.

\begin{proof}[Proof of Theorem \ref{thm:sqrt_lower_bound}]
    We apply Lemmas~\ref{lem:regularity_decomposition} and~\ref{lem:discrepancy_bound} to obtain constants $\eta$ and $\gamma$, respectively. We will show that any protocol $\Pi$ for $\DIHP(G,n,\alpha,K)$ with
    \begin{equation}\label{eq:assumption_Pi}
    |\Pi|\leq \min\left\{\eta,10^{-10}k^{-2}\gamma\right\}\cdot\sqrt{n}
    \end{equation}
    must have $\adv(\Pi)<0.1$.

    We apply Lemma~\ref{lem:regularity_decomposition} to $\Pi$, and let $\calR$ be the collection of structured rectangles obtained. By \eqref{eq:assumption_Pi} and conclusion (2) of Lemma~\ref{lem:regularity_decomposition}, we know that each structured rectangle in $\calR$ is $(\gamma n,\gamma n)$-good. It then follows from Lemma~\ref{lem:discrepancy_bound} that $\calD_{\yes}(R)\geq (1-10^{-3})\cdot \calD_{\no}(R)$ for each $(\bdzeta,R)\in \calR$. 

    By conclusion (3) of Lemma~\ref{lem:regularity_decomposition}, the function
    \[
    \Pi:\prod_{(\sfe,j)\in \calE\times [K]}\Omega^{\calU_{\sfe},\alpha n}\longrightarrow\{0,1\}
    \]
    is constant on $R$ for each $(\bdzeta,R)\in\calR$. Let $\calR_{1}$ be the collection of all structured rectangles in $\calR$ on which $\Pi$ evaluates to 1. Let $p_{R}:=\min\{\calD_{\yes}(R),\calD_{\no}(R)\}$ for each $(\bdzeta,R)\in \calR$, and define 
    \[
    p_{(1)}:=\sum_{(\bdzeta,R)\in \calR_{1}}p_{R},\qquad p_{(0)}:=\sum_{(\bdzeta,R)\in \calR\setminus \calR_{1}}p_{R},\qquad\text{and }p_{(*)}:=p_{(0)}+p_{(1)}.
    \]
    We thus have
    \[
    p_{(*)}=\sum_{(\bdzeta,R)\in\calR}p_{R}\geq \sum_{(\bdzeta,R)\in \calR}(1-10^{-3})\cdot\calD_{\no}(R)\geq (1-10^{-3})\cdot0.99>0.9,
    \]
    and hence
    \[
    p_{(1)}\leq \Pru{\bfY\sim\calD_{\yes}}{\Pi(\bfY)=1}\leq 1-p_{(0)}< p_{(*)}+0.1-p_{(0)}=p_{(1)}+0.1.
    \]
    Similarly, $\Prs{\bfY\sim\calD_{\no}}{\Pi(\bfY)=1}$ also lies in the interval $\big[p_{(1)},p_{(1)}+0.1\big)$, and therefore by Definition~\ref{def:advantage} we have $\adv(\Pi)<0.1$.
\end{proof}

\begin{proof}[Proof of Theorem~\ref{thm:linear_lower_bound}]
We apply Lemmas~\ref{lem:linear_decomposition} and~\ref{lem:discrepancy_bound} to obtain constants $\eta_{1},\eta_{2}$ and $\gamma$, respectively. We show that any $r$-round protocol $\Pi$ for $\DIHP(G,n,\alpha,K)$ with $|\Pi|\leq \min\{\eta_{1}^{r},\eta_{2}^{r}\gamma\}\cdot n$ must have $\adv(\Pi)<0.1$. 

We apply Lemma~\ref{lem:linear_decomposition} to $\Pi$, and let $\calR$ be the collection of structured rectangles obtained. By the upper bound on $|\Pi|$ and conclusion (2) of Lemma~\ref{lem:linear_decomposition}, we know that each structured rectangle in $\calR$ is $(\gamma n,\gamma n)$-good. The rest of the proof is identical to the proof of Theorem~\ref{thm:sqrt_lower_bound} above.
\end{proof}

%% file: decomposition.tex
\section{The Decomposition Lemma}\label{sec:decomposition}
In this section, we prove the second decomposition lemma, Lemma \ref{lem:linear_decomposition}. The proof is similar to that of the first decomposition lemma in \cite[Appendix A]{FMW25b}, consisting of two steps:
\begin{enumerate}
\item We first transform an arbitrary protocol into a \emph{global protocol} --- a protocol where roughly speaking, every message broadcasted by a player corresponds to structured rectangles. In the proof of Lemma~\ref{lem:linear_decomposition}, this first step is  essentially the same as its counterpart in \cite[Sections~A.1 and~A.2]{FMW25b}, which we will elaborate on in Section~\ref{subsec:transform_protocol}. 

\item The second step (carried out in Sections~\ref{subsec:component_growing} and~\ref{subsec:most_are_good}) is to show that in the division of the joint input space resulting from a global protocol, most structured rectangles are good (in the sense of Definition~\ref{def:good_rec_2}). This step deviates from \cite[Appendix A]{FMW25b} because here we need to keep track of the \emph{weight} of the restriction sequences of structured rectangles (as defined in Definition~\ref{def:weight_of_restriction}), due to the condition (3) in Definition~\ref{def:good_rec_2}. Our argument here is inspired by the analysis of the ``component growing protocol'' in \cite[Section 5.2]{KK19}. 
\end{enumerate} 

Throughout this section, we fix a distribution-labeled $k$-graph $G=(\calV,\calE,N,(\mu_{\sfe})_{\sfe\in\calE})$ along with an integer $K > 0$ and a parameter $\alpha > 0$. 

\subsection{Transforming Protocols into Global Protocols}\label{subsec:transform_protocol}

Before formalizing the notion of ``global protocol'' in Definition~\ref{def:global_protocol}, we introduce the following potential function to quantify the amount of information during the communication process.

\begin{definition}\label{def:restriction_potential}
    For restrictions $\boldsymbol{\zeta}=\left(\bfz^{(\sfe,j)}\right)_{(\sfe,j)\in \calE\times [K]}$ and a rectangle $R= \prod_{(\sfe, j) \in \calE\times [K]}A^{(\sfe, j)}$ such that $A^{(\sfe, j)}\subseteq \Omega^{\calU_\sfe,m}_{\bfz^{(\sfe,j)}}$,  we define the potential of $(\boldsymbol{\zeta},R)$ as: 
    \begin{align*}
    \phi(\boldsymbol{\zeta},R):= \sum_{(\sfe,j)\in \calE\times [K]} \left|\supp(\bfz^{(\sfe,j)})\right|+ \log_{2} \left(\frac{\left|\Omega^{\calU,m}_{\bfz^{(i)}}\right|}{|A^{(i)}|}
    \right). 
    \end{align*}
\end{definition}

\begin{definition}\label{def:global_protocol}
    A communication protocol $\Pi$ for $\textsf{DIHP}(G,n,\alpha,K)$ is called an $(r,c)$-round global communication protocol if it specifies the following procedure of communications: 
    \begin{itemize}
        \item the $K|\calE|$ players take turns to send messages according to $\Pi$; 
        \item there are at most $r$ rounds of communication, and there is only one player sending a message in a single round;
        \item the length of message in each round of communication is not bounded; instead, from the perspective of rectangles, after each round of communications, a $\boldsymbol{\zeta}$-global rectangle $R$ is further partitioned into a disjoint union of rectangles $R_{(1)},\dots,R_{(\ell)}$ such that: (1) $R_{(i)}$ is $\boldsymbol{\zeta}_{(i)}$-global; (2) $\boldsymbol{\zeta}_{(i)}$ subsumes $\boldsymbol{\zeta}$; (3) the following inequality holds: 
        \begin{equation}\label{eq:global_protocol_potential_increase}
            \sum_{i=1}^\ell\frac{|R_{(i)}|}{|R|}\phi(\boldsymbol{\zeta}_{(i)},R_{(i)})\leq \phi(\boldsymbol{\zeta},R) + 3 c. 
        \end{equation}
    \end{itemize}
\end{definition}
Note that in global protocols, the ``communication cost'' is measured by the increase of (average) potential instead of total number of bits broadcasted.


The following lemma is slightly different from \cite[Lemma A.4]{FMW25b}, but its proof is the same as the one provided in \cite[Section A.2]{FMW25b} and is thus omitted here.

\begin{lemma}\label{lem:arbitrary_to_global}
    Given an $r$-round communication protocol $\Pi$ for $\mathsf{DIHP}(G,n,\alpha,K)$, we can construct an $(r,|\Pi|)$-round global protocol $\Pi^{\reff}$ for $\mathsf{DIHP}(G,n,\alpha,K)$ such that for any leaf rectangle $R$ of $\Pi^{\reff}$, the output of $\Pi$ is constant on $R$. 
\end{lemma}

\subsection{A Global Protocol Grows Components Slowly}\label{subsec:component_growing}

Kapralov and Krachun~\cite{KK19} observed that, informally, a single-pass fully-structured\footnote{That is, a one-way global protocol in which every communicated message consists solely of structured information, with no pseudorandom noise.} protocol for $\DIHP(G,n,\alpha,K)$ can be viewed as follows: each player reveals a subset of edges from their labeled matching, with the collective goal of forming a cycle using the revealed edges.

To show that such fully structured protocols are unlikely to succeed, Kapralov and Krachun demonstrated in \cite[Section~5.2]{KK19} that the connected components in the (hyper)graph induced by the revealed edges cannot grow rapidly. In this subsection, we develop an analogous argument for \emph{general global protocols}, rather than single-pass fully-structured ones.

\begin{lemma}\label{lem:bounded_growth}
Suppose that in one round of communication of an $(r,c)$-round global protocol for $\DIHP(G,n,\alpha,K)$, a structured rectangle $(\bdzeta, R)$ is partitioned (according to the speaking player's message) into the disjoint union of structured rectangles
\[
(\bdzeta_{(1)},R_{(1)}),(\bdzeta_{(2)},R_{(2)})\dots,(\bdzeta_{(\ell)},R_{(\ell)}).
\]
We assume that $\|\bdzeta\|\leq 6^{-k-1}n$. Then we have
\begin{equation}\label{eq:bounded_growth}
\sum_{i=1}^{\ell}\frac{|R_{(i)}|}{|R|}\cdot\|\bdzeta_{(i)}\|\leq k^{2}\big(2\|\bdzeta\|+\phi(\bdzeta,R)+3c\big),
\end{equation}
and if (in addition) $\bdzeta$ is not cyclic (as per Definition~\ref{def:cyclic_res_seq}), we also have
\begin{equation}\label{eq:bounded_cyclic_prob}
\sum_{i=1}^{\ell}\frac{|R_{(i)}|}{|R|}\cdot\ind{\bdzeta_{(i)}\text{ is cyclic}}\leq \frac{6^{k+1}\alpha\|\bdzeta\|}{n}.
\end{equation}
\end{lemma}

\begin{proof} We divide the proof into the following three steps.

\paragraph{Step 1: understanding the weight increment.} We write
\begin{alignat}{2}
R&=\prod_{(\sfe,j)\in \calE\times [K]}A^{(\sfe,j)},\qquad&R_{(i)}&=\prod_{(\sfe,j)\in\calE\times [K]}A^{(\sfe,j)}_{(i)}\text{ for each }i\in [\ell],\label{eq:rec_partition_notation_1}\\
\bdzeta&=\left(\bfz^{(\sfe,j)}\right)_{(\sfe,j)\in \calE\times[K]},\qquad&\bdzeta_{(i)}&=\left(\bfz^{(\sfe,j)}_{(i)}\right)_{(\sfe,j)\in \calE\times[K]}\text{ for each }i\in [\ell].\label{eq:rec_partition_notation_2}
\end{alignat}
Suppose $(\sfe^*,j^*)\in \calE\times [K]$ is the player who speaks at the current communication round. Then we have $A^{(\sfe,j)}=A^{(\sfe,j)}_{(i)}$ and $\bfz^{(\sfe,j)}=\bfz^{(\sfe,j)}_{(i)}$ for all $(\sfe,j)\neq (\sfe^*,j^*)$ and all $i\in [K]$.

Recalling the notation in Definition~\ref{def:weight_of_restriction}, let $V_{1},\dots,V_{t}\subseteq \calV\times [n]$ denote the nontrivial connected components of the hypergraph $H_{\bdzeta}$. For each $i\in [\ell]$, the hypergraph $H_{\bdzeta_{(i)}}$ is obtained from $H_{\bdzeta}$ by adding the edges in the set
\begin{equation}\label{eq:def_of_E_i}
E_{(i)}:=\supp\left(\bfz^{(\sfe^*,j^*)}_{(i)}\right)\setminus \supp\left(\bfz^{(\sfe^*,j^*)}\right).
\end{equation}
Note that $E_{(i)}$ is a matching in the family $\calM_{\calU_{\sfe^*},\leq \alpha n}$. Therefore, if $V'\subseteq \calV\times [n]$ is a nontrivial connected component of $H_{\bdzeta_{(i)}}$, then either $V'$ spans exactly one edge in $E_{(i)}$, in which case $|V'|=k$, or $V'$ contains $V_{q}$ for some $q\in [t]$, in which case we have\footnote{To see this, note that there are at most $\sum_{q\in [t],\, V_{q}\subseteq V'}|V_{q}|$ edges of $E_{(i)}$ lying within $V'$, each accounting for at most $(k-1)$ vertices in $V'\setminus(\bigcup_{q\in [t],\, V_{q}\subseteq V'}V_{q})$.} 
\[
|V'|\leq k\cdot\sum_{q\in [t],\;V_{q}\subseteq V'}|V_{q}|.
\]

For any $q,s\in [t]$ and any $i\in [\ell]$, we define
\begin{equation}\label{eq:def_of_X_qsi}
X_{q,s,i}:=\begin{cases}
1,&\text{if }V_{q}\text{ and }V_{s}\text{ belong to the same connected component in }H_{\bdzeta_{(i)}},\\
0,&\text{otherwise}.
\end{cases}
\end{equation}
From the discussion in the last paragraph, it is easy to deduce that for each $i\in [\ell]$, we have
\begin{equation}\label{eq:bdzeta_i_expression}
\|\bdzeta_{(i)}\|\leq k^{2}\left|E_{(i)}\right|+k^{2}\sum_{q,s\in [t]}X_{q,s,i}|V_{q}||V_{s}|.
\end{equation}

\paragraph{Step 2: using the globalness condition.} For each positive integer $d$ and each pair of indices $q,s\in[t]$, let $S(q,s,d)$ be the collection of vertex sequences $(v_{1},v_{2},\dots,v_{2d})\in(\tcup \calU_{\sfe^*})^{2d}$ that satisfy the following conditions:
\begin{enumerate}[label=(\arabic*)]
\item We have $v_{1}\in V_{q}$ and $v_{2d}\in V_{s}$.
\item For any $j\in [d]$, the vertices $v_{2j-1}$ and $v_{2j}$ are distinct.
\item For any $j\in [d-1]$, there exists $i\in [t]$ such that $v_{2j},v_{2j+1}\in V_{i}$.
\end{enumerate}
Let $S(q,s)=\bigcup_{d=1}^{+\infty}S(q,s,d)$. Note that the cardinality of the sets $S(q,s,d)$ can be bounded by
\begin{equation}\label{eq:number_of_sequences}
|S(q,s,d)|\leq |V_{q}||V_{s}|\cdot\left(\sum_{i=1}^{t}|V_{i}|^{2}\right)^{d-1}=|V_{q}||V_{s}|\cdot\|\bdzeta\|^{d-1}.
\end{equation}

A matching $M\in \calM_{\calU_{\sfe^*},d}$ is called an \emph{exact cover} of a sequence $(v_{1},\dots,v_{2d})\in S(q,s,d)$ if the edges of $M$ can be listed as $e_{1},\dots,e_{d}$ such that $e_{j}$ contains $v_{2j-1}$ and $v_{2j}$ for all $j\in [d]$. The number of exact covers of a given sequence in $S(q,s,d)$ is clearly at most $n^{(k-2)d}$. By definitions \eqref{eq:def_of_E_i} and \eqref{eq:def_of_X_qsi}, for any distinct $q,s\in [t]$ and any $i\in [\ell]$, if $X_{q,s,i}=1$ then $E_{(i)}$ contains an exact cover of a sequence in $S(q,s)$. Furthermore, if $E_{(i)}$ contains an exact cover of a sequence in $S(q,s)$, then for all $\bfy\in A^{(\sfe^*,j^*)}_{(i)}$, the matching $\supp(\bfy)\setminus \supp\left(\bfz^{(\sfe^*,j^*)}\right)$ contains an exact cover of a sequence in $S(q,s)$. Therefore, for any distinct $q,s\in [t]$ we have
\begin{align*}
&\quad\sum_{i=1}^{\ell}\frac{|R_{(i)}|}{|R|}\cdot X_{q,s,i}=\sum_{i=1}^{\ell}\frac{\left|A^{(\sfe^*,j^*)}_{(i)}\right|}{\left|A^{(\sfe^*,j^*)}\right|}\cdot X_{q,s,i}\\
&\leq \Pru{\bfy\in A^{(\sfe^*,j^*)}}{\supp(\bfy)\setminus \supp\left(\bfz^{(\sfe^*,j^*)}\right)\text{ contains an exact cover of a sequence in }S(q,s)}\\
&\leq \sum_{d=1}^{+\infty}|S(q,s,d)|\cdot n^{(k-2)d}\cdot\left(\frac{6^{k}\alpha n}{n^{k}}\right)^{d}\tag{using Corollary~\ref{cor:edge_appears_in_globalset} and $|\supp(\bfz^{(\sfe^*,j^*)})|\leq \|\bdzeta\|\leq n/2$}\\
&\leq \frac{6^{k}\alpha}{n}\cdot|V_{q}||V_{s}|\cdot\sum_{d=1}^{+\infty}\left(\frac{6^{k}\alpha\|\bdzeta\|}{n}\right)^{d-1}\leq \frac{6^{k+1}\alpha}{n}\cdot|V_{q}||V_{s}|.\tag{using \eqref{eq:number_of_sequences} and $\|\bdzeta\|\leq 6^{-k-1}n$}
\end{align*}
Plugging this into \eqref{eq:bdzeta_i_expression} yields
\begin{align*}
\sum_{i=1}^{\ell}\frac{|R_{(i)}|}{|R|}\cdot\|\bdzeta_{(i)}\|&\leq k^{2}\sum_{\substack{q,s\in[t]\\ q=s}}|V_{q}||V_{s}|+\frac{6^{k+1}k^{2}\alpha}{n}\sum_{\substack{q,s\in [t]\\ q\neq s}}|V_{q}|^{2}|V_{s}|^{2}+k^{2}\sum_{i=1}^{\ell}\frac{|R_{(i)}|}{|R|}\cdot\left|E_{(i)}\right|\\
&\leq k^{2}\|\bdzeta\|+\frac{6^{k+1}k^{2}\alpha}{n}\|\bdzeta\|^{2}+k^{2}\sum_{i=1}^{\ell}\frac{|R_{(i)}|}{|R|}\cdot\left|\supp\left(\bfz^{(\sfe^*,j^*)}_{(i)}\right)\right|\\
&\leq 2k^{2}\|\bdzeta\|+k^{2}\sum_{i=1}^{\ell}\frac{|R_{(i)}|}{|R|}\cdot\phi(\bdzeta_{(i)},R_{(i)})\leq k^{2}\big(2\|\bdzeta\|+\phi(\bdzeta,R)+3c\big),
\end{align*}
where we used \eqref{eq:global_protocol_potential_increase} in the last transition. This proves the inequality \eqref{eq:bounded_growth}.

\paragraph{Step 3: understanding cyclicity.} In the rest of the proof, we assume $\bdzeta$ is not cyclic as in part (2) of the Lemma. For any fixed $i\in [\ell]$, we claim that assuming $E_{(i)}$ does not contain an exact cover of any sequence in $\bigcup_{q\in [t]}S(q,q)$, the restriction sequence $\bdzeta_{(i)}$ cannot be cyclic. Since $\bdzeta$ is not cyclic, it suffices to prove that $E_{(i)}$ does not contain any edge in the hypergraph $H_{\bdzeta}$, and $H_{\bdzeta_{(i)}}$ does not have cycles. The former statement follows easily from the assumption, since any edge in the hypergraph $H_{\bdzeta}$ is an exact cover of a sequence in $\bigcup_{q\in[t]}S(q,q,1)$. We next focus on the latter statement.

Let $E=\bigcup_{(\sfe,j)\in \calE\times [K]}\supp(\bfz^{(\sfe,j)})$ be the edge set of the hypergraph $H_{\bdzeta}$. We run a breadth-first search (BFS) on the hypergraph $H_{\bdzeta_{(i)}}$ (whose edge set is $E\cup E_{(i)}$) and rank all edges in $E\cup E_{(i)}$ by the time they are discovered in the BFS. This yields a total order on $E\cup E_{(i)}$ such that each hyperedge $e\in E\cup E_{(i)}$ is incident to at most one vertex that is covered by some hyperedge preceding $e$ in the order. Indeed, if some hyperedge $e$ violates this condition, then by the time the BFS discovers $e$, it has also found a cycle of distinct hyperedges $e=e_{0},e_{1},e_{2},\dots,e_{d-1}\in E\cup E_{(i)}$ and distinct vertices $v_{0},v_{1},v_{2},\dots,v_{d-1}\in \calV\times [n]$ such that $e_{i}$ is incident to both $v_{i}$ and $v_{i+1\pmod d}$, for each $i\in \{0,1,\dots,d-1\}$. Then the collection of edges in $\{e_{0},e_{1}\dots,e_{d-1}\}$ that are from $E_{(i)}$ is an exact cover of a sequence in $\bigcup_{q\in [t]}S(q,q)$, contradicting the assumption. Therefore, we have a total order on the edge set of $H_{\bdzeta_{(i)}}$ such that each edge is incident to at most one vertex that is covered by edges preceding it. It is then easy to see that $H_{\bdzeta_{(i)}}$ is cycle-free (see \Cref{subsec:general_notations} for the definition of cycle-freeness in hypergraphs).

We can now calculate
\begin{align*}
&\quad\sum_{i=1}^{\ell}\frac{|R_{(i)}|}{|R|}\cdot \ind{\bdzeta_{(i)}\text{ is cyclic}}=\sum_{i=1}^{\ell}\frac{\left|A^{(\sfe^*,j^*)}_{(i)}\right|}{\left|A^{(\sfe^*,j^*)}\right|}\cdot \ind{\bdzeta_{(i)}\text{ is cyclic}}\\
&\leq \sum_{i=1}^{\ell}\frac{\left|A^{(\sfe^*,j^*)}_{(i)}\right|}{\left|A^{(\sfe^*,j^*)}\right|}\cdot \ind{E_{(i)}\text{ contains an exact over of a sequence in }\bigcup\nolimits_{q\in [t]}S(q,q)}\\
&\leq \Pru{\bfy\in A^{(\sfe^*,j^*)}}{\supp(\bfy)\setminus \supp\left(\bfz^{(\sfe^*,j^*)}\right)\text{ contains an exact cover of a sequence in }\bigcup\nolimits_{q\in [t]}S(q,q)}\\
&\leq \sum_{q=1}^{t}\sum_{d=1}^{+\infty}|S(q,q,d)|\cdot n^{(k-2)d}\cdot\left(\frac{6^{k}\alpha n}{n^{k}}\right)^{d}\tag{using Corollary~\ref{cor:edge_appears_in_globalset} and $|\supp(\bfz^{(\sfe^*,j^*)})|\leq \|\bdzeta\|\leq n/2$}\\
&\leq \sum_{d=1}^{+\infty}\left(\frac{6^{k}\alpha\|\bdzeta\|}{n}\right)^{d}\leq \frac{6^{k+1}\alpha \|\bdzeta\|}{n}.\tag{using \eqref{eq:number_of_sequences} and $\|\bdzeta\|\leq 6^{-k-1}n$}
\end{align*}
This proves the inequality \eqref{eq:bounded_cyclic_prob} (under the assumption that $\bdzeta$ is not cyclic).
\end{proof}

\subsection{Most Rectangles in a Global Protocol are ``Good''}\label{subsec:most_are_good}

As mentioned in the introductory text to Section~\ref{sec:decomposition}, our next goal is to show that most structured rectangles produced by a global protocol are good. This is achieved in Lemma~\ref{lem:analysis_of_global_protocols} below, for which the following two definitions provide some convenient notations.

\begin{definition}
Let $\Pi$ be an $r$-round global protocol.  
For each $0 \leq d \leq r$, we let $\mathcal{R}^{d}(\Pi)$ denote the collection of structured rectangles obtained after $d$ rounds of communication in $\Pi$.  
In particular, we write $\calR^\leaf(\Pi)=\calR^r(\Pi)$ for the set of structured rectangles at the leaves of the protocol tree.
\end{definition}

\begin{definition}\label{def:deleting_bad_pairs}
Let $\Pi$ be a global protocol for $\DIHP(G,n,\alpha,K)$ and let $W$ be a positive real number. We define
$\calR^{\mathrm{bad}}(\Pi,W)$ as the following set of structured rectangles:
\begin{align*}
    \calR^{\mathrm{bad}}(\Pi,W) :=\left\{(\boldsymbol{\zeta},R)\in\calR^\leaf(\Pi): \text{$(\boldsymbol{\zeta},R)$ is not $(W,W)$-good as per \Cref{def:good_rec_2}}\right\}.
\end{align*}
\end{definition}

\begin{lemma}\label{lem:analysis_of_global_protocols}
    For any fixed distribution-labeled $k$-graph $G$, integer $K>0$ and parameter $\alpha \in (0,1)$, there exist constants $\eta_{1},\eta_{2}>0$ such that for any $(r,c)$-round global protocol $\Pi$ for $\mathsf{DIHP}(G,n,\alpha,K)$ such that $c\leq \eta_{1}^{-r}n$, we have
    \begin{equation}\label{eq:most_are_good_conclusion}
        \sum_{(\bdzeta,R)\in \calR^{\bad}(\Pi,\eta_{2}^{-r}c)}\calD_{\no}(R) \leq 0.01. 
    \end{equation}
\end{lemma}

We observe that the desired \Cref{lem:linear_decomposition} (restated below) follows immediately from \Cref{lem:analysis_of_global_protocols}. 

\lineardecomposition*

\begin{proof}[Proof of Lemma~\ref{lem:linear_decomposition} assuming Lemma~\ref{lem:analysis_of_global_protocols}]
Let $\eta_{1},\eta_{2}$ be the constants obtained from Lemma~\ref{lem:analysis_of_global_protocols}. We apply Lemma~\ref{lem:arbitrary_to_global} to transform $\Pi$ into an $(r,|\Pi|)$-round global protocol $\Pi^{\reff}$. We claim that $$\calR= \calR^{\mathrm{leaf}}(\Pi^{\mathrm{ref}})\setminus \calR^{\bad}(\Pi^{\mathrm{ref}},\eta_{2}^{-r}c)$$ satisfies the three conditions stated in the lemma.  

The second condition follows directly from Definition~\ref{def:deleting_bad_pairs}, and the third condition follows from the guarantee of Lemma~\ref{lem:arbitrary_to_global}.  
The first condition follows from Lemma~\ref{lem:analysis_of_global_protocols} together with the obvious identity
\[
    \sum_{(\boldsymbol{\zeta}, R) \in \calR^{\mathrm{leaf}}(\Pi^{\mathrm{ref}})} 
    \calD_{\no}(R) = 1.\qedhere
\]
\end{proof}

We are now ready to prove Lemma~\ref{lem:analysis_of_global_protocols}.

\begin{proof}[Proof of Lemma~\ref{lem:analysis_of_global_protocols}] We let $\eta_{2}=10^{-5}k^{-2}$; the value of $\eta_{1}$ will be determined later. For each joint input $\bfY\in \prod_{(\sfe,j)\in\calE\times [K]}\Omega^{\calU_{\sfe},\alpha n}$ and each nonnegative integer $d\leq r$, we let $\big(\bdzeta(\bfY,d),R(\bfY,d)\big)$ be the unique structured rectangle in $\calR^{d}(\Pi)$ that contains $\bfY$. Note that the support of each component of $\bdzeta(\bfY,d)$ is non-decreasing in $d$.

In the rest of the proof, we use $\bfY$ to denote a uniformly random element of the joint input space $\prod_{(\sfe,j)\in\calE\times [K]}\Omega^{\calU_{\sfe},\alpha n}$, i.e. we have $\bfY\sim\calD_{\no}$. Note that the desired conclusion \eqref{eq:most_are_good_conclusion} is equivalent to
\begin{equation}\label{eq:most_are_good_goal}
\Pru{\bfY}{\big(\bdzeta(\bfY,r),R(\bfY,r)\big)\text{ is not }(\eta_{2}^{-r}c,\eta_{2}^{-r}c)\text{-good}}\leq 0.01.
\end{equation}
Recall from Definition~\ref{def:good_rec_2} that a structured rectangle $(\bdzeta,R)$ with
\(R=\prod_{(\sfe, j) \in \calE \times [K]} A^{(\sfe, j)}\) is not $(\eta_{2}^{-r}c,\eta_{2}^{-r}c)$-good only if one of the following holds:
\begin{enumerate}[label=(\arabic*)]
\item For some $(\sfe,j)\in \calE\times [K]$ we have $\left|A^{(\sfe,j)}\right| \leq 2^{-\eta_2^{-r}c}\left|\Omega^{\calU_{\sfe},\alpha n}_{\bfz^{(\sfe,j)}}\right|$. In this case, by Definition~\ref{def:restriction_potential} we also have $\phi(\bdzeta,R)\geq \eta_{2}^{-r}c$.
\item The weight of $\bdzeta$ is at least $\eta_{2}^{-r}c$, i.e. $\|\bdzeta\|\geq \eta_{2}^{-r}c$.
\item The restriction sequence $\bdzeta$ is cyclic.
\end{enumerate}

\paragraph{Bounding potential.} By Definition~\ref{def:global_protocol}, for any $d\in [r]$ we have (with probability 1)
\[
\Exu{\bfY}{\phi\big(\bdzeta(\bfY,d),R(\bfY,d)\big)\,\Big|\,\bdzeta(\bfY,d-1),R(\bfY,d-1)}\leq \phi\big(\bdzeta(\bfY,d-1),R(\bfY,d-1)\big)+3c.
\]
Taking expectation over $\bfY$ and taking sum over $d$ yields
\begin{equation}\label{eq:expectation_potential}
\Exu{\bfY}{\phi\big(\bdzeta(\bfY,d),R(\bfY,d)\big)\Big.}\leq 3dc\qquad\text{for any }d\in[r],
\end{equation}
and by Markov's inequality we have
\begin{equation}\label{eq:most_are_good_1}
\Pru{\bfY}{\phi\big(\bdzeta(\bfY,r),R(\bfY,r)\big)\geq \eta_{2}^{-r}c\Big.}\leq 3r\cdot\eta_{2}^{r}\leq 10^{-3}.
\end{equation}

\paragraph{Bounding weight.} By Lemma~\ref{lem:bounded_growth}, for any $d\in [r]$ we have (with probability 1)
\begin{align*}
&\quad\Exu{\bfY}{\min\left\{\|\bdzeta(\bfY,d)\|,6^{-k-1}n\right\}\,\Big|\,\bdzeta(\bfY,d-1),R(\bfY,d-1)}\\
&\leq k^{2}\left(2\min\left\{\|\bdzeta(\bfY,d-1)\|,6^{-k-1}n\right\}+\phi\big(\bdzeta(\bfY,d-1),R(\bfY,d-1)\big)+3c\right).
\end{align*}
Taking expectation over $\bfY$,  applying \eqref{eq:expectation_potential}, and inducting over $d$, we get
\begin{equation}\label{eq:expectation_weight}
\Exu{\bfY}{\min\left\{\|\bdzeta(\bfY,s)\|,6^{-k-1}n\right\}}\leq \sum_{d=1}^{s}3dc\cdot(2k^{2})^{s+1-d}\leq (10k^{2})^{s}c\qquad\text{for any }s\in [r].
\end{equation}
By Markov's inequality we have
\begin{equation}\label{eq:most_are_good_2}
\Pru{\bfY}{\|\bdzeta(\bfY,r)\|\geq \eta_{2}^{-r}c}\leq (10\eta_{2} k^{2})^{r}\leq 10^{-3}
\end{equation}
as long as $\eta_{2}^{-r}c\leq 6^{-k-1}n$, which we can easily guarantee since $c\leq \eta_{1}^{r}n$ and we can choose $\eta_{1}$ to be sufficiently small.

\paragraph{Probability of cyclicity.} By Lemma~\ref{lem:bounded_growth}, for any $d\in [r]$ we have (with probability 1)
\begin{align*}
&\quad\Pru{\bfY}{\bdzeta(\bfY,d)\text{ is cyclic but }\bdzeta(\bfY,d-1)\text{ is not cyclic}\,\Big|\,\bdzeta(\bfY,d-1),R(\bfY,d-1)}\\
&\leq \frac{6^{k+1}}{n}\min\left\{\|\bdzeta(\bfY,d-1)\|, 6^{-k-1}n\right\}.
\end{align*}
Taking expectation over $\bfY$, applying \eqref{eq:expectation_weight} and taking union bound over $d$, we get
\begin{equation}\label{eq:most_are_good_3}
\Pru{\bfY}{\bdzeta(\bfY,r)\text{ is cyclic}}\leq \frac{6^{k+1}}{n}\cdot (10k^{2})^{r}c\leq 10^{-3}
\end{equation}
as long as $c\leq 6^{-k-5}(10k^{2})^{-r}n$, which we can easily guarantee since $c\leq \eta_{1}^{r}n$ and we can choose $\eta_{1}$ to be sufficiently small.

Combining \eqref{eq:most_are_good_1}, \eqref{eq:most_are_good_2} and \eqref{eq:most_are_good_3}, we conclude that \eqref{eq:most_are_good_goal} holds if $\eta_{1}$ is a small enough constant.
\end{proof}

%% file: discrepancy.tex
\section{One-Sided Discrepancy of Good Rectangles}\label{sec:discrepancy}
The goal of this section is to prove the discrepancy lemma (Lemma~\ref{lem:discrepancy_bound}), modulo a Fourier analytic lemma that we prove in Section~\ref{sec:induction}. We begin with a high-level overview of the proof strategy.

Recall that in \Cref{lem:discrepancy_bound}, we are given a restriction sequence \(\bdzeta\) and a $\bdzeta$-global rectangle $R = \prod_{(\sfe,j)\in \calE\times [K]} A^{(\sfe,j)}$. We will define a probability density function $f^{(\sfe,j)}$ on $\ZmodN^{\calV\times [n]}$, induced by the set $A^{(\sfe,j)}$. As is shown in \cite[Section 7.1]{FMW25b}, the quantities $\calD_{\yes}(R)$ and $\calD_{\no}(R)$ can be related by the identity (see Lemma~\ref{lem:relating_yes_no})
\begin{align}\label{eq:relating_yes_no}
\calD_{\yes}(R)=\calD_{\no}(R)\cdot \Exu{x\in \ZmodN^{\calV\times [n]}}{\prod_{(\sfe,j)\in \calE\times [K]}f^{(\sfe,j)}(x)}.
\end{align}
Thus, it suffices to show that the expectation of the product \(\prod_{(\sfe,j)} f^{(\sfe,j)}(x)\) is not much smaller than 1.

To analyze each function $f^{(\sfe,j)}$ (associated to a single player), it is shown in \cite[Section 7.2]{FMW25b} that $f^{(\sfe,j)}$ can be decomposed as the product of two functions, which represent the ``structure part'' and ``pseudorandom part'' of $f^{(\sfe,j)}$, respectively. In Section~\ref{subsec:separating_structure_random}, we will present the same decomposition as in \cite{FMW25b} but using a slightly different formalism.

It is after this decomposition that our strategy diverges from \cite{FMW25b}. 

In \cite{FMW25b}, the analytical tool of \emph{global hypercontractivity} is used to obtain $L^{2}$-control over (the Fourier-level weights of) the ``pseudorandom part'' functions. Although this $L^{2}$-control is, in a sense, nearly optimal on its own, the analysis of \cite{FMW25b}\footnote{For the specific case of Max-Cut, the paper~\cite{FMWa} combines the pseudorandom-part and structure-part analyses in a slightly different manner, which, however, suffers the same loss as in \cite{FMW25b}.} incurs a significant loss when integrating the pseudorandom-part control with the structure-part information.\footnote{Consequently, both~\cite{FMW25b} and~\cite{FMWa} obtain only $\Omega(n^{1/3})$ communication lower bounds.}

It turns out that the ``structure'' inherent in the structure-part functions is more optimally captured through the lens of \emph{Fourier $\ell^{1}$-norm} than via $L^{2}$-based controls. However, maintaining Fourier $\ell^{1}$-norm controls for \emph{products} of multiple functions is a much more intricate task than in the $L^{2}$ world. Notably, this challenge was successfully addressed in~\cite{KK19} in the context of Max-Cut. In Section~\ref{subsec:discrepancy_bound}, we demonstrate how to preserve Fourier $\ell^{1}$-norm control through an induction lemma (Lemma~\ref{lem:induction}, the proof of which is deferred to Section~\ref{sec:induction}), resembling the approach of~\cite{KK19}. Before that, we lay out some Fourier-analytic framework in Section~\ref{subsec:fourier_growth}, and in Section~\ref{subsec:structured} we use the framework to analyze products of structure-part functions.

\subsection{Relating YES and NO Distributions}\label{subsec:calculating_discrepancy}
In this subsection, we introduce a formula from \cite{FMW25b} that expresses the ratio $\calD_{\yes}(R)/\calD_{\no}(R)$ as the expectation of a product of density functions. 

Since the YES distribution $\calD_{\yes}$ is defined from the Markov kernel in Definition~\ref{def:Markov_kernel}, the main task is to analyze that Markov kernel. As demonstrated in \cite{FMW25b}, it is convenient to view Markov kernels as \emph{pull-back} linear operators that map functions on the target space to functions on the source space.

\begin{notation}\label{notation:linear_operator}
Fix a $k$-universe $\calU$, a nonnegative integer $m\leq |\calU|$, and a one-wise independent distribution $\mu$ over $\ZNk$. The (right) stochastic matrix \(\bfP^{\calU,m}_{\mu}: \mathbb{Z}_N^{\bigcup\calU} \times \Omega^{\calU,m} \to \mathbb{R}\), defined in \Cref{def:Markov_kernel}, can be viewed as a linear operator
\[
\bdP^{\calU,m}_{\mu} : L^2(\Omega^{\calU, m}) \to L^2\big(\mathbb{Z}_N^{\bigcup\calU}\big)
\]
given by
\[
\bdP^{\calU,m}_{\mu}[f](x) = \sum_{\bfy \in \Omega^{\calU_{\sfe}, \alpha n}} \bfP^{\calU,m}_{\mu}(x, \bfy) f(\bfy),
\]
for all \(x \in \mathbb{Z}_N^{\bigcup\calU}\) and \(f \in L^2(\Omega^{\calU,m})\).
\end{notation}

We denote this pull-back operator by the italic bold symbol $\bdP^{\calU,m}_{\mu}[\cdot]$, distinguishing it from the matrix expression $\bfP^{\calU,m}_{\mu}(\cdot, \cdot)$ to reflect that, while formally distinct, the two represent the same underlying Markov transition.

The operator $\bdP_\mu^{\calU,m}$ satisfies the following two basic properties (see \cite[Section 7.1]{FMW25b} for their proofs). 
\begin{proposition}\label{prop:infinity_norm_contraction_P}
For any $f\in L^{2}(\Omega^{\calU,m})$, we always have $\left\|\bdP^{\calU,m}_{\mu}[f]\right\|_{\infty}\leq \|f\|_{\infty}$.
\end{proposition}

\begin{proposition}
The operator $\bdP^{\calU,m}_{\mu}$ maps a density function on $\Omega^{\calU,m}$ to a density function on $\ZmodN^{\bigcup\calU}$.
\end{proposition}

The following notation will be useful throughout this section.
\begin{notation}\label{notation:density_function}
Given a $k$-universe $\calU$, a nonnegative integer $m\leq |\calU|$, and a nonempty set $A\subseteq \Omega^{\calU,m}$, the density function of the uniform distribution on $A$ is denoted by $\phi_{A}:\Omega^{\calU,m}\rightarrow[0,\infty)$, specifically defined as 
\[\phi_{A}(\bfy):=\begin{cases}
\left|\Omega^{\calU,m}\right|/|A|, &\text{if }\bfy\in A,\\
0, &\text{if }\bfy\not\in A.
\end{cases}\]
\end{notation}
We remark that the density function \(\bdP^{\calU,m}_{\mu}[\phi_A]\) corresponds to taking the joint distribution of $(x,\bfy)\in \ZmodN^{\bigcup\calU}\times \Omega^{\calU,m}$ in the Markov transition (described in Definition~\ref{def:Markov_kernel}), conditioning on the event $\{\bfy\in A\}$, and then taking the marginal distribution of $x$.

We are now ready to present the formula for $\calD_{\yes}(R)/\calD_{\no}(R)$ from \cite{FMW25b}.

\begin{lemma}[\cite{FMW25b}, Lemma 7.5]\label{lem:relating_yes_no}
Fix a $\DIHP(G,n,\alpha,K)$ communication game, where $G=(\calV,\calE,N,(\mu_{\sfe})_{\sfe\in\calE})$. Given a rectangle $R=\prod_{(\sfe,j)\in\calE\times [K]}A^{(\sfe,j)}$, where $A^{(\sfe,j)}\subseteq \Omega^{\calU_{\sfe},\alpha n}$, we have
\[
\calD_{\yes}(R)=\calD_{\no}(R)\cdot \Exu{x\in \ZmodN^{\calV\times [n]}}{\prod_{(\sfe,j)\in \calE\times [K]}\bdP^{\calU_{\sfe},\alpha n}_{\mu_{\sfe}}\Big[\phi_{A^{(\sfe,j)}}\Big]\circ\proj_{\sfe}(x)}.
\]
\end{lemma}

\subsection{Separating Structured and Pseudorandom Parts}\label{subsec:separating_structure_random}
As mentioned earlier, the goal of this subsection is to express each function
\begin{equation}\label{eq:function_to_separate}
    \bdP^{\calU_{\sfe},\alpha n}_{\mu_{\sfe}}\Big[\phi_{A^{(\sfe,j)}}\Big]\circ\proj_{\sfe},\qquad\text{for }(\sfe,j)\in\calE\times [K]
\end{equation}
as the product of two functions: the structured part and the pseudorandom part. We start with the definition of the structured part. 
\begin{definition}\label{def:structured_function}
Fix a $\DIHP(G,n,\alpha,K)$ communication game, where $G=(\calV,\calE,N,(\mu_{\sfe})_{\sfe\in\calE})$. Let $\bfz$ be a restriction on the space $\Omega^{\calU,m}$, and an edge $\sfe\in \calE$. We define a density function $g_{\sfe,\bfz}:\ZmodN^{\calV\times [n]}\rightarrow [0,+\infty)$ by
\[
g_{\sfe,\bfz}(x):=\prod_{e\in \supp(\bfz)}N^{k}\mu_{\sfe}\big(x_{|e}-\bfz(e)\big).
\]
\end{definition}

Next, we define the pseudorandom part. For that purpose, we introduce another Markov kernel closely related to the one in Definition~\ref{def:Markov_kernel}. While the Markov transition in Definition~\ref{def:Markov_kernel} samples a labeled matching $\bfy\in\Omega^{\calU,m}$ based on an input $x\in\ZmodN^{\bigcup\calU}$, it is also natural to fix a matching $M\in\calM_{\calU,m}$ and sample a labeling of $M$ based on an input $x\in \ZmodN^{\bigcup\calU}$ (or an $x\in\ZmodN^{\Lambda}$, if $\tcup\calU\subseteq \Lambda$), as captured by the next definition.

\begin{definition}\label{def:Markov_R}
Let $\calU$ be a $k$-universe embedded in a finite set $\Lambda$ (recall Notation~\ref{notation:embed_universe}). Fix a one-wise independent distribution $\mu$ over $\ZNk$ and a matching $M\in \calM_{\calU,m}$ of size $m\leq |\calU|$. We define a right stochastic matrix $\bfR^{\Lambda,M}_{\mu}:\ZmodN^{\Lambda}\times \Map{M}{\ZmodN^{k}}\rightarrow [0,+\infty)$ by
\[
\bfR^{\Lambda,M}_{\mu}(x,\bdxi):=\prod_{e\in M}\mu\big(x_{|e}-\bdxi(e)\big).
\]
We may then define a linear operator $\bdR^{\Lambda,M}_{\mu}:L^{2}\big(\Map{M}{\ZNk}\big)\rightarrow L^{2}\big(\ZmodN^{\Lambda}\big)$ as in Notation~\ref{notation:linear_operator}. 
\end{definition}

Analogously to Propositions~\ref{prop:infinity_norm_contraction_P} and~\ref{prop:preverses_density_function}, the operator $\bdR_\mu^{\Lambda,M}$ satisfies the following two basic properties. 
\begin{proposition}\label{prop:infinity_norm_contraction}
For any $f\in L^{2}\big(\Map{M}{\ZNk}\big)$, we always have $\left\|\bdR^{\Lambda,M}_{\mu}[f]\right\|_{\infty}\leq \|f\|_{\infty}$.
\end{proposition}

\begin{proposition}\label{prop:preverses_density_function}
The operator $\bdR^{\Lambda,M}_{\mu}$ maps a density function on the space $\Map{M}{\ZNk}$ to a density function on $\ZmodN^{\Lambda}$.
\end{proposition}

In order to obtain a formal relation between the Markov operator $\bdR^{\Lambda,M}_{\mu}$ and $\bdP^{\calU,m}_{\mu}$, we define a canonical embedding map in Definition~\ref{def:embedding}. The following notations prepare for the definition.

\begin{definition}\label{def:embedding}
Fix a $k$-universe $\calU$ and a nonnegative integer $m\leq |\calU|$. For any restriction $\bfz$ on the space $\Omega^{\calU,m}$ and any matching $M\in \calM_{\calU,m,\bfz}$, there is a canonical embedding
\[
\fraki_{M,\bfz}:\Map{M}{\ZNk}\xhookrightarrow{\phantom{mmm}} \Omega^{\calU,m}_{\bfz}.
\]
This embedding proceeds by mapping any $\bdxi\in \Map{M}{\ZNk}$ to a labeled matching $\bfy\in \Omega^{\calU,m}_{\bfz}$ defined by
\begin{enumerate}
\item $\bfy(e)=\bfz(e)$ for $e\in \supp(\bfz)$,
\item $\bfy(e)=\bdxi(e)$ for $e\in M$,
\item $\bfy(e)=\nil$ for all other $e\in \tprod\calU$.
\end{enumerate}
\end{definition}

Note that, as $M$ ranges over all matchings in $\calM_{\calU,m,\bfz}$, the images of the maps $\fraki_{M,\bfz}$ form a partition of the space $\Omega^{\calU,m}$. 

We are now ready to present the separation of the function \eqref{eq:function_to_separate} into structured and pseudorandom parts.

\begin{lemma}\label{lem:separating_structure_pseudorandom}
Fix a $\DIHP(G,n,\alpha,K)$ communication game, where $G=(\calV,\calE,N,(\mu_{\sfe})_{\sfe\in\calE})$, and fix an edge $\sfe\in \calE$. Given a restriction $\bfz$ on $\Omega^{\calU_{\sfe},\alpha n}$ and a $\bfz$-global set $A\subseteq \Omega^{\calU_{\sfe},\alpha n}_{\bfz}$, for any $x\in \ZmodN^{\calV\times [n]}$ we have
\begin{equation}\label{eq:separating_structure_from_randomness}
\bdP^{\calU_{\sfe},\alpha n}_{\mu_{\sfe}}[\phi_{A}]\big(\proj_{\sfe}(x)\big)=g_{\sfe,\bfz}(x)\sum_{M}\frac{\left|\fraki_{M,\bfz}^{-1}(A)\right|}{|A|}\bdR^{\calV\times [n],\,M}_{\mu_{\sfe}}\left[\phi_{\fraki_{M,\bfz}^{-1}(A)}\right](x),
\end{equation}
where the sum is over all matchings $M\in \calM_{\calU_{\sfe},\alpha n,\bfz}$.
\end{lemma}

\begin{proof}
For elements $x\in \ZmodN^{\calV\times [n]}$ and $\bfy\in \Omega^{\calU_{\sfe},m}_{\bfz}$, it is easy to deduce from Definition~\ref{def:Markov_kernel} that
\[
\bfP^{\calU_{\sfe},\alpha n}_{\mu_{\sfe}}(\proj_{\sfe}(x),\bfy)=\frac{(N^{k})^{\alpha n}}{\left|\Omega^{\calU_{\sfe},\alpha n}\right|}\prod_{e \in\supp(\bfy)} \mu_{\sfe}\big(x_{|e} - \bfy(e)\big).
\]
Therefore, for any $x\in\ZmodN^{\calV\times [n]}$ we have
\begin{align}
&\quad\bdP^{\calU_{\sfe},\alpha n}_{\mu_{\sfe}}[\phi_{A}]\big(\proj_{\sfe}(x)\big)\nonumber\\
&=\sum_{\bfy\in A}\left(\frac{\left|\Omega^{\calU_{\sfe},\alpha n}\right|}{|A|}\cdot \frac{(N^{k})^{\alpha n}}{\left|\Omega^{\calU_{\sfe},\alpha n}\right|}\prod_{e \in\supp(\bfy)} \mu_{\sfe}\big(x_{|e} - \bfy(e)\big)\right)\nonumber\\
&=\Exu{\bfy\in A}{\prod_{e\in \supp(\bfz)}\left(N^{k}\mu_{\sfe}\big(x_{|e}-\bfz(e)\big)\right)\cdot \prod_{e\in \supp(\bfy)\setminus \supp(\bfz)}\left(N^{k}\mu_{\sfe}\big(x_{|e}-\bfy(e)\big)\right)}\nonumber\\
&=N^{k(\alpha n-|\supp(\bfz)|)}\cdot g_{\sfe,\bfz}(x)\sum_{M}\Exu{\bfy\in A}{\ind{\supp(\bfy)\setminus\supp(\bfz)=M}\cdot\bfR^{\calV\times [n],\,M}_{\mu_{\sfe}}\left(x,\fraki_{M,\bfz}^{-1}(\bfy)\right)},\label{eq:separating_lemma_LHS}
\end{align}
where the sum on the last line is over all matchings $M\in \calM_{\calU_{\sfe},\alpha n,\bfz}$.

On the other hand, for any $x\in\ZmodN^{\calV\times [n]}$ and any $M\in\calM_{\calU_{\sfe},\alpha n,\bfz}$, we have
\begin{align}
\bdR^{\calV\times [n],\,M}_{\mu_{\sfe}}\left[\phi_{\fraki_{M,\bfz}^{-1}(A)}\right](x)&=\sum_{\bdxi\in\fraki_{M,\bfz}^{-1}(A)}\left(\frac{(N^{k})^{|M|}}{\left|\fraki_{M,\bfz}^{-1}(A)\right|}\bfR^{\calV\times [n],\,M}_{\mu_{\sfe}}(x,\bdxi)\right)\nonumber\\
&=N^{k(\alpha n-|\supp(\bfz)|)}\cdot\Exu{\bdxi\in\fraki_{M,\bfz}^{-1}(A)}{\bfR^{\calV\times [n],\,M}_{\mu_{\sfe}}(x,\bdxi)}.\label{eq:separating_lemma_RHS}
\end{align}
Comparing the right-hand sides of \eqref{eq:separating_lemma_LHS} and \eqref{eq:separating_lemma_RHS} immediately yields the conclusion.
\end{proof}

We remark that on the right-hand side of \eqref{eq:separating_structure_from_randomness}, the first factor $g_{\sfe,\bfz}(x)$ is the structured part, and the sum over $M$ is the pseudorandom part.

\subsection{Interlude: Controlling Fourier Growth}\label{subsec:fourier_growth}

As mentioned earlier, the key idea of~\cite{KK19} is to maintain suitable Fourier $\ell^{1}$-type control under taking products of functions. To this end, they introduce the notion of \emph{$(C,s^*)$-boundedness}, which imposes a carefully calibrated upper bound on the Fourier $\ell^{1}$-norm of the degree-$d$ component of a function, separately for each $d$. 

For fixed parameters $C$ and $s^*$, the bound at level $d$ takes the following form:
\begin{definition}[{\cite[Definition 4.3]{KK19}}]\label{def:Fourier_growth_bound}
Let $n,d$ be integers such that $1\leq d\leq n$, and let $C,s^*$ be positive real numbers. We define a value $F_{C}(n,d,s^*)$ by
\[
F_{C}(n,d,s^*)=\left(\frac{C\sqrt{n\cdot\max\{s^*,d\}}}{d}\right)^{d/2}=\begin{cases}
\left(C\sqrt{s^*n}/d\right)^{d/2},&\text{if }d\leq s^*,\\
\left(C^{2}n/d\right)^{d/4},&\text{if }d>s^*.
\end{cases}
\]
\end{definition}

We record the following two simple observations about the Fourier growth bound $F_{C}(n,\ell,s^*)$.
\begin{proposition}\label{prop:Fourier_bound_non_decreasing}
The function $F_{C}(n,d,s^*)$ is non-decreasing in the first and third variables.
\end{proposition}
\begin{proposition}\label{prop:Fourier_growth_crude_bound}
Let $n,d$ be integers such that $1\leq d\leq n$, and let $C,s^*$ be positive real numbers. We always have
\[
F_{C}(n,d,s^*)\leq 2^{s^*/4}\left(\frac{C^{2}n}{d}\right)^{d/4}.
\]
\end{proposition}
\begin{proof}
In the case $d>s^*$, this directly follows from Definition~\ref{def:Fourier_growth_bound}. In the case $d\leq s^*$, the conclusion follows by simply using the inequality $s^*/d\leq 2^{s^*/d}$. 
\end{proof}

We now introduced a slightly generalized version of the $(C,s^*)$-boundedness notion in \cite[Definition 4.3]{KK19}. Recall from Section~\ref{subsec:general_notations} that the notation $\|\cdot\|_{\sfW}$ stands for the Fourier $\ell^{1}$-norm of a function.

\begin{definition}\label{def:Cs_star_bounded}
    Fix parameters $n,C,s^*\geq 1$ and $\delta\geq 0$. Given a finite set $\Lambda$ and a function $f: \ZmodN^{\Lambda}\rightarrow [0,+\infty)$, we say that $f$ is $(n,C,s^*,\delta)$-bounded if it satisfies the following conditions:
    \begin{enumerate}[label=(\arabic*)]
        \item The expected value of $f$ lies in the range $[1-\delta,1+\delta]$, i.e. $1-\delta\leq \Ex{f}\leq 1+\delta$. 
        \item For any integer $d$ such that $1\leq d\leq  C^{-2}n$, we have $\left\|f^{=d}\right\|_{\sfW}\leq F_{C}(n,d,s^*)$.
        \item The maximum value of $f$ is at most $2^{s^*}$, i.e. $\|f\|_{\infty}\leq 2^{s^*}$.
    \end{enumerate}
\end{definition}

Although the second condition in Definition~\ref{def:Cs_star_bounded} imposes upper bounds only on the first $\lfloor C^{-2}n\rfloor +1$ levels, the following proposition shows that some control over higher levels can be obtained via the third condition.

\begin{proposition}\label{prop:Cs_star_high_degree_bound}
Fix parameters $n,C,s^*\geq 1$ and $\delta\geq 0$. For any $f:\ZmodN^{\Lambda}\rightarrow[0,+\infty)$ that is $(n,C,s^*,\delta)$-bounded, we have
\[
\left\|f^{=d}\right\|_{\sfW}\leq 2^{s^*/2}\left(\frac{3N|\Lambda|}{d}\right)^{d/2}
\]
for any nonnegative integer $d$.
\end{proposition}
\begin{proof}
We can simply apply Cauchy-Schwarz inequality to the Fourier 1-norm and calculate as follows:
\begin{align*}
\left\|f^{=d}\right\|_{\sfW}&=\sum_{x\in\ZmodN^{\Lambda},\;\|x\|_{\sfH}=d}\left|\widehat{f}(x)\right|\leq \sqrt{\sum_{x\in\ZmodN^{\Lambda},\;\|x\|_{\sfH}=d}1}\cdot \sqrt{\sum_{x\in\ZmodN^{\Lambda},\;\|x\|_{\sfH}=d}\left|\widehat{f}(x)\right|^{2}}\\
&\leq \sqrt{\binom{|\Lambda|}{d}N^{d}}\cdot \left\|f\right\|_{2}\leq \sqrt{\binom{|\Lambda|}{d}N^{d}}\cdot \left\|f\right\|_{\infty}\leq 2^{s^*/2}\left(\frac{3N|\Lambda|}{d}\right)^{d/2}.\qedhere
\end{align*}
\end{proof}

The next proposition shows that in the notion of $(n,C,s^*,\delta)$-boundedness, small modifications of the first parameter $n$ can be offset by adjustments on the second parameter $C$.
\begin{proposition}\label{prop:boundedness_adjustment}
Fix parameters $n,C,s^*\geq 1$ and $\delta\geq 0$. For any integer $n'$ such that $n/2\leq n'\leq 2n$, any $(n,C,s^*,\delta)$-bounded function is also $(n',2C,s^*,\delta)$-bounded.
\end{proposition}
\begin{proof}
This is because $(2C)^{-2}n'\leq C^{-2}n$ and $F_{2C}(n',d,s^*)\geq F_{C}(n,d,s)$ for any $d$.
\end{proof}
\subsection{Analyzing the Structured Part}\label{subsec:structured}

Recall that Lemma~\ref{lem:relating_yes_no} expresses the ratio $\calD_{\yes}(R)/\calD_{\no}(R)$ as the expectation of a product of $K|\calE|$ functions. By Lemma~\ref{lem:separating_structure_pseudorandom}, each of these functions admits a decomposition into a structured part and a pseudorandom part. Consequently, we are led to analyze a product of $2K|\calE|$ functions.

Our strategy is to first control the Fourier growth of the product of the $K|\calE|$ structure-part functions. We then show that this control does not deteriorate too rapidly as the $K|\calE|$ pseudorandom-part functions are incorporated into the product one at a time.

The product of the $K|\calE|$ structure-part functions, as highlighted in the following definition, is the object of study of this subsection. 

\begin{definition}\label{def:g_bdzeta}
Given a restriction sequence \(\bdzeta = \left(\bfz^{(\sfe, j)}\right)_{(\sfe, j) \in \calE \times [K]}\), we define the product function $g_{\bdzeta}\in L^{2}\big(\ZmodN^{\calV\times [n]}\big)$ by (recalling Definition~\ref{def:structured_function})
\begin{equation}\label{eq:g_full_expansion}
g_{\bdzeta}(x)=\prod_{(\sfe,j)\in\calE\times [K]}g_{\sfe,\bfz^{(\sfe,j)}}(x)=\prod_{(\sfe,j)\in\calE\times [K]}\prod_{e\in \supp(\bfz^{(\sfe,j)})}N^{k}\mu_{\sfe}\big(x_{|e}-\bfz^{(\sfe,j)}(e)\big).
\end{equation}
\end{definition}

The function $g_{\bdzeta}$ defined in Definition~\ref{def:g_bdzeta} satisfies the following nice properties.

\begin{lemma}\label{lem:structure_part_mean}
Fix a $\DIHP(G,n,\alpha,K)$ communication game, where $G=(\calV,\calE,N,(\mu_{\sfe})_{\sfe\in\calE})$. Let $\bdzeta$ be a restriction sequence that is not cyclic (as per Definition~\ref{def:cyclic_res_seq}), and let $V_{1},\dots,V_{t}\subseteq \calV\times [n]$ be the nontrivial connected components of the hypergraph $H_{\bdzeta}$ (as defined in Definition~\ref{def:weight_of_restriction}). Then the following holds:
\begin{enumerate}[label=(\arabic*)]
\item The function $g_{\bdzeta}$ is a probability density function, i.e. $\Ex{g_{\bdzeta}}=1$.
\item For any $b\in \ZmodN^{\calV\times [n]}$, we have $\inp{g_{\bdzeta}}{\chi_{b}}\neq 0$ only if $\supp(b)\subseteq \bigcup_{i\in [t]}V_{i}$ and $|\supp(b)\cap V_{i}|\neq 1$ for any $i\in [t]$. 
\end{enumerate}
\end{lemma}
\begin{proof}
The first statement is identical to \cite[Lemma 7.10]{FMW25b}, so we omit its proof. We prove the second statement below.

Suppose $b\in \ZmodN^{\calV\times [n]}$ is a character index such that $\supp(b)\not\subseteq \bigcup_{i\in [t]}V_{i}$. Then it is easy to see that when $x$ is a uniformly random vector in $\ZmodN^{\calV\times [n]}$, the restriction of $x$ to $\bigcup_{i\in [t]}V_{i}$ is independent with $\sum_{v\in \calV\times [n]}b_{v}x_{v}$ (which is uniformly distributed in $\ZmodN$). Since $g_{\bdzeta}(x)$ only depends on the restriction of $x$ to $\bigcup_{i\in [t]}V_{i}$ and $\chi_{b}(x)$ only depends on $\sum_{v\in \calV\times [n]}b_{v}x_{v}$, it follows that $g_{\bdzeta}(x)$ and $\chi_{b}(x)$ are independent, and hence we have \[\inp{g_{\bdzeta}}{\chi_{b}}=\Exu{x}{g_{\bdzeta}(x)\overline{\chi_{b}(x)}}=\Exu{x}{g_{\bdzeta}(x)}\Exu{x}{\overline{\chi_{b}(x)}}=0.\]

In the rest of the proof, we assume that $\supp(b)\subseteq \bigcup_{i\in [t]}V_{i}$ and (without loss of generality) $|\supp(b)\cap V_{1}|=1$. 

We write $\bdzeta=\big(\bfz^{(\sfe,j)}\big)_{(\sfe,j)\in \calE\times [K]}$. Let $E$ be the edge set of $H_{\bdzeta}$ and let $E_{1}$ be the collection of edges in $E$ that lie in $V_{1}$. For each edge $e\in E$, we let $\langle e\rangle$ be the original hyperedge $\sfe\in \calE$ such that $e\in \tprod \calU_{\sfe}$. Since the edge sets $\supp(\bfz^{(\sfe,j)})$ are pairwise disjoint as $(\sfe,j)$ ranges in $\calE\times [K]$, we can define a map $\widetilde{\bfz}:E\rightarrow\ZNk$ by letting $\widetilde{\bfz}(e)=\bfz^{(\sfe,j)}(e)$ for each $(\sfe,j)\in \calE\times [K]$ and each $e\in \supp(\bfz^{(\sfe,j)})$. Now we can write
\begin{equation}\label{eq:structure_part_rewrite}
g_{\bdzeta}(x)=\prod_{e\in E}N^{k}\mu_{\langle e\rangle}\Big(x_{|e}-\widetilde{\bfz}(e)\Big).
\end{equation}

Let $v^*$ be the single element of $\supp(b)\cap V_{1}$. We run a breadth-first search (BFS) on the hypergraph $H_{\bdzeta}$ starting at the vertex $v^*$, and rank all edges in $E_{1}$ by the time they are discovered in the BFS.\footnote{The BFS argument has been used in the proofs of Lemma~\ref{lem:bounded_growth} and \cite[Lemma 7.10]{FMW25b}.} This yields a total order $\prec$ on $E_{1}$ such that each hyperedge $e\in E_{1}$ is incident to at most one vertex that is covered by some hyperedge preceding $e$ in the order. Indeed, if some hyperedge $e$ violates this condition, then by the time the BFS discovers $e$, it has also found a cycle of distinct hyperedges $e=e_{0},e_{1},e_{2},\dots,e_{\ell-1}\in E_{1}$ and distinct vertices $v_{0},v_{1},v_{2},\dots,v_{\ell-1}\in V_{1}$ such that $e_{i}$ is incident to both $v_{i}$ and $v_{i+1\pmod \ell}$, for each $i\in \{0,1,\dots,\ell-1\}$. The $\ell$ hyperedges $e_{0},\dots,e_{\ell-1}$ thus together cover at most $\ell(k-1)$ vertices, violating the assumption that $H_{\bdzeta}$ is cycle-free.

For notational convenience, we extend the total order $\prec$ on $E_{1}$ to $E$, by letting $e'\prec e$ for all edges $e'\in E\setminus E_{1}$ and $e\in E_{1}$. Using the total order on $E$, we now analyze the expectation $\Exs{x}{g_{\bdzeta}(x)\overline{\chi_{b}(x)}}$. For any $e \in E_{1}$, we have:
\begin{align}
&\quad \Exu{x\in \ZmodN^{\calV \times [n]}} {\overline{\chi_{b}(x)}\prod_{e' \preceq e} N^k \mu_{\langle e' \rangle} \Big( x_{|e'} - \widetilde{\bfz}(e') \Big)}\label{eq:product_at_most_e}\\
&= \Exu{x} {\overline{\chi_{b}(x)}\prod_{e' \prec e} N^k \mu_{\langle e' \rangle} \Big( x_{|e'} - \widetilde{\bfz}(e') \Big) \cdot 
\Exu{\widetilde{x}}{ N^k \mu_{\langle e \rangle} \Big( x_{|e} - \widetilde{\bfz}(e) \Big) \,\middle|\,x_{v^*},\left(x_{|e'} \right)_{e' \prec e}}}. \nonumber
\end{align}
Note that by our choice of the acyclic ordering $\prec$, we know that conditioning on $x_{v^*}$ and $(\widetilde{x}_{\mid e'})_{e'\prec e}$ only fixes at most one coordinate of the coordinates in $e$. If no coordinate is fixed, it is easy to see that the inner conditional expectation evaluates to $1$. Otherwise, suppose the $i$-th vertex $v$ of $e$ is fixed to $x_{v}=c\in \ZmodN$ by the conditioning. In this case, the inner conditional expectation equals
\begin{align*}\Exu{x}{ N^k \mu_{\langle e \rangle} \Big( x_{|e} - \widetilde{\bfz}(e) \Big) \,\middle|\, x_{v^*},\left( x_{|e'} \right)_{e' \prec e}} = N\cdot \sum_{z\in \bZ_N^k,\, z_{i} = c} \mu_{\langle e\rangle}\left(z - \widetilde{\bfz}(e)\right) = N\cdot \frac{1}{N} =1 .
\end{align*}
due to the one-wise independence of $\mu_{\langle e \rangle}$. Therefore, the expectation in \eqref{eq:product_at_most_e} remains unchanged when the range of the product is changed from $e'\preceq e$ to $e'\prec e$. 

By applying this argument recursively --- removing the maximal element under $\prec$ at each step, we can derive that
\begin{equation}\label{eq:structure_part_remove_component}
\Exu{x\in \ZmodN^{\calV \times [n]}} { \overline{\chi_{b}(x)}\prod_{e\in E} N^k \mu_{\langle e\rangle} \Big( x_{|e} - \widetilde{\bfz}(e) \Big)}=\Exu{x\in \ZmodN^{\calV \times [n]}} {\overline{\chi_{b}(x)}\prod_{e\in E\setminus E_{1}} N^k \mu_{\langle e \rangle} \Big( x_{|e} - \widetilde{\bfz}(e) \Big)}.
\end{equation}
Since $v^*\in \supp(b)$, when $x$ is a uniformly random vector in $\ZmodN^{\calV\times [n]}$, the restriction of $x$ to $\bigcup_{i=2}^{t}V_{i}$ is independent with $\chi_{b}(x)$. Similarly to the argument in the second paragraph of this proof, we conclude that the right-hand side of \eqref{eq:structure_part_remove_component} equals 0. Since the left-hand side of \eqref{eq:structure_part_remove_component} equals $\langle g_{\bdzeta},\chi_{b}\rangle$, it follows that $\langle g_{\bdzeta},\chi_{b}\rangle=0$.
\end{proof}

To analyze the Fourier growth of $g_{\bdzeta}$, we are interested in knowing how many level-$\ell$ Fourier coefficients of $g_{\bdzeta}$ can be nonzero, for a given integer $\ell$. This task is reduced by the second statement of Lemma~\ref{lem:structure_part_mean} to the following purely set-theoretic question: how many $\ell$-element subsets $S\subseteq \bigcup_{i\in [t]}V_{i}$ satisfy $|S\cap V_{i}|\neq 1$ for all $i$, where $V_{1},\dots,V_{t}$ are pairwise disjoint finite sets? The next lemma provides an answer to this question.

\begin{lemma}\label{lem:no_size_1_by_induction}
Suppose $V_{1},\dots,V_{t}$ are pairwise disjoint finite sets such that $|V_{i}|\geq 2$ for any $i\in [t]$. Then for any integer $\ell\geq 2$, the number of subsets $S\subseteq \bigcup_{i\in [t]}V_{i}$ such that $|S|=\ell$ and $|S\cap V_{i}|\neq 1$ for any $i$ is at most
\[
\left(\frac{20}{\ell}\big(|V_{1}|^{2}+\dots+|V_{t}|^{2}\big)\right)^{\ell/2}.
\]
\end{lemma}
\begin{proof}
In the polynomial ring $\bR[x_{1},\dots,x_{t}]$, a monomial is said to be singleton-free if the degree of any variable is not 1. For a polynomial $f\in \bR[x_{1},\dots,x_{t}]$, let $\mathsf{SF}(f)$ denote the sum of the coefficients of all singleton-free monomials in the expansion of $f$. We will show that
\begin{equation}\label{eq:generating_polynomial}
\mathsf{SF}\big((a_{1}x_{1}+\dots+a_{t}x_{t})^{\ell}\big)\leq \big(2\ell\cdot(a_{1}^{2}+\dots+a_{t}^{2})\big)^{\ell/2}\qquad\text{for any }a_{1},\dots,a_{t}\in \bR.
\end{equation}
Once we have \eqref{eq:generating_polynomial}, it is easy to see that the number of subsets $S\subseteq \bigcup_{i\in [t]}V_{i}$ such that $|S|=\ell$ and $|S\cap V_{i}|\neq 1$ for any $i$ is at most
\[
\frac{1}{\ell!}\cdot\mathsf{SF}\Big(\big(|V_{1}|x_{1}+\dots+|V_{t}|x_{t}\big)^{\ell}\Big)\leq \frac{\left(2\ell\cdot\big(|V_{1}|^{2}+\dots+|V_{t}|^{2}\big)\right)^{\ell/2}}{(\ell/3)^{\ell}}\leq \left(\frac{20}{\ell}\big(|V_{1}|^{2}+\dots+|V_{t}|^{2}\big)\right)^{\ell/2}.
\]

We prove \eqref{eq:generating_polynomial} by induction on $t$. For $t=1$ the inequality clearly holds, and we next assume $t\geq 2$ and that \eqref{eq:generating_polynomial} holds for smaller values of $t$. Without loss of generality assume $a_{1},\dots,a_{t}\geq 0$. If $a_{i}=0$ for all $i\geq 2$ then the desired inequality trivially holds, so we can assume the normalization $a_{2}^{2}+\dots+a_{t}^{2}=1$ (noting that \eqref{eq:generating_polynomial} is homogeneous in the $a_{i}$'s). We have
\begin{align*}
\mathsf{SF}\big((a_{1}x_{1}+\dots+a_{t}x_{t})^{\ell}\big)&= \mathsf{SF}\big((a_{2}x_{2}+\dots+a_{t}x_{t})^{\ell}\big)+\sum_{r=2}^{\ell}\binom{\ell}r{}a_{1}^{r}\cdot\mathsf{SF}\big((a_{2}x_{2}+\dots+a_{t}x_{t})^{\ell-r}\big)\\
&\leq (2\ell)^{\ell/2}+\sum_{r=2}^{\ell}\binom{\ell}r{}a_{1}^{r}\cdot(2\ell)^{(\ell-r)/2}\tag{by the induction hypothesis}\\
&=(2\ell)^{\ell/2}\left(\left(1+\frac{a_{1}}{\sqrt{2\ell}}\right)^{\ell}-\ell\cdot\frac{a_{1}}{\sqrt{2\ell}}\right)\\
&\leq \big(2\ell\cdot (1+a_{1}^{2})\big)^{\ell/2},\tag{using Proposition~\ref{prop:basic_calculus}}
\end{align*}
as desired.
\end{proof}

We are now ready to complete our analysis of the function $g_{\bdzeta}$ using the framework of Section~\ref{subsec:fourier_growth}.

\begin{lemma}\label{lem:structured_part}
    Fix a $\DIHP(G,n,\alpha,K)$ communication game, where $G=(\calV,\calE,N,(\mu_{\sfe})_{\sfe\in \calE})$. Let $\bdzeta$ be a restriction sequence that is not cyclic. If $\|\bdzeta\|\leq \gamma n$ for some constant $\gamma\in (0,1)$, then the function $g_{\bdzeta}$ (defined in Definition~\ref{def:g_bdzeta}) is $\big(n,20N^{2},\gamma n\log_{2} N,0\big)$-bounded. 
\end{lemma}

\begin{proof} It suffices to show that $g_{\bdzeta}$ satisfies the three conditions of $\big(n,20N^{2},\gamma n\log_{2}N,0\big)$-boundedness required by Definition~\ref{def:Cs_star_bounded}.

\paragraph{Expectation equals $1$.} It follows directly from Lemma~\ref{lem:structure_part_mean}(1) that $\Ex{g_{\bdzeta}}=1$.

\paragraph{Infinity-norm bound.}
Let  \(\bdzeta = \left(\bfz^{(\sfe, j)}\right)_{(\sfe, j) \in \calE \times [K]}\), and let $V_{1},\dots,V_{t}\subseteq \calV\times [n]$ be the list of nontrivial connected components of the hypergraph $H_{\bdzeta}$. For each $i\in [t]$, let $E_{i}$ be the set of edges of $H_{\bdzeta}$ within the component $V_{i}$. Since $H_{\bdzeta}$ does not contain cycles, for each $i\in [t]$ we have $(k-1)|E_{i}|<|V_{i}|$ and consequently $k|E_{i}|\leq |V_{i}|^{2}$. Therefore, we have
\begin{equation}\label{eq:g_infinity_norm}
    \sum_{(\sfe,j)\in\calE\times [K]} k\left|\supp(\bfz^{(\sfe,j)})\right| = \sum_{i=1}^{t} k|E_i|\leq \sum_{i=1}^t |V_i|^2 = \lVert \bdzeta\rVert, 
\end{equation}
where the first equality is due to the condition that the edge sets \(\left\{\supp(\bfz^{(\mathsf{e},j)})\right\}_{(\sfe,j)\in \calE\times [K]}\) are pairwise disjoint. Combining \eqref{eq:g_full_expansion}, \eqref{eq:g_infinity_norm} and the fact that $\lVert \mu_{\sfe}\rVert_{\infty} \leq 1$ for each $\sfe\in \calE$, we conclude that $\|g_{\bdzeta}\|_{\infty}\leq N^{\|\bdzeta\|}\leq N^{\gamma n}$.
    
\paragraph{Fourier growth bound.} For any $b\in \ZmodN^{\calV\times [n]}$, there are two possibilities:
\begin{enumerate}[label=(\arabic*)]
\item If $\supp(b)\not\subseteq \bigcup_{i\in [t]}V_{i}$ or $|\supp(b)\cap V_{i}|=1$ for some $i\in [t]$, then by Lemma~\ref{lem:structure_part_mean}(2) we have $\langle g_{\bdzeta},\chi_{b}\rangle=0$,
\item If $\supp(b)\subseteq  \bigcup_{i\in [t]}V_{i}$ and $|\supp(b)\cap V_{i}|\neq 1$ for all $i\in [t]$, then we have $\left|\langle g_{\bdzeta},\chi_{b}\rangle\right|\leq \|g_{\bdzeta}\|_{1}\cdot\|\chi_{b}\|_{\infty}\leq 1$, using Lemma~\ref{lem:structure_part_mean}(1).
\end{enumerate}
Therefore, for any integer $\ell\geq 1$, the Wiener norm $\left\|g_{\bdzeta}^{=\ell}\right\|_{\sfW}$ is upper bounded by the number of $b\in \ZmodN^{\calV\times [n]}$ such that $\supp(b)$ is an $\ell$-element subset of $\bigcup_{i\in [t]}V_{i}$ whose intersection with each $V_{i}$ has size not equal to 1. It then follows from Lemma~\ref{lem:no_size_1_by_induction} that
\[
\left\|g_{\bdzeta}^{=\ell}\right\|_{\sfW}\leq N^{\ell}\left(\frac{20}{\ell}\big(|V_{1}|^{2}+\dots+|V_{t}|^{2}\big)\right)^{\ell/2}\leq \left(\frac{20N^{2}\gamma n}{\ell}\right)^{\ell/2}\leq F_{20N^{2}}(n,\ell,\gamma^{2}n).
\]
Using Proposition~\ref{prop:Fourier_bound_non_decreasing}, we conclude that \(\left\|g_{\bdzeta}^{=\ell}\right\|_{\sfW}\leq F_{20N^{2}}\big(n,\ell,\gamma n\log_{2}N\big)\), as desired.
\end{proof}

\subsection{Analyzing the Pseudorandom Part by Induction}\label{subsec:discrepancy_bound}

As mentioned earlier, Lemmas~\ref{lem:relating_yes_no} and~\ref{lem:separating_structure_pseudorandom} together express the ratio $\calD_{\yes}(R)/\calD_{\no}(R)$ as the expectation of a product of $2K|\calE|$ functions. The product of the $K|\calE|$ structured-part functions has already been analyzed in Section~\ref{subsec:structured}, and in this subsection we will incorporate the remaining $K|\calE|$ pseudorandom-part functions into the product.

We will incorporate the $K|\calE|$ pseudorandom-part functions into the product \emph{one at a time}, using the following ``induction lemma'' (inspired by \cite[Lemma 6.1]{KK19}).
\begin{restatable}{lemma}{leminduction}\label{lem:induction}
Let $\calU$ be a $k$-universe embedded in a finite set $\Lambda$. Fix constants $C\in [1,+\infty)$ and $\delta_{0},\delta\in [0,\frac{1}{2}]$. Let $m\in \bZ$ and $s^*\in \bR$ be parameters satisfying
\[0< m\leq \frac{|\calU|^{3}}{2^{18}N^{10k}|\Lambda|^{2}}\quad\text{and}\quad \log_{2}(4/\delta)\leq s^*\leq (10C)^{-2}\delta^{4}m.\]
Suppose $f:\ZmodN^{\Lambda}\rightarrow[0,+\infty)$ is a $(|\calU|,C,s^*,\delta_{0})$-bounded function and $\calD$ is a pseudo-uniform distribution over $\calM_{\calU,m}$. For any parameter $\eta\in (0,\frac{1}{2})$, a random sample $M\sim\calD$ satisfies the following with probability at least $1-3\eta$: for any $A\subseteq \Map{M}{\ZNk}$ of size $|A|\geq 2^{-s^*}N^{km}$ and any one-wise independent distribution $\mu$ over $\ZNk$,
\[\text{the function }f\cdot \bdR^{\Lambda,M}_{\mu}[\phi_{A}]\text{ is }\Big(|\calU|,\,2^{22}\eta^{-2}N^{8k}C,\,2s^*,\,\delta_{0}+\eta^{-1}\delta\Big)\text{-bounded.}\]
\end{restatable}
Lemma~\ref{lem:induction} is proved in Section \ref{sec:induction}. In the rest of this section, we use the induction lemma to complete the proof of Lemma \ref{lem:discrepancy_bound} (restated below). 

\discrepancybound*

\begin{proof}[Proof of Lemma \ref{lem:discrepancy_bound} assuming Lemma~\ref{lem:induction}]
We begin by setting some notation for the proof. We write \(R=\prod_{(\sfe, j) \in \calE \times [K]} A^{(\sfe, j)}\) and \(\bdzeta = \left(\bfz^{(\sfe, j)}\right)_{(\sfe, j) \in \calE \times [K]}\). We fix an arbitrary total order on $\calE\times [K]$, and for each index $i\in \{1,2,\dots,K|\calE|\}$, if $(\sfe,j)\in \calE\times [K]$ is the $i$-th-ranked player, we use the following notations:

\begin{enumerate}[label=(\arabic*)]
\item Denote $n_{i}=n-|\supp(\bfz^{(\sfe,j)})|$ and $m_{i}=\alpha n-|\supp(\bfz^{(\sfe,j)})|$. Note that $m_{i}\leq \alpha n_{i}$. As long as $\gamma\leq \alpha /2$, we have $n_{i}\geq n-\|\bdzeta\|\geq n-\gamma n \geq n/2$ and $m_{i}\geq \alpha n-\gamma n\geq \alpha n/2$.
\item We use $\calM^{(\sfe,j)}$ as a shorthand for the matching space $\calM_{\calU_{\sfe},\,\alpha n,\,\bfz^{(\sfe,j)}}$. 
\item Let $\calD_{i}$ be the distribution of $\supp(\bfy)\setminus \supp(\bfz^{(\sfe,j)})$ where $\bfy$ is a uniformly random labeled matching in $A^{(\sfe,j)}$. Note that $\calD_{i}$ is a probability distribution over $\calM^{(\sfe,j)}$, and it is pseudo-uniform due to Proposition~\ref{prop:global_vs_pseudouniformity}.
\item For any matching $M\in \calM^{(\sfe,j)}$, let (see Definition~\ref{def:embedding})
\[A_{i,M}=A^{(\sfe,j)}_{M}:=\fraki^{-1}_{M,\bfz^{(\sfe,j)}}\big(A^{(\sfe,j)}\big)\subseteq \Map{M}{\ZNk}.\]
Note that the ratio $\left|A^{(\sfe,j)}_{M}\right|/\left|A^{(\sfe,j)}\right|$ is the probability of the matching $M$ under the distribution $\calD_{i}$. We also record the fact that
\begin{equation}\label{eq:expectation_of_A_M}
\Exu{M\sim \calD_{i}}{\frac{1}{|A_{i,M}|}}= \sum_{M\in \calM^{(\sfe,j)}}\frac{\left|A^{(\sfe,j)}_{M}\right|}{\left|A^{(\sfe,j)}\right|}\cdot\frac{1}{\left|A^{(\sfe,j)}_{M}\right|}=\frac{N^{-km_{i}}\left|\Omega^{\calU_{\sfe},\alpha n}_{\bfz^{(\sfe,j)}}\right|}{\left|A^{(\sfe,j)}\right|}\leq 2^{\gamma n}N^{-km_{i}},
\end{equation}
where we used the definition of $(\gamma n,\gamma n)$-goodness in the last transition.
\item For any matching $M\in \calM^{(\sfe,j)}$, let
\[
h_{i,M}=h^{(\sfe,j)}_{M}:=\bdR^{\calV\times [n],\,M}_{\mu_{\sfe}}\left[\phi_{A^{(\sfe,j)}_{M}}\right]\in L^{2}\big(\ZmodN^{\calV\times [n]}\big).
\]
\end{enumerate}

By Lemmas~\ref{lem:relating_yes_no} and~\ref{lem:separating_structure_pseudorandom}, we have
\begin{align}
\frac{\calD_{\yes}(R)}{\calD_{\no}(R)}&=\Exu{x\in \ZmodN^{\calV\times [n]}}{g_{\bdzeta}(x)\prod_{(\sfe,j)\in \calE\times [K]}\left(\sum_{M\in \calM^{(\sfe,j)}}\frac{\left|A^{(\sfe,j)}_{M}\right|}{\left|A^{(\sfe,j)}\right|}h^{(\sfe,j)}_{M}(x)\right)}\nonumber\\
&=\Exu{x\in \ZmodN^{\calV\times [n]}}{g_{\bdzeta}(x)\prod_{i=1}^{K|\calE|}\left(\Exu{{M\sim\calD_{i}}}{h_{i,M}(x)}\right)} \nonumber\\
&= \Exu{{(M_i)_{i=1}^{K|\calE|} \sim \prod_{i=1}^{K|\calE|}\calD_{i}}}{\Exu{x\in \ZmodN^{\calV\times [n]}}{g_{\bdzeta}(x)\prod_{i=1}^{K|\calE|}{h_{i,M_i}(x)}}}.\label{eq:relating_yesno}
\end{align}

\paragraph{Induction setup.} In the rest of the proof, we fix the following constants:
\begin{equation}\label{eq:choice_B_eta}
B=10^{16}N^{8k}K^{2}|\calE|^{2}\quad\text{and}\quad\eta=10^{-4}K^{-1}|\calE|^{-1}.
\end{equation}
We will prove that $\calD_{\yes}(R)/\calD_{\no}(R)\geq 1-10^{-3}$ if $\gamma$ is sufficiently small and $n$ is sufficiently large (after $\gamma$ is fixed). In other words, we taken $n\gg \gamma^{-1}\gg 1$.

For any tuple of sampled matchings $(M_{i})_{i=1}^{K|\calE|}$, where each $M_{i}$ is a possible sample from $\calD_{i}$, we say $(M_{i})_{i=1}^{K|\calE|}$ \emph{passes the $r$-th test} if 
\begin{equation}\label{eq:def_pass_test}
\text{the function}\qquad g_{\bdzeta}\prod_{i=1}^{r}h_{i,M_{i}}\qquad\text{is}\qquad\Big(n,\,20N^{2}B^{r},\,2^{r+1}\gamma n\log_{2}N,\,r\eta\Big)\text{-bounded,}
\end{equation}
for any given index $r\in\{0,1,2,\dots,K|\calE|\}$. By Lemma~\ref{lem:structured_part} we know that any tuple $(M_{i})_{i=1}^{K|\calE|}$ passes the 0-th test. If $(M_{i})_{i=1}^{K|\calE|}$ passes the $K|\calE|$-th test, then by definition we have
\[
\Ex{g_{\bdzeta}\prod_{i=1}^{K|\calE|}h_{i,M_{i}}}\geq 1-K|\calE|\eta.
\]
We claim that for each $r\in \{1,2,\dots,K|\calE|\}$, when each $M_{i}$ is sampled independently from $\calD_{i}$, the probability that $(M_{i})_{i=1}^{K|\calE|}$ passes the $r$-th test conditioned on it passes the $(r-1)$-th test is at least $1-4\eta$. Once we have that, it follows from union bound that $(M_{i})_{i=1}^{K|\calE|}$ passes the $K|\calE|$-th test with probability $1-4K|\calE|\eta$, and thus the expectation on the right-hand side of \eqref{eq:relating_yesno} is least $(1-4K|\calE|\eta)(1-K|\calE|\eta)\geq 1-5K|\calE|\eta\geq 1-10^{-3}$, as desired.

\paragraph{Induction step.} We next prove the induction-step claim using Lemma~\ref{lem:induction}. We fix an index $r\in\{1,2,\dots,K|\calE|\}$ and assume that $(M_{i})_{i=1}^{K|\calE|}$ passes the $(r-1)$-th test. By the definition \eqref{eq:def_pass_test} and Proposition~\ref{prop:boundedness_adjustment}, we have
\begin{equation}
\text{the function }g_{\bdzeta}\prod_{i=1}^{r-1}h_{i,M_{i}}\text{ is }\Big(n_{r},\,40N^{2}B^{r-1},\,2^{r}\gamma n\log_{2}N,\,(r-1)\eta\Big)\text{-bounded.}
\end{equation}
We apply Lemma~\ref{lem:induction} to the function $f=g_{\bdzeta}\prod_{i=1}^{r-1}h_{i,M_{i}}$ with the parameters
\begin{equation}\label{eq:choice_C_delta_s_star}
C=40N^{2}B^{r-1},\qquad \delta_{0}=(r-1)\eta,\qquad\delta=\eta^{2},\qquad s^*=2^{r}\gamma n\log_{2}N,
\end{equation}
and
\[
m=m_{r}\leq \alpha n_{r}\leq \frac{n_{r}}{2^{20}N^{10k}|\calV|^{2}}\leq \frac{n_{r}^{3}}{2^{18}N^{10k}|\calV|^{2}n^{2}}.
\]
The condition $\log_{2}(4/\delta)\leq s^*\leq (10C)^{-2}\delta^{4}m$ stated in Lemma~\ref{lem:induction} is satisfied as long as we take $n\gg \gamma^{-1}\gg 1$.  

Recalling that the probability of any matching $M$ under $\calD_{r}$ is proportional to $|A_{r,M}|$, we have
\begin{align*}
\Pru{M_{r}\sim \calD_{r}}{\left|A_{r,M_{r}}\right|\leq 2^{-s^*}N^{km_{r}}}\leq \frac{\Exs{M_{r}\sim \calD_{r}}{|A_{r,M_{r}}|^{-1}}}{2^{s^*}N^{-km_{r}}}\leq 2^{-s^*+\gamma n}\leq 2^{-\gamma n}\leq \eta,
\end{align*}
where we used \eqref{eq:expectation_of_A_M} in the second transition and $n\gg  \gamma^{-1}$ in the last transition. Therefore, when $M_{r}$ is sampled from $\calD_{r}$, with probability at least $1-4\eta$ both of the following statements hold:
\begin{enumerate}[label=(\arabic*)]
\item $|A_{r,M_{r}}|\geq  2^{-s^*}N^{km_{r}}$.
\item The conclusion of Lemma~\ref{lem:induction} holds for $M_{r}$, i.e. for any $A\subseteq \Map{M_{r}}{\ZNk}$ of size $|A|\geq 2^{-s^*}N^{km_{r}}$ and any one-wise independent distribution $\mu$, the function $f\cdot \bdR^{\calV\times [n],\,M_{r}}_{\mu}[\phi_{A}]$ is
\[\Big(n_{r},\,(2^{22}\eta^{-2}N^{8k}C)\cdot40N^{2}B^{r-1},\,2s^*,(r-1)\eta+\eta^{-1}\delta\Big)\text{-bounded.}
\]
By \eqref{eq:choice_B_eta}, \eqref{eq:choice_C_delta_s_star} and Proposition~\ref{prop:boundedness_adjustment}, it is also
\[
\Big(n,\,20N^{2}B^{r},\,2^{r+1}\gamma n\log_{2}N,\,r\eta\Big)\text{-bounded}.
\]
\end{enumerate}
Combining the two statements above, we conclude that with probability at least $1-4\eta$ over the sampled matching $M_{r}\sim \calD_{r}$, we have
\[
f\cdot \bdR^{\calV\times [n],\,M_{r}}_{\mu}[\phi_{A}]=g_{\bdzeta}\prod_{i=1}^{r}h_{i,M_{i}}\qquad\text{is}\qquad\Big(n,\,20N^{2}B^{r},\,2^{r+1}\gamma n\log_{2}N,\,r\eta\Big)\text{-bounded}.
\]
In other words, $(M_{i})_{i=1}^{K|\calE|}$ passes the $r$-th test with conditional probability at least $1-4\eta$.
\end{proof}

%% file: induction_lemma.tex
\section{Proof of the Induction Lemma}\label{sec:induction}

The goal of this section is to prove Lemma~\ref{lem:induction}. The lemma features several combinatorial/analytic objects of study, each requiring separate treatment. We begin in Section~\ref{subsec:svd} by studying the Fourier-analytic properties of the operator $\bdR^{\Lambda,M}_{\mu}$. In Section~\ref{subsec:random_matchings}, we examine the pseudo-uniform distribution of matchings and establish its relevant combinatorial properties. We then address the technical core of the lemma in Sections~\ref{subsec:mass_transfer} and~\ref{subsec:induction_step}, where we analyze the Fourier growth of the product $f\cdot \bdR^{\Lambda,M}_{\mu}[\phi_{A}]$ via the convolution theorem.

\subsection{Singular Value Decomposition}\label{subsec:svd}

A key property of the operator is that it admits a clean singular value decomposition: it maps Fourier characters on the space $\Map{M}{\ZNk}$ to scalar multiples of character functions on $\ZmodN^{\Lambda}$. We first define the natural character functions on $\Map{M}{\ZNk}$.

\begin{definition}
For a finite set $M$ and a map $\bfa:M\rightarrow \ZNk$, we define a character function $\psi_{\bfa}:\Map{M}{\ZNk}\rightarrow \bC$ by
\[
\psi_{\bfa}(\bdxi):=\prod_{e\in M}\chi_{\bfa(e)}(\bdxi(e)).
\]
For any $f\in L^{2}\big(\Map{M}{\ZNk}\big)$, let $\widehat{f}(\bfa):=\langle f,\psi_{\bfa}\rangle$. 
\end{definition}

Due to the one-wise independence of the distribution $\mu$, some characters on the space $\Map{M}{\ZNk}$ are ``killed'' by the operator $\bdR^{\Lambda,M}_{\mu}$ (and we must use that to our advantage in order to prove Lemma~\ref{lem:induction}). The following definition collects the characters that are \emph{not} killed.

\begin{definition}
For a finite set $M$, we let $\calX(M)$ be the collection of character indices $\bfa:M\rightarrow\ZNk$ such that $\|\bfa(e)\|_{\sfH}\neq 1$ for any $e\in M$.
\end{definition}

\begin{remark}
Note that the Hamming weight of $\bfa(e)$, denoted by $\|\bfa(e)\|_{\sfH}$, is the number of $i\in [k]$ such that $\bfa(e)_{i}\neq 0$. This is to be distinguished from the Hamming weight of $\bfa$ itself, denoted by $\|\bfa\|_{\sfH}$, which stands for the number of $e\in M$ such that $\bfa(e)\neq 0$. See Section~\ref{subsec:general_notations} for the general notational convention of Hamming weights.
\end{remark}

The following notations will be used many times throughout Section~\ref{sec:induction}.

\begin{notation}
Suppose $S$ and $T$ are disjoint finite sets. For two maps $\bfa_{1}:S\rightarrow \ZNk$ and $\bfa_{2}:T\rightarrow\ZNk$, we define their concatenation $\bfa_{1}\uplus \bfa_{2}:S\sqcup T\rightarrow \ZNk$ by setting $(\bfa_{1}\uplus \bfa_{2})(e):=\bfa_{1}(e)$ for $e\in S$ and $(\bfa_{1}\uplus \bfa_{2})(e):=\bfa_{2}(e)$ for $e\in T$.
\end{notation}

\begin{notation}
Suppose $S$ and $M$ are finite sets such that $S\subseteq M$. For a map $\bfa:M\rightarrow\ZNk$, we define $\bfa_{|S}:S\rightarrow \ZNk$ to be the restriction of $\bfa$ to $S$, and define $\bfa_{\setminus S}:M\setminus S\rightarrow\ZNk$ to be the restriction of $\bfa$ to $M\setminus S$. 
\end{notation}

\begin{definition}
Let $\calU$ be a $k$-universe embedded in a finite set $\Lambda$, and let $M\in \calM_{\calU,m}$ be a matching of size $m\leq |\calU|$. For any character index $\bfa:M\rightarrow \ZNk$, we use $[\bfa]$ to denote the corresponding character index in $\ZmodN^{\Lambda}$ defined by
\begin{enumerate}
\item $[\bfa]_{v_{i}}=\bfa(e)_{i}$ for any edge $e=(v_{1},\dots,v_{k})\in M$,
\item $[\bfa]_{v}=0$ for any vertex $v\in \Lambda$ that is not contained in any edge of $M$.
\end{enumerate}
\end{definition}

We are now ready to present the singular value decomposition lemma. 

\begin{lemma}\label{lem:SVD}
Let $\calU$ be a $k$-universe d in a finite set $\Lambda$, and let $\mu$ be a one-wise independent distribution over $\ZNk$. Fix a matching $M\in \calM_{\calU,m}$ of size $m\leq |\calU|$. For any $f\in L^{2}\big(\Map{M}{\ZNk}\big)$ and any $b\in \ZmodN^{\Lambda}$, we have
\[
\begin{cases}\left|\left\langle \bdR^{\Lambda,M}_{\mu}[f], \chi_b \right\rangle_{L^2(\ZmodN^{\Lambda})}\right|\leq \left|\widehat{f}(\bfa)\right|,&\text{if }b=[\bfa]\text{ for some }\bfa\in \calX(M),\\
\left\langle \bdR^{\Lambda,M}_{\mu}[f], \chi_b \right\rangle_{L^2(\ZmodN^{\Lambda})}=0,&\text{if }b\neq [\bfa]\text{ for any }\bfa\in \calX(M).
\end{cases}
\]
\end{lemma}
\begin{proof}
Recall from Section~\ref{subsec:general_notations} that $\mu(\cdot)$ denotes the probability mass function of $\mu$. For each character index \(t \in \mathbb{Z}_N^k\), define
\[
r(t) := \sum_{z \in \mathbb{Z}_N^k} \mu(z)\, \overline{\chi_t(z)} = N^k \cdot \widehat{\mu}(t).
\]
Since \(\mu\) is assumed to be one-wise independent, we know that \(\widehat{\mu}(t) = 0\) for any \(t \in \mathbb{Z}_N^k\) with exactly one nonzero coordinate. Thus, \(r(t) = 0\) for such \(t\). Additionally, we have \(r(0) = 1\) and \(|r(t)| \leq 1\) for all \(t\) since $|r(t)|\leq  \mathbb{E}_{z \sim \mu}\, \left|\overline{\chi_t(z)}\right|=1$.

By Definition~\ref{def:Markov_R} we have
\begin{equation}\label{eq:SVD_first_step}
\left\langle \bdR^{\Lambda,M}_{\mu}[f], \chi_b \right\rangle_{L^2(\ZmodN^{\Lambda})}=\sum_{\bdxi:M\rightarrow \ZNk}f(\bdxi)\left\langle\bfR^{\calU,m}_{\mu}(\cdot,\bdxi),\chi_{b}\right\rangle_{L^{2}(\mathbb{Z}_N^{\Lambda})}.
\end{equation}
If $b\neq [\bfa]$ for any $\bfa\in\Map{M}{\ZNk}$, from Definition~\ref{def:Markov_R} it is easy to see that 
\[\left\langle\bfR^{\calU,m}_{\mu}(\cdot,\bdxi),\chi_{b}\right\rangle_{L^{2}(\mathbb{Z}_N^{\Lambda})}=0\text{ for any }\bdxi\in \Map{M}{\ZNk},\] in which case \eqref{eq:SVD_first_step} evaluates to 0, as desired. We next assume $b=[\bfa]$ for some $\bfa\in\Map{M}{\ZNk}$. We can then calculate from \eqref{eq:SVD_first_step} and Definition~\ref{def:Markov_R} as follows:
\begin{align*}
\left\langle \bdR^{\Lambda,M}_{\mu}[f], \chi_{[\bfa]} \right\rangle_{L^2(\ZmodN^{\Lambda})}&=\sum_{\bdxi:M\rightarrow\ZNk}f(\bdxi)\cdot\Exu{x\in \ZmodN^{\Lambda}}{\overline{\chi_{[\bfa]}(x)}\cdot\prod_{e\in M}\mu\big(x_{|e}-\bdxi(e)\big)}\\
&=\sum_{\bdxi:M\rightarrow\ZNk}\left(f(\bdxi)\prod_{e\in M}\Exu{z\in \ZNk}{\mu\big(z-\bdxi(e)\big)\,\overline{\chi_{\bfa(e)}(z)}}\right)\\
&=\sum_{\bdxi:M\rightarrow\ZNk}\left(f(\bdxi)\prod_{e\in M}\left(\widehat{\mu}(\bfa(e))\,\overline{\chi_{\bfa(e)}(\bdxi(e))}\right)\right)\\
&=N^{-k|M|}\prod_{e\in M}r(\bfa(e))\cdot\sum_{\bdxi:M\rightarrow\ZNk}\left(f(\bdxi)\prod_{e\in M}\overline{\chi_{\bfa(e)}(\bdxi(e))}\right)\\
&=\prod_{e\in M}r(\bfa(e))\cdot\langle f,\psi_{\bfa}\rangle=\widehat{f}(\bfa)\prod_{e\in M}r(\bfa(e)).
\end{align*}
Recall that for a fixed $e\in M$, the complex number $r(\bfa(e))$ always has absolute value at most 1, and equals 0 if $\bfa(e)$ has exactly one nonzero coordinate. Therefore, $\prod_{e\in M}r(\bfa(e))$ has absolute value at most 1, and vanishes if $\bfa\not\in\calX(M)$. The desired conclusion thus follows from the calculation above.
\end{proof}

\subsection{Random Matchings}\label{subsec:random_matchings}

Throughout this subsection, we fix a $k$-universe $\calU$ embedded in a finite set $\Lambda$, and a nonnegative integer $m\leq |\calU|$. Our goal is to study the behavior  of a pseudo-uniformly random matching $M\in \calM_{\calU,m}$ with respect to a fixed vector $z\in \ZmodN^{\Lambda}$, in order to prepare for the convolution analysis in Section~\ref{subsec:mass_transfer}.

As will become clear in Section~\ref{subsec:mass_transfer}, given an edge $e\in M$ that intersects with the vertex set $\supp(z)$, it is important to distinguish between the case where $e$ has exactly one common vertex with $\supp(z)$ and the case where they intersect in more than one vertex. We will classify all vertices in $\supp(z)$ that are touched by the matching $M$ into the following two types.

\begin{definition}\label{def:internal_boundary_vertex}
Let $e = (v_{1},\dots,v_{k}) \in \tprod \mathcal{U}$ be an edge, and let 
$z \in \ZmodN^{\Lambda}$ be a vector. For any $i \in [k]$ with $z_{v_i} \neq 0$, the vertex $v_{i}$ is called a \emph{boundary vertex of $z$ to $e$} if $z_{v_j} = 0$ for all $j \in [k]\setminus\{i\}$; it is called \emph{an internal vertex of $z$ to $e$} if $z_{v_j} \neq  0$ for some $j \in [k]\setminus\{i\}$.
\end{definition}

\begin{definition}\label{def:internal_boundary_set}
For any matching $M\in \calM_{\calU,m}$ and any vector $z\in \ZmodN^{\Lambda}$, we define two vertex sets as follows:
\begin{align*}
\bd_{M,z}&:=\{v\in \supp(z)\mid v\text{ is boundary to some }e\in M\},\\
\itn_{M,z}&:=\{v\in \supp(z)\mid v\text{ is internal to some }e\in M\}.
\end{align*}
\end{definition}
For fixed $M,z$, the two sets $\bd_{M,z}$ and $\itn_{M,z}$ are always disjoint since $M$ is a matching.  
Intuitively, when $m$ is a small fraction of $|\calU|$ and $M$ is a pseudo-uniformly random matching in $\calM_{\calU,m}$, the subsets $\bd_{M,z}$ and $\itn_{M,z}$ are typically not too large relative to $\supp(z)$. To formalize this idea, we make the following definition.

\begin{definition}\label{def:q_for_pseudo_uniform}
For any nonnegative integers $t,i,b$ such that $i+b\leq t$, we define a quantity $q^{\calU,m}_{\Lambda}(t,i,b)\in [0,1]$ by
\[
q^{\calU,m}_{\Lambda}(t,i,b):=\max_{z,\calD}\Pru{M\sim\calD}{|\itn_{M,z}|=i\text{ and }|\bd_{M,z}|=b\big.},
\]
where the maximum is taken over all vectors $z\in \ZmodN^{\Lambda}$ of Hamming weight $\|z\|_{\sfH}=t$ and all pseudo-uniform distributions $\calD$ over $\calM_{\calU,m}$.
\end{definition}

We then prove an upper bound on the quantity $q_{\Lambda}^{\calU,m}(t,i,b)$.

\begin{lemma}\label{lem:q_calculation}
For any nonnegative integers $t,i,b$ such that $i+b\leq t$, we have
\[
q^{\calU,m}_{\Lambda}(t,i,b)\leq (24k^{2})^{t}\left(\frac{im}{|\calU|^{2}}\right)^{i/2}\left(\frac{m}{|\calU|}\right)^{b}.
\]
\end{lemma}
\begin{proof} The definition of $q^{\calU,m}_{\Lambda}(t,i,b)$ features a pseudo-uniform distribution $\calD$ over $\calM_{\calU,m}$. In the first part of the proof, we reduce the problem to the case where $\calD$ is the uniform distribution over $\calM_{\calU,m}$; the uniform case will then be handled in the second part of the proof.

\paragraph{Reducing to the uniform case.} For a fixed vector $z\in\ZmodN^{\Lambda}$ such that $\|z\|_{\sfH}=t$, we let $\calS_{z,i,b}$ be the collection of partial matchings $S\in \calM_{\calU,\leq t}$ such that $|\itn_{S,z}|=i$, $|\bd_{S,z}|=b$ and every edge of $S$ contains at least one vertex in $\supp(z)$. We clearly have
\begin{equation}\label{eq:pseudo-uniform-i-b}
\Pru{M\sim\calD}{|\itn_{M,z}|=i\text{ and }|\bd_{M,z}|=b\big.}\leq \sum_{S\in \calS_{z,i,b}}\Pru{M\sim\calD}{S\subseteq M}.
\end{equation}
For each $S\in \calS_{z,i,b}$, by Definition~\ref{def:pseudo_uniform} we have 
\begin{equation}\label{eq:reduce_to_uniform-i-b}
\Pru{M\sim\calD}{S\subseteq M}\leq 2^{|S|}\Pru{M\in\calM_{\calU,m}}{S\subseteq M}\leq 2^{t}\Pru{M\in\calM_{\calU,m}}{S\subseteq M}.
\end{equation}
For any $M\in\calM_{\calU,m}$, let $\mathsf{cnt}(M,z,i,b)$ be the number of partial matchings $S\in \calS_{z,i,b}$ that are contained in $M$. Combining \eqref{eq:pseudo-uniform-i-b} and \eqref{eq:reduce_to_uniform-i-b}, we have
\begin{equation}\label{eq:expected_count_of_submatchings}
\Pru{M\sim\calD}{|\itn_{M,z}|=i\text{ and }|\bd_{M,z}|=b\big.}\leq 2^{t}\sum_{S\in \calS_{z,i,b}}\Pru{M\in\calM_{\calU,m}}{S\subseteq M}=2^{t}\Exu{M\in \calM_{\calU,m}}{\mathsf{cnt}(M,z,i,b)}.
\end{equation}
It then suffices to give an upper bound for the expectation of $\mathsf{cnt}(M,z,i,b)$ over a uniformly random $M\in\calM_{\calU,m}$.

\paragraph{Analysis in the uniform case.} We fix a vector $z\in \ZmodN^{\Lambda}$ with $|\supp(z)\cap\tcup\calU|=r$, and proceed to estimate the expected number of partial matchings in $\calS_{z,i,b}$ that are contained in a uniformly random matching $M\in\calM_{\calU,m}$. Note that this number $\mathsf{cnt}(M,z,i,b)$ only depends on $\supp(z)\cap\tcup\calU$, rather than on the full vector $z$. 

Suppose $|\supp(z)\cap U_{j}|=r_{j}$, for $j\in [k]$. We may then equivalently think of $M$ as a fixed matching in $\calM_{\calU,m}$, and $R_{z}:=\supp(z)\cap\bigcup\calU$ as a random subset of $\tcup \calU$ conditioned on $|R_{z}\cap U_{j}|=r_{j}$ for all $j\in [k]$. The number of such subsets $R_{z}$ is
\[
\prod_{j=1}^{k}\binom{|\calU|}{r_{j}}\geq \prod_{j=1}^{k}\left(\frac{|\calU|}{r_{j}}\right)^{r_{j}}\geq \left(\frac{|\calU|}{r}\right)^{r}.
\]
A valid choice of a pair $(R_{z},S)$ such that $S\subseteq M$ and $S\in \calS_{z,i,b}$ must comply with the following three steps: 

\begin{enumerate}[label=(\arabic*)]
\item First choose a set $\bd_{S,z}$ with size $b$. The number of such choices is $\binom{m}{b}k^{b}$.
\item Next choose a set $\itn_{S,z}$ with size $i$. Note that since every edge of $M$ that intersects with $\itn_{S,z}$ must contain at least 2 vertices from $\itn_{S,z}$, the number of such edges is at most $i/2$. Therefore, (after the first step) the number of choices for $\itn_{S,z}$ is at most $\binom{m-b}{d}\binom{kd }{i}$, where $d:=\min\{\lfloor i/2\rfloor ,m-b\}$.
\item The first two steps completely determine the partial matching $S$. The third step is to choose the set $R_{z}\setminus(\bd_{S,z}\cup \itn_{S,z})$. The number of such choices is at most $\binom{k|\calU|}{r-i-b}$.
\end{enumerate}
We can thus conclude that for a fixed $M\in \calM_{\calU,m}$ and a random $z\in \ZmodN^{\Lambda}$ with $|\supp(z)\cap U_{j}|=r_{j}$ for each $j\in [k]$, the expected number of $S\in \calS_{z,i,b}$ such that $S\subseteq M$ is at most
\begin{equation}\label{eq:q_calculation}
\binom{m}{b}k^{b}\binom{m-b}{d}\binom{kd}{i}\binom{k|\calU|}{r-i-b}\cdot\left(\frac{|\calU|}{r}\right)^{-r},
\end{equation}
where $r:=\sum_{j=1}^{k}r_{j}\geq i+b$ (the probability is 0 if $r<i+b$).

To simplify this bound, we first note that if $i/2\leq m-b$, then $d=\lfloor i/2\rfloor$ and
\[
\binom{m-b}{d}\binom{kd}{i}\leq \left(\frac{3(m-b)}{i/2}\right)^{i/2}\left(\frac{3ki/2}{i}\right)^{i}\leq \left(\frac{15k^{2}(m-b)}{i}\right)^{i/2}.
\]
If $i/2>m-b$, then $d=m-b$ and
\[
\binom{m-b}{d}\binom{kd}{i}\leq \left(\frac{3k(m-b)}{i}\right)^{i}\leq \left(\frac{3k}{2}\right)^{i}.
\]
Since it is always true that $i\leq k(m-b)$,\footnote{Otherwise $q^{\calU,m}_{\Lambda}(t,i,b)=0$.} in both cases we have
\[
\binom{m-b}{d}\binom{kd}{i}\leq\left(\frac{15 k^{3}(m-b)}{i}\right)^{i/2}.
\]
Plugging this into \eqref{eq:q_calculation} and simplifying the other terms in \eqref{eq:q_calculation}, we get
\begin{align*}
\max_{\substack{z\in \ZmodN^{\Lambda}\\ \|z\|_{\sfH}=t}}\Exu{M\in \calM_{\calU,m}}{\mathsf{cnt}(M,z,i,b)}&\leq \max_{i+b\leq r\leq t}\left(\frac{3m}{b}\right)^{b}k^{b}\left(\frac{15k^{3}(m-b)}{i}\right)^{i/2}\left(\frac{3k|\calU|}{r-i-b}\right)^{r-i-b}\left(\frac{r}{|\calU|}\right)^{r}\\
&\leq \max_{i+b\leq r\leq t}(12k^{2})^{r}\left(\frac{im}{|\calU|^{2}}\right)^{i/2}\left(\frac{m}{|\calU|}\right)^{b}=(12k^{2})^{t}\left(\frac{im}{|\calU|^{2}}\right)^{i/2}\left(\frac{m}{|\calU|}\right)^{b}, 
\end{align*}
where in the second transition we used the convexity inequality $b^{b}i^{i}(r-i-b)^{r-i-b}\geq (r/3)^{r}$. 

Combining the above with \eqref{eq:expected_count_of_submatchings}, for any pseudo-uniform distribution $\calD$ over $\calM_{\calU,m}$ we have
\[
\max_{z\in \ZmodN^{\Lambda},\; \|z\|_{\sfH}=t}\Pru{M\sim\calD}{|\itn_{M,z}|=i\text{ and }|\bd_{M,z}|=b\big.}\leq (24k^{2})^{t}\left(\frac{im}{|\calU|^{2}}\right)^{i/2}\left(\frac{m}{|\calU|}\right)^{b},
\]
as desired.
\end{proof}

\subsection{Transfer of Fourier Mass}\label{subsec:mass_transfer}

Throughout this subsection, we fix a $k$-universe $\calU$ embedded in a finite set $\Lambda$, and a nonnegative integer $m\leq |\calU|$. The goal of this subsection is to analyze the following quantity.

\begin{definition}\label{def:transfer_Q}
Suppose $s^*$ is a real number with $0\leq s^*\leq |\calU|$. For any nonnegative integers $t$ and $\ell$, we define a quantity $Q^{\calU,m}_{\Lambda}(s^*,t,\ell)$ by
\begin{equation}\label{eq:def_of_Q}
Q^{\calU,m}_{\Lambda}(s^*,t,\ell):=\max_{z,\calD}\Exu{M\sim\calD}{\max_{A}\sum_{\bfa\in \calX(M),\;\|z+[\bfa]\|_{\sfH}=\ell} \left|\widehat{\phi_{A}}(\bfa)\right|},
\end{equation}
where the first maximum is taken over all vectors $z\in \ZmodN^{\Lambda}$ of Hamming weight $\|z\|_{\sfH}=t$ and all pseudo-uniform distributions $\calD$ over $\calM_{\calU,m}$, while the second maximum taken is over all sets $A\subseteq \Map{M}{\ZNk}$ of size $|A|\geq 2^{-s^*}N^{km}$. 
\end{definition}

Intuitively, the quantity $Q^{\calU,m}_{\Lambda}(s^*,t,\ell)$ measures the ``rate of transfer'' of Fourier mass from level $t$ to level $\ell$ when a function is multiplied by $\bdR^{\Lambda,M}_{\mu}[\phi_{A}]$. In Sections~\ref{subsubsec:mass_transfer_to_0_weight}, \ref{subsubsec:mass_transfer_to_low_weight}, and~\ref{subsubsec:mass_transfer_to_intermediate_weight}, we derive upper bounds on $Q^{\calU,m}_{\Lambda}(s^*,t,\ell)$ for different ranges of $\ell$.

The reason for treating these ranges separately is that, in the ``low-level'' regime where $1 \le \ell \le s^*$, we can invoke the level-$d$ inequality (Proposition~\ref{prop:level_d_classical}) to obtain sharper bounds; see Section~\ref{subsubsec:apply_hypercontract}.

Before starting the concrete calculations, we record the following simple observation that handles the case $t>|\calU|$.

\begin{proposition}\label{prop:Q_tilde}
Suppose $s^*$ is a real number with $0\leq s^*\leq |\calU|$. For any nonnegative integers $t$ and $\ell$, we have
\[
Q^{\calU,m}_{\Lambda}(s^*,t,\ell)\leq \max_{r}Q^{\calU,m}_{\Lambda}(s^*,r,\ell-t+r),
\]
where the maximum is taken over all nonnegative integers $r$ such that $(t-\ell)^{+}\leq r\leq \min\{t,k|\calU|\}$. 
\end{proposition}
\begin{proof}
Notice that for any $z\in \ZmodN^{\Lambda}$, if we define $z'\in\ZmodN^{\Lambda}$ by
\[
z'_{v}=\begin{cases}
z_{v},&\text{if }v\in\tcup\calU,\\
0,&\text{if }v\not\in\tcup\calU
\end{cases}\qquad\text{for any }v\in\Lambda,
\]
then for any $\bfa\in\calX(M)$ we have $\|z+[\bfa]\|_{\sfH}=\ell$ if and only if $\|z'+[\bfa]\|_{\sfH}=\ell-(\|z\|_{\sfH}-\|z'\|_{\sfH})$. Therefore, in the defining equation \eqref{eq:def_of_Q} of $Q^{\calU,m}_{\Lambda}(s^*,t,\ell)$, we may add the requirement $\supp(z)\subseteq \tcup\calU$ to the first maximum and replace the condition $\|z+[\bfa]\|_{\sfH}=\ell$ in the sum with the condition $\|z+[\bfa]\|_{\sfH}=\ell-t+\|z\|_{\sfH}$, without changing the value of the right-hand side. The desired conclusion then immediately follows.
\end{proof}

\subsubsection{Mass Transfer to the Zero-Weight Coefficient}\label{subsubsec:mass_transfer_to_0_weight}

The case where $\ell=0$ can be handled relatively easily, as the following lemma shows.

\begin{lemma}\label{lem:mass_transfer_to_0_weight}
Suppose $s^*$ is a real number with $0\leq s^*\leq |\calU|$. For any nonnegative integer $t$, we have
\[
Q^{\calU,m}_{\Lambda}(s^*,t,0)\leq q^{\calU,m}_{\Lambda}(t,t,0).
\]
\end{lemma}
\begin{proof}
Note that for any vector $z\in \ZmodN^{\Lambda}$ and any matching $M\in \calM_{\calU,m}$, there is at most one character index $\bfa\in \calX(M)$ such that $z+[\bfa]=0$. Furthermore, by the definition of $\calX(M)$, there exists such an $\bfa\in \calX(M)$ only if every vertex in $\supp(z)$ is internal to some edge of $M$ (as defined in Definition~\ref{def:internal_boundary_vertex}). Equivalently, this happens only if $|\itn_{M,z}|=t$ and $|\bd_{M,z}|=0$ (as defined in Definition~\ref{def:internal_boundary_set}). Therefore, by Definition~\ref{def:transfer_Q} we have
\[
Q^{\calU,m}_{\Lambda}(s^*,t,0)\leq \max_{z,\calD}\Exu{M\sim \calD}{\ind{|\itn_{M,z}|=t\text{ and }|\bd_{M,z}|=0\big.}},
\]
where the maximum is taken over all vectors $z\in \ZmodN^{\Lambda}$ of Hamming weight $\|z\|_{\sfH}=t$ and all pseudo-uniform distributions $\calD$ over $\calM_{\calU,m}$. It then follows from Definition~\ref{def:q_for_pseudo_uniform} that $Q^{\calU,m}_{\Lambda}(s^*,t,0)\leq q^{\calU,m}_{\Lambda}(t,t,0)$.
\end{proof}
\subsubsection{Applying Hypercontractivity}\label{subsubsec:apply_hypercontract}

To handle the cases where $\ell\geq 1$, we prepare the following lemma using the level-$d$ inequality on the product space $\Map{M}{\ZNk}$.

\begin{lemma}\label{lem:induction_level_d}
    Fix a matching $M\in \calM_{\calU,m}$, a vector $z\in \ZmodN^{\Lambda}$ of Hamming weight $\|z\|_{\sfH}=t$, and a set $A \subseteq \Map{M}{\ZNk}$ of size $|A|\geq 2^{-s^*}N^{km}$. For any nonnegative integer $\ell$, we have
        \begin{equation}\label{eq:induction_level_d_1}
             \sum_{\bfa\in \calX(M),\;\|z+[\bfa]\|_{\sfH}=\ell} \left|\widehat{\phi_{A}}(\bfa)\right|\leq 2^{s^*/2}N^{2kt}  \cdot\left(\frac{6m}{\ell-t+|\mathsf{in}_{M,z}|}\right)^{(\ell-t+|\mathsf{in}_{M,z}|)/4}.
        \end{equation}
        If $\ell\leq s^*$, then we also have
        \begin{equation}\label{eq:induction_level_d_2}
            \sum_{\bfa\in \calX(M),\;\|z+[\bfa]\|_{\sfH}=\ell} \left|\widehat{\phi_{A}}(\bfa)\right|\leq N^{2kt} \cdot \left(\frac{8\sqrt{m s^*}}{\ell -t+|\mathsf{in}_{M,z}|}\right)^{(\ell -t+|\mathsf{in}_{M,z}|)/2}.
        \end{equation}
\end{lemma}
\begin{proof}    
    Let $M'\subseteq M$ be the set of edges of $M$ that contains at least one vertex in $\supp(z)$. For any $\bfa'\in \calX(M')$, we have
    \begin{equation}\label{eq:pre_level_d}
        \sum_{\substack{\bfa\in \calX(M),\;\|z+[\bfa]\|_{\sfH}=\ell,\\
        \bfa_{|M'}=\bfa'}}\left|\widehat{\phi_A}(\bfa)\right| 
        \leq \sqrt{\sum_{\substack{\bfa\in \calX(M),\;\|z+[\bfa]\|_{\sfH}=\ell,\\
        \bfa_{|M'}=\bfa'}} 1}\cdot \sqrt{\sum_{\substack{\bfa\in \calX(M),\;\|z+[\bfa]\|_{\sfH}=\ell,\\
        \bfa_{|M'}=\bfa'}} \left|\widehat{\phi_A}(\bfa)\right|^2}.
    \end{equation}

    Note that for any $\bfa\in \calX(M)$, we have
    \[
    \|z+[\bfa]\|_{\sfH}=\left\|z+[\bfa_{|M'}]\right\|_{\sfH}+\left\|[\bfa_{\setminus M'}]\right\|_{\sfH}\geq t-\left|\itn_{M,z}\right|+2\left\|\bfa_{\setminus M'}\right\|_{\sfH}.
    \]
    If $\|z+[\bfa]\|_{\sfH}=\ell$, it follows that 
    \begin{equation}\label{eq:i_at_least_t-ell}
    \left\|\bfa_{\setminus M'}\right\|_{\sfH}\leq (\ell-t+|\itn_{M,z}|)/2.
    \end{equation}
    Therefore, 
    \begin{equation}\label{eq:induction_level_L0}
    \sum_{\substack{\bfa\in \calX(M),\;\|z+[\bfa]\|_{\sfH}=\ell,\\
        \bfa_{|M'}=\bfa'}} 1\leq \sum_{j=0}^{\lfloor (\ell-t+|\itn_{M,z}|)/2\rfloor}\binom{m}{j}N^{kj}\leq \left(\frac{6mN^{k}}{\ell-t+|\itn_{M,z}|}\right)^{(\ell-t+|\itn_{M,z}|)/2}.
    \end{equation}
    Similarly, we also have
    \begin{equation}\label{eq:induction_level_d}
    \sum_{\substack{\bfa\in \calX(M),\;\|z+[\bfa]\|_{\sfH}=\ell,\\
        \bfa_{|M'}=\bfa'}} \left|\widehat{\phi_A}(\bfa)\right|^2 \leq
    \sum_{\substack{\bfa''\in \calX(M\setminus M')\\ \|\bfa''\|_{\sfH}\leq (\ell-t+|\itn_{M,z}|)/2}} \left|\widehat{\phi_A}(\bfa'\uplus \bfa'')\right|^2
    \leq \left\|\big(\phi_{A}\cdot \overline{\psi_{\bfa'\uplus \mathbf{0}}}\big)^{\leq (\ell-t+|\itn_{M,z}|)/2}\right\|_{2}^{2},
    \end{equation}
    Note that
    \[
\frac{\left\|\phi_{A}\cdot \overline{\psi_{\bfa'\uplus \mathbf{0}}}\right\|_{2}}{\left\|\phi_{A}\cdot \overline{\psi_{\bfa'\uplus \mathbf{0}}}\right\|_{1}}=\frac{\|\phi_{A}\|_{2}}{\|\phi_{A}\|_{1}}=\sqrt{\frac{2^{m}}{|A|}}\leq 2^{s^{*}/2}.
    \]
    If $\ell\leq s^*$, we have $(\ell-t+|\itn_{M,z}|)/2\leq \ell/2\leq \log_{2}(\|\phi_{A}\|_{2}/\|\phi_{A}\|_{1})$, so we may apply Proposition~\ref{prop:level_d_classical} to \eqref{eq:induction_level_d} and get
    \begin{equation}\label{eq:induction_level_d_bound}
     \sum_{\substack{\bfa\in \calX(M),\;\|z+[\bfa]\|_{\sfH}=\ell,\\
        \bfa_{|M'}=\bfa'}} \left|\widehat{\phi_A}(\bfa)\right|^2\leq \left(\frac{8N^{k}s^*}{\ell-t+|\itn_{M,z}|}\right)^{(\ell-t+|\itn_{M,z}|)/2}.
    \end{equation}

    We are now ready to finish the proof. We first establish \eqref{eq:induction_level_d_1}. Apply \eqref{eq:induction_level_L0} to the first square-root term on the right-hand side of \eqref{eq:pre_level_d}, and bound the second term by \(\|\phi_A\|_2 \le 2^{s^*/2}\).
Summing the resulting inequality over all $\bfa' \in \calX(M')$ immediately yields \eqref{eq:induction_level_d_1}, since $|\calX(M')| \le  N^{k|M'|}\leq N^{kt}$.

We next prove \eqref{eq:induction_level_d_2} for $\ell \le s^*$. As before, we apply \eqref{eq:induction_level_L0} to the first square-root term in \eqref{eq:pre_level_d}. This time, however, we bound the second square-root term using \eqref{eq:induction_level_d} together with \eqref{eq:induction_level_d_bound}. Summing over all $\bfa' \in \calX(M')$ then immediately gives \eqref{eq:induction_level_d_2}.
\end{proof}

\subsubsection{Mass Transfer to Low-Weight Coefficients}\label{subsubsec:mass_transfer_to_low_weight}

We can now use \eqref{eq:induction_level_d_2} to calculate an upper bound on $Q_{\Lambda}^{\calU,m}(s^*,t,\ell)$ in the case $1\leq \ell \leq s^*$.

\begin{lemma}\label{lem:mass_transfer_to_low_weight}
Suppose $s^*$ is a real number with $0\leq s^*\leq |\calU|$. For nonnegative integers $t$ and $\ell$ such that $\ell\leq s^*$, we have
\[Q^{\calU,m}_{\Lambda}(s^*,t,\ell)\leq 4\left(\frac{72N^{8k}t}{\sqrt{ms^*}}\right)^{t/2}\left(\frac{16\sqrt{ms^*}}{\ell}\right)^{\ell/2}\left(\frac{16m\sqrt{ms^*}}{|\calU|^{2}}\right)^{(t-\ell)^{+}/2}.\]
\end{lemma}

\begin{proof}
We first expand the definition \eqref{eq:def_of_Q} as follows:
\begin{equation}\label{eq:def_of_Q_expanded}
Q^{\calU,m}_{\Lambda}(s^*,t,\ell)=\max_{z,\calD}\Exu{M\sim\calD}{\sum_{i,b\geq 0}\ind{|\itn_{M,z}|=i\text{ and }|\bd_{M,z}|=b}\cdot\max_{A}\sum_{\bfa\in \calX(M),\;\|z+[\bfa]\|_{\sfH}=\ell} \left|\widehat{\phi_{A}}(\bfa)\right|},
\end{equation}
where the first maximum is taken over all vectors $z\in \ZmodN^{\Lambda}$ of Hamming weight $\|z\|_{\sfH}=t$ and all pseudo-uniform distributions $\calD$ over $\calM_{\calU,m}$, while the second maximum taken is over all sets $A\subseteq \Map{M}{\ZNk}$ of size $|A|\geq 2^{-s^*}N^{km}$. Note that only when $|\itn_{M,z}|\geq t-\ell$ can there exist $\bfa\in \calX(M)$ such that $\|z+[\bfa]\|_{\sfH}=\ell$ (see \eqref{eq:i_at_least_t-ell}). Therefore, applying Definition~\ref{def:q_for_pseudo_uniform} and Lemma~\ref{lem:induction_level_d} to \eqref{eq:def_of_Q_expanded} yields
\begin{align*}
    Q^{\calU,m}_{\Lambda}(s^*,t,\ell)&\leq \sum_{i=(t-\ell)^{+}}^{+\infty}\sum_{b=0}^{+\infty}q^{\calU,m}_{\Lambda}(t,i,b)\cdot N^{2kt}\left(\frac{8\sqrt{ms^*}}{\ell-t+i}\right)^{(\ell-t+i)/2}\\
    &\leq \sum_{i=(t-\ell)^{+}}^{+\infty}\sum_{b=0}^{+\infty}(24k^{2}N^{2k})^{t}\left(\frac{im}{|\calU|^{2}}\right)^{i/2}\left(\frac{m}{|\calU|}\right)^{b}\left(\frac{8\sqrt{ms^*}}{\ell-t+i}\right)^{(\ell-t+i)/2}\tag{using Proposition~\ref{lem:q_calculation}}\\
    &\leq 2(24N^{4k})^{t}
    \sum_{i=(t-\ell)^{+}}^{+\infty}\left(\frac{im}{|\calU|^{2}}\right)^{i/2}\left(\frac{8\sqrt{ms^*}}{\ell-t+i}\right)^{(\ell-t+i)/2}\tag{using $m/|\calU|\leq 1/2$}\\
    &\leq \frac{2\left(24N^{4k}\sqrt{t}\right)^{t}}{(\ell/2)^{\ell/2}}\sum_{i=(t-\ell)^{+}}^{+\infty}\left(\frac{2m}{|\calU|^{2}}\right)^{i/2}\left(8\sqrt{ms^*}\right)^{(\ell-t+i)/2}\tag{using $i^{i}\ell^{\ell}t^{-t}(\ell-t+i)^{-(\ell-t+i)}\leq 2^{i+\ell}$, which is due to Proposition~\ref{prop:calculus}}\\
    &=2\left(\frac{72N^{8k}t}{\sqrt{ms^*}}\right)^{t/2}\left(\frac{16\sqrt{ms^*}}{\ell}\right)^{\ell/2}\sum_{i=(t-\ell)^+}^{+\infty}\left(\frac{16m\sqrt{ms^*}}{|\calU|^{2}}\right)^{i/2}\tag{rearranging}\\
    &\leq 4\left(\frac{72N^{8k}t}{\sqrt{ms^*}}\right)^{t/2}\left(\frac{16\sqrt{ms^*}}{\ell}\right)^{\ell/2}\left(\frac{16m\sqrt{ms^*}}{|\calU|^{2}}\right)^{(t-\ell)^{+}/2},\tag{using $16m\sqrt{ms^*}/|\calU|^{2}\leq 1/4$}
\end{align*}
as desired.
\end{proof}

\subsubsection{Mass Transfer to Intermediate-Weight Coefficients}\label{subsubsec:mass_transfer_to_intermediate_weight}

To handle the case where $\ell>s^*$, we will use \eqref{eq:induction_level_d_1} instead of \eqref{eq:induction_level_d_2}.

\begin{lemma}\label{lem:mass_transfer_to_intermediate_weight}
Suppose $m \le (6k)^{-1}|\mathcal{U}|$ and $s^*$ is a real number with $0\leq s^*\leq |\calU|$. For nonnegative integers $t$ and $\ell$ such that $\ell>s^*$, we have
\[Q^{\calU,m}_{\Lambda}(s^*,t,\ell)\leq 4\left(\frac{2^{16}N^{16k}t}{m}\right)^{t/4}\left(\frac{96m}{\ell}\right)^{\ell/4}\left(\frac{12m^{3}t}{|\calU|^{3}}\right)^{(t-\ell)^{+}/4}.\]
\end{lemma}
\begin{proof}
We first calculate an upper bound for $Q^{\calU,m}_{\Lambda}(s^*,r,\ell)$, for any nonnegative integers $r$ and $\ell$ such that $r\leq k|\calU|$ (we do not assume $\ell>s^*$ for now). 

As in the proof of Lemma~\ref{lem:mass_transfer_to_low_weight}, we apply Definition~\ref{def:q_for_pseudo_uniform} and Lemma~\ref{lem:induction_level_d} to \eqref{eq:def_of_Q_expanded}. The only difference is that rather than applying the bound \eqref{eq:induction_level_d_2} in Lemma~\ref{lem:induction_level_d} that works for $\ell\leq s^*$, this time we apply the bound \eqref{eq:induction_level_d_1} which works for all nonnegative integers $\ell$. We get
\begin{align*}
Q^{\calU,m}_{\Lambda}(s^*,r,\ell)&\leq \sum_{i=(r-\ell)^{+}}^{r}\sum_{b=0}^{+\infty}q^{\calU,m}_{\Lambda}(r,i,b)\cdot 2^{s^*/2}N^{2kr}\left(\frac{6m}{\ell-r+i}\right)^{(\ell-r+i)/4}\\
&\leq 2^{s^*/2}\sum_{i=(r-\ell)^{+}}^{r}\sum_{b=0}^{+\infty}(24k^{2}N^{2k})^{r}\left(\frac{im}{|\calU|^{2}}\right)^{i/2}\left(\frac{m}{|\calU|}\right)^{b}\left(\frac{6m}{\ell-r+i}\right)^{(\ell-r+i)/4}\tag{using Proposition~\ref{lem:q_calculation}}\\
&\leq 2^{s^*/2+1}(24N^{4k})^{r}\sum_{i=(r-\ell)^{+}}^{r}\left(\frac{im}{|\calU|^{2}}\right)^{i/2}\left(\frac{6m}{\ell-r+i}\right)^{(\ell-r+i)/4}\tag{using $m/|\calU|\leq 1/2$}\\
&\leq \frac{2^{s^*/2+1}\left(24N^{4k}r^{1/4}\right)^{r}}{(\ell/2)^{\ell/4}}\sum_{i=(r-\ell)^{+}}^{r}\left(\frac{m\sqrt{2i}}{|\calU|^{2}}\right)^{i/2}(6m)^{(\ell-r+i)/4}
\tag{using $i^{i}\ell^{\ell}r^{-r}(\ell-r+i)^{-(\ell-r+i)}\leq 2^{i+\ell}$, which is due to Proposition~\ref{prop:calculus}}\\
&=2^{s^*/2+1}\left(\frac{2^{11}3^{3}N^{16k}r}{m}\right)^{r/4}\left(\frac{12m}{\ell}\right)^{\ell/4}\sum_{i=(r-\ell)^{+}}^{r}\left(\frac{12m^{3}i}{|\calU|^{4}}\right)^{i/4}\tag{rearranging}\\
&\leq 2^{s^*/2+2}\left(\frac{2^{16}N^{16k}r}{m}\right)^{r/4}\left(\frac{12m}{\ell}\right)^{\ell/4}\left(\frac{12m^{3}r}{|\calU|^{4}}\right)^{(r-\ell)^{+}/4}\tag{using $12m^{3}r/|\calU|^{4}\leq 1/16$}
.\end{align*}

Now we give an upper bound for $Q^{\calU,m}_{\Lambda}(s^*,t,\ell)$ for nonnegative integers $t$ and $\ell$ such that $\ell>s^*$. By Proposition~\ref{prop:Q_tilde} and the calculation above, we have
\begin{align*}
Q^{\calU,m}_{\Lambda}(s^*,t,\ell)&\leq
\max_{r}Q^{\calU,m}_{\Lambda}(s^*,r,\ell-t+r)\tag{where $r$ ranges over all integers such that $(t-\ell)^+\leq r\leq \min\{k|\calU|,t\}$}\\
&\leq \max_{r}\left(2^{s^*/2+2}\left(\frac{2^{16}N^{16k}r}{m}\right)^{r/4}\left(\frac{12m}{\ell-t+r}\right)^{(\ell-t+r)/4}\left(\frac{12m^{3}r}{|\calU|^{4}}\right)^{(t-\ell)^{+}/4}\right)\\
&\leq 2^{s^*/2+2}\left(\frac{2^{16}N^{16k}t}{m}\right)^{t/4}\left(\frac{24m}{\ell}\right)^{\ell/4}\left(\frac{12m^{3}t}{|\calU|^{4}}\right)^{(t-\ell)^{+}/4}\tag{using $r^{r}\ell^{\ell}t^{-t}(\ell-t+r)^{-(\ell-t+r)}\leq 2^{\ell}$, which is due to Proposition~\ref{prop:calculus}}\\
&\leq 4\left(\frac{2^{16}N^{16k}t}{m}\right)^{t/4}\left(\frac{96m}{\ell}\right)^{\ell/4}\left(\frac{12m^{3}t}{|\calU|^{4}}\right)^{(t-\ell)^{+}/4},\tag{using $s^*\leq \ell$}
\end{align*}
as desired.
\end{proof}
\subsection{The Induction Step}\label{subsec:induction_step}
In this subsection, we use the upper bounds on ``Fourier mass transfer'' calculated in Section~\ref{subsec:mass_transfer} to control the Fourier growth of the product function $f\cdot \bdR^{\Lambda,M}_{\mu}[\phi_{A}]$ featured in Lemma~\ref{lem:induction}. In Sections~\ref{subsubsec:inductive_0},~\ref{subsubsec:inductive_low} and~\ref{subsubsec:inductive_intermediate}, we apply the upper bounds from Sections~\ref{subsubsec:mass_transfer_to_0_weight},~\ref{subsubsec:mass_transfer_to_low_weight} and~\ref{subsubsec:mass_transfer_to_intermediate_weight}, respectively, to control the Fourier growth of $f\cdot \bdR^{\Lambda,M}_{\mu}[\phi_{A}]$ on different levels.

\subsubsection{Inductive Bound for the Zero-Weight Coefficient}\label{subsubsec:inductive_0}
\begin{lemma}\label{lem:inductive_bound_for_0}
Let $\calU$ be a $k$-universe embedded in a finite set $\Lambda$. Fix constants $C\in [1,+\infty)$ and $\delta_{0},\delta\in [0,\frac{1}{2}]$. Let $m\in \bZ$ and $s^*\in \bR$ be parameters satisfying
\[0< m\leq \frac{|\calU|^{2}}{2^{14}k^{4}N|\Lambda|}\quad\text{and}\quad \log_{2}(4/\delta)\leq s^*\leq (\delta/2)^{4}C^{-2}|\calU|.\]Suppose $f:\ZmodN^{\Lambda}\rightarrow [0,+\infty)$ is a $(|\calU|,C,s^*,\delta_{0})$-bounded function. Then we have
\begin{equation}\label{eq:induction_to_0_weight}
\sum_{z\in \ZmodN^{\Lambda}\setminus\{0\}}\left|\widehat{f}(z)\right|Q^{\calU,m}_{\Lambda}\big(s^*,\|z\|_{\sfH},0\big)\leq \delta.
\end{equation}
\end{lemma}
\begin{proof}
The left-hand side of \eqref{eq:induction_to_0_weight} is at most $S_{1}+S_{2}$, where
\[
S_{1}:=\sum_{t=1}^{\lfloor s^*\rfloor}Q^{\calU,m}_{\Lambda}(s^*,t,0)\cdot \left\|f^{=t}\right\|_{\sfW},\quad\text{and}\quad S_{2}:=\sum_{t=\lceil s^*\rceil}^{|\Lambda|}Q^{\calU,m}_{\Lambda}(s^*,t,0)\cdot\left\|f^{=t}\right\|_{\sfW}.
\]
Applying Lemma~\ref{lem:mass_transfer_to_0_weight} and Definition~\ref{def:Cs_star_bounded} to $S_{1}$, we can calculate
\begin{align*}
S_{1}&\leq \sum_{t=1}^{+\infty}q^{\calU,m}_{\Lambda}(t,t,0)\cdot F_{C}(|\calU|,t,s^*)\\
&\leq \sum_{t=1}^{+\infty}(24k^{2})^{t}\left(\frac{tm}{|\calU|^{2}}\right)^{t/2}\cdot \left(\frac{C\sqrt{|\calU|s^*}}{t}\right)^{t/2}\tag{using Proposition~\ref{lem:q_calculation} and Definition~\ref{def:Fourier_growth_bound}}\\
&=\sum_{t=1}^{+\infty}\left(\frac{576k^4Cm\sqrt{s^*}}{\left|\calU\right|^{3/2}}\right)^{t/2}\tag{rearranging}\\
&\leq \sum_{t=1}^{\infty}\left(\frac{\delta^{2}}{16}\right)^{t/2}\leq \frac{\delta}{2}.\tag{using the upper bounds on $m$ and $s^*$}
\end{align*}
Applying Lemma~\ref{lem:mass_transfer_to_0_weight} and Proposition~\ref{prop:Cs_star_high_degree_bound} to $S_{2}$, we can calculate
\begin{align*}
S_{2}&\leq \sum_{t=\lceil s^*\rceil}^{+\infty}(24k^{2})^{t}\left(\frac{tm}{|\calU|^{2}}\right)^{t/2}\cdot 2^{s^*/2}\left(\frac{3N|\Lambda|}{t}\right)^{t/2}\tag{using Proposition~\ref{lem:q_calculation}}\\
&\leq \sum_{t=\lceil s^*\rceil}^{+\infty}\left(\frac{2^{12}k^{4}Nm|\Lambda|}{|\calU|^{2}}\right)^{t/2}\tag{replacing $s^*$ by $t$}\\
&\leq \sum_{t=\lceil s^*\rceil}^{+\infty}4^{-t/2}\leq 2^{-s^*+1}\leq \frac{\delta}{2}.\tag{using the upper bound on $m$ and the lower bound on $s^*$}
\end{align*}
Therefore, we have $S_{1}+S_{2}\leq \frac{\delta}{2}+\frac{\delta}{2}=\delta$, as desired.
\end{proof}

\subsubsection{Inductive Bound for Low-Weight Coefficients}\label{subsubsec:inductive_low}
\begin{lemma}\label{lem:inductive_bound_for_low}
Let $\calU$ be a $k$-universe embedded in a finite set $\Lambda$. Fix constants $C\in [1,+\infty)$ and $\delta\in [0,\frac{1}{2}]$. Let $m\in \bZ$ and $s^*\in \bR$ be parameters satisfying
\[0< m\leq \frac{|\calU|^{2}}{2^{15}N^{8k+1}|\Lambda|}\quad\text{and}\quad 1\leq s^*\leq (2C)^{-2}m.\]Suppose $f:\ZmodN^{\Lambda}\rightarrow [0,+\infty)$ is a $(|\calU|,C,s^*,\delta)$-bounded function. Then for each integer $\ell$ such that $1\leq \ell\leq s^*$, we have
\begin{equation}\label{eq:induction_to_low_weight}
\sum_{z\in \ZmodN^{\Lambda}}\left|\widehat{f}(z)\right|Q^{\calU,m}_{\Lambda}\big(s^*,\|z\|_{\sfH},\ell\big)\leq \left(\frac{2^{22}N^{8k}C\sqrt{|\calU|s^*}}{\ell}\right)^{\ell/2}.
\end{equation}
\end{lemma}

\begin{proof}
The left-hand side of \eqref{eq:induction_to_low_weight} is at most $S_{0}+S_{1}+S_{2}+S_{3}$, where
\begin{alignat*}{2}
S_{0}&:=Q^{\calU,m}_{\Lambda}(s^*,0,\ell)\cdot \left\|f^{=0}\right\|_{\sfW},\qquad&
S_{1}&:=\sum_{t=1}^{\lfloor s^* \rfloor}
Q^{\calU,m}_{\Lambda}(s^*,t,\ell)\cdot \left\|f^{=t}\right\|_{\sfW},\\
S_{2}&:=\sum_{t=\lceil s^* \rceil}^{\lfloor C^{-2}|\calU| \rfloor}
Q^{\calU,m}_{\Lambda}(s^*,t,\ell)\cdot \left\|f^{=t}\right\|_{\sfW},\qquad&
S_{3}&:=\sum_{t=\lceil C^{-2}|\calU| \rceil}^{|\Lambda|}
Q^{\calU,m}_{\Lambda}(s^*,t,\ell)\cdot \left\|f^{=t}\right\|_{\sfW}.
\end{alignat*}
It suffices to show that $S_{0},S_{1},S_{2},S_{3}$ are all at most $S:=\left(2^{18}N^{8k}C\sqrt{|\calU|s^*}/\ell\right)^{\ell/2}$.

Applying Lemma~\ref{lem:mass_transfer_to_low_weight} and Definition~\ref{def:Cs_star_bounded} to $S_{0}$, we can calculate
\[
S_{0}\leq 4\left(\frac{16\sqrt{ms^*}}{\ell}\right)^{\ell/2}\cdot(1+\delta)\leq S.
\]
Applying Lemma~\ref{lem:mass_transfer_to_low_weight} and Definition~\ref{def:Cs_star_bounded} to $S_{1}$, we can calculate
\begin{align*}
S_{1}&\leq \sum_{t=1}^{+\infty} 4\left(\frac{72N^{8k}t}{\sqrt{ms^*}}\right)^{t/2}\left(\frac{16\sqrt{ms^*}}{\ell}\right)^{\ell/2}\left(\frac{16m\sqrt{ms^*}}{|\calU|^{2}}\right)^{(t-\ell)^{+}/2}\cdot \left(\frac{C\sqrt{|\calU|s^*}}{t}\right)^{t/2}\\
&=4\left(\frac{16\sqrt{ms^*}}{\ell}\right)^{\ell/2}\cdot \sum_{t=1}^{+\infty}\left(72N^{8k}C\sqrt{\frac{|\calU|}{m}}\right)^{t/2}\left(\frac{16m\sqrt{ms^*}}{|\calU|^{2}}\right)^{(t-\ell)^{+}/2}\tag{rearranging}\\
&\leq 4\left(\frac{2^{11}N^{8k}C\sqrt{|\calU|s^*}}{\ell}\right)^{\ell/2}\cdot \sum_{t=1}^{+\infty}\left(\frac{2^{11}N^{8k}Cm\sqrt{s^*}}{|\calU|^{3/2}}\right)^{(t-\ell)^{+}/2}\tag{using $t\leq \ell+(t-\ell)^+$}\\
&\leq 4\left(\frac{2^{11}N^{8k}C\sqrt{|\calU|s^*}}{\ell}\right)^{\ell/2}\cdot (\ell+1)\leq S\tag{using the upper bounds on $m$ and $s^*$}.
\end{align*}
Applying Lemma~\ref{lem:mass_transfer_to_low_weight} and Definition~\ref{def:Cs_star_bounded} to $S_{2}$, we can calculate
\begin{align*}
S_{2}&\leq \sum_{t=\lceil s^*\rceil}^{\lfloor C^{-2}|\calU|\rfloor} 4\left(\frac{72N^{8k}t}{\sqrt{ms^*}}\right)^{t/2}\left(\frac{16\sqrt{ms^*}}{\ell}\right)^{\ell/2}\left(\frac{16m\sqrt{ms^*}}{|\calU|^{2}}\right)^{(t-\ell)/2}\cdot \left(\frac{C^{2}|\calU|}{t}\right)^{t/4}\\
&=4\left(\frac{|\calU|^{2}}{m\ell}\right)^{\ell/2}\cdot \sum_{t=\lceil s^*\rceil}^{\lfloor C^{-2}|\calU|\rfloor}\left(\frac{1152N^{8k}Cm\sqrt{t}}{|\calU|^{3/2}}\right)^{t/2}\tag{rearranging}\\
&\leq 4\left(\frac{|\calU|^{2}}{m\ell}\right)^{\ell/2}\left(\frac{2^{11}N^{8k}Cm\sqrt{s^*}}{|\calU|^{3/2}}\right)^{s^*/2}\cdot \sum_{t=\lceil s^*\rceil}^{+\infty}\left(\frac{2^{12}N^{8k}m}{|\calU|}\right)^{(t-s^*)/2}\tag{using $t^{t}\leq (s^*)^{s^*}(3t)^{t-s^*}$}\\
&\leq 8 \left(\frac{|\calU|^{2}}{m\ell}\right)^{\ell/2}\left(\frac{2^{11}N^{8k}Cm\sqrt{s^*}}{|\calU|^{3/2}}\right)^{\ell/2}\leq S\tag{using $\ell\leq s^*$ and the upper bound on $m$}.
\end{align*}
Applying Lemma~\ref{lem:mass_transfer_to_low_weight} and Proposition~\ref{prop:Cs_star_high_degree_bound} to $S_{3}$, we can calculate
\begin{align*}
S_{3}&\leq \sum_{t=\lceil C^{-2}|\calU|\rceil}^{|\Lambda|}4\left(\frac{72N^{8k}t}{\sqrt{ms^*}}\right)^{t/2}\left(\frac{16\sqrt{ms^*}}{\ell}\right)^{\ell/2}\left(\frac{16m\sqrt{ms^*}}{|\calU|^{2}}\right)^{(t-\ell)/2}\cdot 2^{s^*/2}\left(\frac{3N|\Lambda|}{t}\right)^{t/2}\\
&\leq 4\left(\frac{|\calU|^{2}}{m\ell}\right)^{\ell/2}\cdot\sum_{t=\lceil C^{-2}|\calU|\rceil}^{|\Lambda|}\left(\frac{2^{13}N^{8k+1}m|\Lambda|}{|\calU|^{2}}\right)^{t/2}\tag{rearranging and using $s^*\leq C^{-2}|\calU|$}\\
&\leq 8\left(\frac{|\calU|^{2}}{m\ell}\right)^{\ell/2}\cdot 4^{-C^{-2}|\calU|/2}\tag{using the upper bound on $m$}\\
&\leq 8\left(\frac{|\calU|^{2}}{m\ell}\right)^{\ell/2}\left(\frac{C^{-2}|\calU|}{\ell}\right)^{-\ell}\tag{using the inequality $2^{-a}\leq a^{-1}$ for any $a>0$}\\
&=8\left(\frac{C^{2}\ell}{m}\right)^{\ell/2}\leq 8\leq S.\tag{using $\ell\leq s^*\leq (2C)^{-2}m$}
\end{align*}

In conclusion, we have shown $S_{0},S_{1},S_{2},S_{3}\leq S$ and thus $S_{0}+S_{1}+S_{2}+S_{3}\leq 4S$, as desired.
\end{proof}

\subsubsection{Inductive Bound for Intermediate-Weight Coefficients}\label{subsubsec:inductive_intermediate}

\begin{lemma}\label{lem:inductive_bound_for_intermediate}
Let $\calU$ be a $k$-universe embedded in a finite set $\Lambda$. Fix constants $C\in [1,+\infty)$ and $\delta\in [0,\frac{1}{2}]$. Let $m\in \bZ$ and $s^*\in \bR$ be parameters satisfying
\[0< m\leq \frac{|\calU|^{3}}{2^{18}kN^{8k+2}|\Lambda|^{2}}\quad\text{and}\quad 2\leq s^*\leq C^{-2}|\calU|.\]Suppose $f:\ZmodN^{\Lambda}\rightarrow [0,+\infty)$ is a $(|\calU|,C,s^*,\delta)$-bounded function. Then for each integer $\ell$ such that $s^*\leq \ell\leq (2C)^{-2}|\calU|$, we have
\begin{equation}\label{eq:induction_to_intermediate_weight}
\sum_{z\in \ZmodN^{\Lambda}}\left|\widehat{f}(z)\right|Q^{\calU,m}_{\Lambda}\big(s^*,\|z\|_{\sfH},\ell\big)\leq \left(\frac{(2^{22}N^{8k}C)^{2}|\calU|}{\ell}\right)^{\ell/4}.
\end{equation}
\end{lemma}
\begin{proof}
The left-hand side of \eqref{eq:induction_to_low_weight} is at most $S_{0}+S_{1}+S_{2}$, where
\begin{alignat*}{2}
S_{0}:=Q^{\calU,m}_{\Lambda}(s^*,0,\ell)\cdot \left\|f^{=0}\right\|_{\sfW},&\qquad&
S_{1}&:=\sum_{t=1}^{\lfloor C^{-2}|\calU| \rfloor}
Q^{\calU,m}_{\Lambda}(s^*,t,\ell)\cdot \left\|f^{=t}\right\|_{\sfW},\\
\text{and}&\qquad&
S_{2}&:=\sum_{t=\lceil C^{-2}|\calU| \rceil}^{|\Lambda|}
Q^{\calU,m}_{\Lambda}(s^*,t,\ell)\cdot \left\|f^{=t}\right\|_{\sfW}.
\end{alignat*}
It suffices to show that $S_{0},S_{1},S_{2}$ are all at most $S:=\left((2^{20}N^{8k}C)^{2}|\calU|/\ell\right)^{\ell/4}$.

Applying Lemma~\ref{lem:mass_transfer_to_intermediate_weight} and Definition~\ref{def:Cs_star_bounded} to $S_{0}$, we can calculate
\[
S_{0}\leq 4\left(\frac{96m}{\ell}\right)^{\ell/4}\cdot(1+\delta)\leq S.
\]
Applying Lemma~\ref{lem:mass_transfer_to_intermediate_weight} and Proposition~\ref{prop:Fourier_growth_crude_bound} to $S_{1}$, we can calculate
\begin{align*}
S_{1}&\leq \sum_{t=1}^{\lfloor C^{-2}|\calU| \rfloor} 4\left(\frac{2^{16}N^{16k}t}{m}\right)^{t/4}\left(\frac{96m}{\ell}\right)^{\ell/4}\left(\frac{12m^{3}t}{|\calU|^{4}}\right)^{(t-\ell)^{+}/4}\cdot 2^{s^*/4}\left(\frac{C^{2}|\calU|}{t}\right)^{t/4}\\
&=2^{s^*/4+2}\left(\frac{96m}{\ell}\right)^{\ell/4}\cdot \sum_{t=1}^{\lfloor C^{-2}|\calU| \rfloor}\left(2^{8}N^{8k}C\sqrt{\frac{|\calU|}{m}}\right)^{t/2}\left(\frac{12m^{3}t}{|\calU|^{4}}\right)^{(t-\ell)^{+}/4}\tag{rearranging}\\
&\leq 4\left(\frac{(2^{12}N^{8k}C)^{2}|\calU|}{\ell}\right)^{\ell/4}\cdot \sum_{t=1}^{\lfloor C^{-2}|\calU| \rfloor}\left(\frac{(2^{10}N^{8k}C)^{2}m^{2}t}{|\calU|^{3}}\right)^{(t-\ell)^{+}/4}\tag{using $s^*\leq \ell$ and $t\leq \ell+(t-\ell)^+$}\\
&\leq 4\left(\frac{(2^{12}N^{8k}C)^{2}|\calU|}{\ell}\right)^{\ell/4}\cdot (\ell+1)\leq S\tag{using the upper bound on $m$}.
\end{align*}
Applying Lemma~\ref{lem:mass_transfer_to_intermediate_weight} and Proposition~\ref{prop:Cs_star_high_degree_bound} to $S_{2}$, we can calculate
\begin{align*}
S_{2}&\leq \sum_{t=\lceil C^{-2}|\calU|\rceil}^{|\Lambda|}4\left(\frac{2^{16}N^{16k}t}{m}\right)^{t/4}\left(\frac{96m}{\ell}\right)^{\ell/4}\left(\frac{12m^{3}t}{|\calU|^{4}}\right)^{(t-\ell)/4}\cdot 2^{s^*/2}\left(\frac{3N|\Lambda|}{t}\right)^{t/2}\\
&\leq 4\left(\frac{8|\calU|^{4}}{m^{2}\ell}\right)^{\ell/4}\cdot\sum_{t=\lceil C^{-2}|\calU|\rceil}^{|\Lambda|}t^{-\ell/4}\left(\frac{2^{16}N^{8k+1}m|\Lambda|}{|\calU|^{2}}\right)^{t/2}\tag{using $s^*\leq C^{-2}|\calU|$ and rearranging}\\
&\leq 4\left(\frac{8C^{2}|\calU|^{3}}{m^{2}\ell}\right)^{\ell/4}\cdot\sum_{t=\lceil C^{-2}|\calU|\rceil}^{|\Lambda|}\left(\frac{2^{16}N^{8k+1}m|\Lambda|}{|\calU|^{2}}\right)^{t/2}\\
&\leq 8\left(\frac{8C^{2}|\calU|^{3}}{m^{2}\ell}\right)^{\ell/4}\left(\frac{2^{16}N^{8k+1}m|\Lambda|}{|\calU|^{2}}\right)^{\ell}\tag{using the upper bound on $m$ and $\ell\leq (2C)^{-2}|\calU|$}\\
&= 8\left(\frac{8C^{2}|\calU|}{\ell}\right)^{\ell/4}\left(\frac{2^{16}N^{8k+1}\sqrt{m}|\Lambda|}{|\calU|^{3/2}}\right)^{\ell}\tag{rearranging}\\
&\leq 8\left(\frac{(2^{16}N^{8k}C)^{2}|\calU|}{\ell}\right)^{\ell/4}\leq S.\tag{using the upper bound on $m$ and rearranging}
\end{align*}

In conclusion, we have shown $S_{0},S_{1},S_{2}\leq S$ and thus $S_{0}+S_{1}+S_{2}\leq 3S$, as desired.
\end{proof}

\subsubsection{Finishing the Proof}

We are now ready to finish the proof of the induction lemma, restated below.

\leminduction*

\begin{proof}
For any fixed matching $M\in\calM_{\calU,m}$, subset $A\subseteq \Map{M}{\ZNk}$, one-wise independent distribution $\mu$ over $\ZNk$, function $f\in L^{2}(\ZmodN^{\Lambda})$ and nonnegative integer $\ell$, we have
\begin{align*}
\left\|\left(f\cdot \bdR^{\Lambda,M}_{\mu}[\phi_{A}]\right)^{=\ell}\right\|_{\sfW}
&\leq \sum_{z\in\ZmodN^{\Lambda}}\left(\left|\widehat{f}(z)\right|\sum_{b\in \ZmodN^{\Lambda},\;\|z+b\|_{\sfH}=\ell}\left|\left\langle \bdR^{\Lambda,M}_{\mu}[\phi_{A}],\chi_{b}\right\rangle\right|\right)\tag{by the convolution theorem}\\
&\leq \sum_{z\in\ZmodN^{\Lambda}}\left(\left|\widehat{f}(z)\right|\sum_{\bfa\in\calX(M),\;\|z+[\bfa]\|_{\sfH}=\ell}\left|\widehat{\phi_{A}}(\bfa)\right|\right).\tag{using Lemma~\ref{lem:SVD}}
\end{align*}
Taking maximum over all $A$ with $|A|\geq 2^{-s^*}N^{km}$, maximum over all one-wise independent $\mu$, and expectation over $M\sim\calD$, we get (using Definition~\ref{def:transfer_Q})
\begin{equation}\label{eq:turn_into_Q}
\Exu{M\sim \calD}{\max_{A,\,\mu}\left\|\left(f\cdot \bdR^{\Lambda,M}_{\mu}[\phi_{A}]\right)^{=\ell}\right\|_{\sfW}}\leq \sum_{z\in\ZmodN^{\Lambda}}\left|\widehat{f}(z)\right|Q^{\calU,m}_{\Lambda}(s^*,\|z\|_{\sfH},\ell).
\end{equation}

Now suppose $f:\ZmodN^{\Lambda}\rightarrow[0,+\infty)$ is a $(|\calU|,C,s^*,\delta_{0})$-bounded function. Taking $\ell=0$ and substituting $(f-1)$ for $f$ in \eqref{eq:turn_into_Q} yields
\[
\Exu{M\sim \calD}{\max_{A,\,\mu}\Big|\Ex{(f-\Ex{f})\cdot\bdR^{\Lambda,M}_{\mu}[\phi_{A}]}\Big|}\leq \sum_{z\in\ZmodN^{\Lambda}\setminus\{0\}}\left|\widehat{f}(z)\right|Q^{\calU,m}_{\Lambda}(s^*,\|z\|_{\sfH},0)\leq \delta,
\]
where we used Lemma~\ref{lem:inductive_bound_for_0} in the last transition. Since the function $\bdR^{\Lambda,M}_{\mu}[\phi_{A}]$ has expected value 1 by Proposition~\ref{prop:preverses_density_function} , this means
\[
\Exu{M\sim \calD}{\max_{A,\,\mu}\Big|\Ex{f\cdot\bdR^{\Lambda,M}_{\mu}[\phi_{A}]}-\Ex{f}\Big|}\leq \delta.
\]
Noting that $|\Ex{f}-1|\leq \delta_{0}$ (due to the $(|\calU|,C,s^*,\delta_{0})$-boundedness assumption on $f$) and using Markov's inequality, we obtain
\begin{equation}\label{eq:induction_level_0_prob}
\Pru{M\sim \calD}{\max_{A,\,\mu}\Big|\Ex{f\cdot\bdR^{\Lambda,M}_{\mu}[\phi_{A}]}-1\Big|\geq \delta_{0}+\eta^{-1}\delta}\leq \eta.
\end{equation}
For any integer $\ell$ such that $1\leq \ell\leq s^*$, combining \eqref{eq:turn_into_Q} with Lemma~\ref{lem:inductive_bound_for_low} yields
\[
\Exu{M\sim \calD}{\max_{A,\,\mu}\left\|\left(f\cdot \bdR^{\Lambda,M}_{\mu}[\phi_{A}]\right)^{=\ell}\right\|_{\sfW}}\leq \left(\frac{2^{22}N^{8k}C\sqrt{|\calU|s^*}}{\ell}\right)^{\ell/2}=F_{2^{22}N^{8k}C}\big(|\calU|,\ell,s^*\big),
\]
so by Markov's inequality and Proposition~\ref{prop:Fourier_bound_non_decreasing} we have
\begin{equation}\label{eq:induction_level_low_prob}
\Pru{M\sim \calD}{\max_{A,\,\mu}\left\|\left(f\cdot \bdR^{\Lambda,M}_{\mu}[\phi_{A}]\right)^{=\ell}\right\|_{\sfW}\geq F_{2^{22}\eta^{-2}N^{8k}C}\big(|\calU|,\ell,2s^*\big)}\leq (\eta^{2})^{\ell/2}=\eta^{\ell}.
\end{equation}
For any integer $\ell$ such that $s^*\leq \ell\leq (2C)^{-2}|\calU|$, combining \eqref{eq:turn_into_Q} and Lemma~\ref{lem:inductive_bound_for_intermediate} yields
\[
\Exu{M\sim \calD}{\max_{A,\,\mu}\left\|\left(f\cdot \bdR^{\Lambda,M}_{\mu}[\phi_{A}]\right)^{=\ell}\right\|_{\sfW}}\leq \left(\frac{(2^{22}N^{8k}C)^{2}|\calU|}{\ell}\right)^{\ell/4}=F_{2^{22}N^{8k}C}\big(|\calU|,\ell,s^*\big),
\]
so by Markov's inequality and Proposition~\ref{prop:Fourier_bound_non_decreasing} we have
\begin{equation}\label{eq:induction_level_intermediate_prob}
\Pru{M\sim \calD}{\max_{A,\,\mu}\left\|\left(f\cdot \bdR^{\Lambda,M}_{\mu}[\phi_{A}]\right)^{=\ell}\right\|_{\sfW}\geq F_{2^{22}\eta^{-2}N^{8k}C}\big(|\calU|,\ell,2s^*\big)}\leq (\eta^{4})^{\ell/4}=\eta^{\ell}.
\end{equation}
Finally, noting that $\|f\|_{\infty}\leq 2^{s^*}$ (due to the $(|\calU|,C,s^*,\delta_{0})$-boundedness assumption on $f$) and using Proposition~\ref{prop:infinity_norm_contraction}, we have\footnote{Here we need to recall that the maximum is taken over all $A\subseteq \Map{M}{\ZNk}$ such that $|A|\geq 2^{-s^*}N^{km}$, for which $\|\phi_{A}\|_{\infty}=|A|^{-1}N^{km}\leq 2^{s^*}$.}
\begin{equation}\label{eq:induction_infinity_norm_prob}
\Pru{M\sim \calD}{\max_{A,\,\mu}\left\|f\cdot \bdR^{\Lambda,M}_{\mu}[\phi_{A}]\right\|_{\infty}\leq 2^{2s^*}}=1.
\end{equation}
Combining \eqref{eq:induction_level_0_prob}, \eqref{eq:induction_level_low_prob}, \eqref{eq:induction_level_intermediate_prob} and \eqref{eq:induction_infinity_norm_prob} by union bound, we conclude that with probability at least
\(
1-\eta-\sum_{\ell=1}^{+\infty}\eta^{\ell}\geq 1-3\eta
\)
over $M\sim \calD$, any one-wise independent distribution $\mu$ and any $A\subseteq \Map{M}{\ZNk}$ with $|A|\geq 2^{-s^*}N^{km}$ satisfy:
\[\text{the function }f\cdot \bdR^{\Lambda,M}_{\mu}[\phi_{A}]\text{ is }\Big(|\calU|,\,2^{22}\eta^{-2}N^{8k}C,\,2s^*,\,\delta_{0}+\eta^{-1}\delta\Big)\text{-bounded.}\qedhere\]
\end{proof}

%% file: twowise_independence.tex
\section{The Two-Wise Independent Case}

In Sections~\ref{sec:communication_lower_bound} to~\ref{sec:induction}, we have shown that given a distribution-labeled $k$-graph $G=(\calV,\calE,N,(\mu_{\sfe})_{\sfe\in \calE})$ and a sufficiently small constant $\alpha$, the communication game $\DIHP(G,n,\alpha,K)$ has communication complexity at least $\Omega(\sqrt{n})$. The assumption that the distributions $\mu_{\sfe}$ (for $\sfe\in\calE$) are one-wise independent has played a crucial role in the proof (see e.g. the proofs of Lemmas~\ref{lem:structure_part_mean} and~\ref{lem:SVD}). In this section, we show that if the distributions $\mu_{\sfe}$ are further assumed to be \emph{two-wise independent}, then the communication lower bound for $\DIHP(G,n,\alpha,K)$ can be improved from $\Omega(\sqrt{n})$ to $\Omega(n)$.

\begin{restatable}{theorem}{twowiselowerbound}\label{thm:twowise}
    Fix a independent distribution-labeled $k$-graph $G=(\calV,\calE,N,(\mu_{\sfe})_{\sfe\in \calE})$, an integer $K>0$ and a parameter $\alpha\in \big(0,2^{-20}N^{-10k}|\calV|^{-2}\big]$. If $\mu_{\sfe}$ is two-wise independent for any $\sfe\in\calE$, then there exists a constant $\gamma=\gamma(G,\alpha,K)>0$ such that $\CC(G,n,\alpha,K)\geq \gamma n$.
\end{restatable}

The proof of Theorem~\ref{thm:twowise} follows the same approach as the proof of Theorem~\ref{thm:sqrt_lower_bound} in Sections~\ref{sec:communication_lower_bound} to~\ref{sec:induction}. Given a low-cost communication protocol for $\DIHP(G,n,\alpha, K)$, the plan is to first decompose the joint input space into structured rectangles satisfying certain conditions, and then prove discrepancy bounds for these rectangles. For Theorem~\ref{thm:twowise}, the decomposition step is carried out in Section~\ref{subsec:decomposition_twowise}, followed by the analysis of the discrepancy in Section~\ref{subsec:discrepancy_twowise}. In Section~\ref{subsec:implication_twowise}, we prove Thoerems~\ref{thm:multi_pass_lin} and~\ref{thm:unbounded_width} using Theorem~\ref{thm:twowise}.

\subsection{Decomposition into ``Fair'' Rectangles}\label{subsec:decomposition_twowise}

Throughout this subsection, we fix a distribution-labeled $k$-graph $G=(\calV,\calE,N,(\mu_{\sfe})_{\sfe\in \calE]})$ and a communication game $\DIHP(G,n,\alpha,K)$.

Lemma~\ref{lem:regularity_decomposition}, we showed that the rectangle decomposition induced by any protocol $\Pi$ with $|\Pi| = o(\sqrt{n})$ can be refined into ``good'' structured rectangles. 
For a structured rectangle $(\bdzeta, R)$ to be deemed ``good,'' the restriction sequence $\bdzeta$ must be acyclic (see Definition~\ref{def:good_rec_2}). 
Informally, this reflects the fact that a communication protocol with $o(\sqrt{n})$ communication cost is incapable of detecting cycles in the hypergraph formed by the joint input.

This phenomenon no longer holds if the communication bound is relaxed from $o(\sqrt{n})$ to $o(n)$ (see e.g. Theorem~\ref{thm:bipartiteness_tester}). 
Consequently, to obtain an analogous decomposition result for protocols with $o(n)$ communication cost, we must correspondingly weaken the notion of goodness imposed on the structured rectangles.

To this end, we seek a relaxation of acyclicity (cf.~Definition~\ref{def:cyclic_res_seq}) for restriction sequences $\bdzeta$ and their associated hypergraphs $H_{\bdzeta}$. Recall from Section~\ref{subsec:general_notations} that a $k$-uniform hypergraph is said to be acyclic if any $\ell$ edges together cover more than $\ell(k-1)$ vertices. This motivates the following definition.

\begin{definition}
A $k$-uniform hypergraph $H$ is said to be \emph{$C$-locally-almost-acyclic}, where $C$ is a positive integer, if for any nonnegative integer $\ell\leq C$, any $\ell$ distinct edges in $H$ together cover at least $\ell(k-1.1)$ vertices. If $H$ is $C$-locally-almost-acyclic for any $C\geq 1$, we simply say that $H$ is \emph{almost-acyclic}.
\end{definition}

In Definition~\ref{def:weight_of_restriction}, we defined a hypergraph $H_{\bdzeta}$ for any restriction sequence $\bdzeta$. For notational convenience in this subsection, we will use the same notation for a general sequence of labeled matchings.

\begin{notation}
Given a sequence $\bfZ=(\bfz^{(\sfe,j)})_{(\sfe,j)\in \calE\times [K]}$ of labeled matchings, where $\bfz^{(\sfe,j)}\in \Omega^{\calU_{\sfe},\leq \alpha n}$ for any $(\sfe,j)\in \calE\times [K]$, we let $H_{\bfZ}$ denote the hypergraph with vertex set $\calV\times [n]$ and edge set
\[
\bigcup_{(\sfe,j)\in \calE\times [K]}\supp(\bfz^{(\sfe,j)}).
\]
\end{notation}

\begin{definition}\label{def:almost_acyclic_labeled_matchings}
A sequence of labeled matchings $\bfZ=(\bfz^{(\sfe,j)})_{(\sfe,j)\in \calE\times [K]}$ is said to be $C$-locally-almost-acyclic (respectively, almost-acyclic) if the edge sets $\big(\supp(\bfz^{(\sfe,j)})\big)_{(\sfe,j)\in \calE\times [K]}$ are pairwise disjoint, and the hypergraph $H_{\bfZ}$ is $C$-locally-almost-acyclic (respectively, almost-acyclic). 
\end{definition}

We can now state the relaxed version of the goodness notion in Definition~\ref{def:good_rec_2}, by replacing acyclicity with almost-acyclicity.

\begin{definition}\label{def:fair_rec}
    Let $W$ be a positive real number. We say a structured rectangle $(\boldsymbol{\zeta},R)$, where \(R=\prod_{(\sfe, j) \in \calE \times [K]} A^{(\sfe, j)}\) and \(\bdzeta = \left(\bfz^{(\sfe, j)}\right)_{(\sfe, j) \in \calE \times [K]}\), is \emph{$W$-fair} if the following conditions hold:
    \begin{enumerate}[label=(\arabic*)]
        \item The restriction sequence $\bdzeta$ is almost-acyclic.
        \item $\sum_{(\mathsf{e},j)}\left|\supp(\bfz^{(\mathsf{e},j)} )\right| \leq W$.
        \item  $\left|A^{(\sfe,j)}\right| / \left|\Omega^{\calU_{\sfe},\alpha n}_{\bfz^{(\sfe,j)}}\right|\geq 2^{-W}$ for all $(\sfe,j)\in \calE\times [K]$. 
    \end{enumerate}
\end{definition}

The remaining goal of this subsection is to prove that the rectangle decomposition induced by any protocol $\Pi$ with $|\Pi|=o(n)$ can be refined into ``fair'' structured rectangles. The proof of this decomposition lemma turns out to be much simpler than its counterpart Lemma~\ref{lem:regularity_decomposition} (which is proved in \cite[Appendix A]{FMW25b}). The main reason this decomposition lemma is simpler is the fact that, roughly speaking, a ``random regular hypergraph'' is with high probability locally-almost-acyclic, as we formlize below.

Recall that in the $\DIHP(G,n,\alpha,K)$ game, each player $(\sfe,j)$ receives a labeled matching $\bfy^{(\sfe,j)}\in\Omega^{\calU_{\sfe},\alpha n}$. In both the yes case and the no case, the support of $\bfy^{(\sfe,j)}$ is a uniformly random matching $M^{(\sfe,j)}\in \calM_{\calU_{\sfe},\alpha n}$ (independent from the other players). For any matching $M\in \calM_{\calU_{\sfe},\alpha n}$ and any vertex set $V\subseteq \calV\times [n]$, we let $M[V]$ denote the set of hyperedges $(v_{1},\dots,v_{k})\in M$ such that $v_{1},\dots,v_{k}\in V$. The following lemma is a slightly extended version of \cite[Lemma 6]{BOT02}:

\begin{lemma}\label{lem:random_graph_no_dense_sets}
Fix a distribution-labeled $k$-graph $G=(\calV,\calE,N,(\mu_{\sfe})_{\sfe\in \calE})$, an integer $K>0$ and a parameter $\alpha\in (0,1]$.\footnote{In this lemma, one can without loss of generality assume $\alpha=1$.} For any constant $p>\frac{1}{k-1}$, there exists a constant $\delta\in (0,1)$ such that the following holds: if $M^{(\sfe,j)}$ is sampled independently and uniformly from $\calM_{\calU_{\sfe},\alpha n}$ for each player $(\sfe,j)\in \calE\times [K]$, then the condition
\begin{equation}\label{eq:locally_sparse}
\sum_{(\sfe,j)\in \calE\times [K]}\left|M^{(\sfe,j)}[V]\right|\leq  p|V|\qquad\text{for all }V\subseteq \calV\times [n]\text{ with }|V|\leq \delta n
\end{equation}
is satisfied with probability $1-o(1)$, where $o(1)$ denotes a term tending to 0 as $n\rightarrow +\infty$.
\end{lemma}

\begin{proof}
For a fixed $V\subseteq\calV\times [n]$ such that $|V|\leq \delta n$, let $S_{V}$ be the set of tuples $(v,\sfe,j)\in V\times \calE\times [K]$ such that some edge $(v_{1},\dots,v_{k})\in M^{(\sfe,j)}[V]$ has $v$ as its first vertex, i.e. $v_{1}=v$. We therefore have
\[
|S_{V}|=\sum_{(\sfe,j)\in \calE\times [K]}\left|M^{(\sfe,j)}[V]\right|.
\]
For any fixed subset $S\subseteq V\times \calE\times [K]$, it is not hard to see that when the random matchings $M^{(\sfe,j)}$ are independently sampled, we have
\[
\Pr{S\subseteq S_{V}}\leq \left(\frac{|V|^{k-1}}{n^{k-1}}\right)^{|S|}.
\]
Therefore we have
\begin{align}
\Pr{|S_{V}|\geq p|V|}&\leq\sum_{\substack{S\subseteq V\times \calE\times [K]\\ |S|=\lceil p|V|\rceil}}\Pr{S\subseteq S_{V}}\leq \binom{K|\calE|\cdot |V|}{\lceil p|V|\rceil}\left(\frac{|V|}{n}\right)^{(k-1)\lceil p|V|\rceil}\nonumber\\
&\leq \left(\frac{3K|\calE|}{p}\right)^{p|V|}\left(\frac{|V|}{n}\right)^{(k-1)p|V|}=\left(C_{1}\cdot\frac{ |V|^{(k-1)p}}{n^{(k-1)p}}\right)^{|V|},\label{eq:locally_sparse_fixed_V}
\end{align}
where we used $|V|\leq n$ in the third transition and denoted $C_{1}=(3K|\calE|/p)^{p}$.

Let $C_{2}=3C_{1}|\calV|$ and pick $\delta\in (0,1)$ such that  $C_{2}\cdot\delta^{(k-1)p-1}\leq 1/2$. Applying union bound over all $V\subseteq \calV\times [n]$ with $|V|\leq \delta n$ to \eqref{eq:locally_sparse_fixed_V}, we get
\[
\Pr{\text{condition \eqref{eq:locally_sparse} fails}}\leq \sum_{r=1}^{\lfloor \delta n\rfloor}\binom{n|\calV|}{r}\left(C_{1}\cdot\frac{r^{(k-1)p}}{n^{(k-1)p}}\right)^{r}\leq \sum_{r=1}^{\lfloor\delta n\rfloor}\left(C_{2}\cdot\frac{r^{(k-1)p-1}}{n^{(k-1)p-1}}\right)^{r}\leq o(1).\qedhere
\]
\end{proof}

Lemma~\ref{lem:random_graph_no_dense_sets} immediately implies the following corollary.

\begin{corollary}\label{cor:random_graph_no_dense_sets}
Fix a distribution-labeled $k$-graph $G=(\calV,\calE,N,(\mu_{\sfe})_{\sfe\in \calE})$, an integer $K>0$ and a parameter $\alpha\in (0,1]$. There exists a constant $\delta\in (0,1)$ such that 
\[
\Pru{\bfY\sim\calD_{\no}}{\bfY\text{ is }\delta n\text{-locally-almost-acyclic}}=1-o(1).
\]
\end{corollary}
\begin{proof}
It suffices to note that if two of the matchings $M^{(\sfe,j)}=\supp(\bfy^{(\sfe,j)})$, where $(\sfe,j)$ ranges in $\calE\times [K]$, have a common edge $(v_{1},\dots,v_{k})$, then condition~\eqref{eq:locally_sparse} is violated for $p=1/(k-1.1)$ and $V=\{v_{1},\dots,v_{k}\}$.
\end{proof}

We are now ready to state and prove the decomposition lemma needed for Theorem~\ref{thm:twowise}.

\begin{lemma}\label{lem:decomposition_twowise}
    Fix a distribution-labeled $k$-graph $G=(\calV,\calE,N,(\mu_{\sfe})_{\sfe\in \calE})$, an integer $K>0$ and a parameter $\alpha \in (0,1)$. Assume $k\geq 3$. Then there exists a constant $\eta\in (0,1)$ such that given any communication protocol $\Pi$ for $\DIHP(G,n,\alpha,K)$ with $|\Pi|\leq \eta\cdot  n$, there exists a collection $\calR$ of pairwise-disjoint structured rectangles $(\boldsymbol{\zeta},R)$ in the space \(\prod_{(\sfe, j) \in \calE \times [K]} \Omega^{\calU_{\sfe}, \alpha n}\) such that the following conditions hold: 
    \begin{enumerate}[label=(\arabic*)]
        \item \(\calD_{\no}\left(\bigcup_{(\boldsymbol{\zeta},R)\in \calR}R\right)\geq 0.99\).
        \item Each $(\boldsymbol{\zeta},R)\in\calR$ is $(10^{3}|\Pi|)$-fair.
         \item For each $(\boldsymbol{\zeta},R)\in \calR $, there exists $a_R \in \{0,1\}$ such that $\Pi(\bfY) = a_R$ for every $\bfY \in R$. 
    \end{enumerate}
\end{lemma}

\begin{proof}[Proof]
By transforming the protocol $\Pi$ into a global protocol $\Pi^{\reff}$ and taking the leaf rectangles of $\Pi^{\reff}$, it is shown (implicitly) in \cite[Proof of Lemma A.8]{FMW25b} that the space \(\prod_{(\sfe, j) \in \calE \times [K]} \Omega^{\calU_{\sfe}, \alpha n}\) can be partitioned into a collection $\calR^{\textup{leaf}}$ of structured rectangles such that:
\begin{itemize}
    \item $\sum_{(\bdzeta,R)\in \calR^{\textup{leaf}}} \calD_{\no}(R)\cdot \phi(\bdzeta,R)  \leq 3|\Pi|$ (the potential function $\phi$ is defined in Definition~\ref{def:restriction_potential}).
    \item For each $(\boldsymbol{\zeta},R)\in \calR^{\textup{leaf}}$, there exists $a_R \in \{0,1\}$ such that $\Pi(\bfY) = a_R$ for every $\bfY \in R$. 
\end{itemize}
Let $\calR_{1}\subseteq \calR^{\textup{leaf}}$ be the sub-collection of structured rectangles $(\bdzeta,R)$ such that $\bdzeta$ is not almost-acyclic, and let $\calR_{2}\subseteq \calR^{\textup{leaf}}$ be the collection of those violating the second or third conditions of $(10^{3}|\Pi|)$-fairness (as per Definition~\ref{def:fair_rec}). Using the definition of $\phi(\bdzeta,R)$, we know that $\phi(\bdzeta,R)\geq 10^{3}|\Pi|$ for any $(\bdzeta,R)\in \calR_{2}$. Therefore, by Markov's inequality we have
\[
\sum_{(\bdzeta,R)\in \calR_{2}}\calD_{\no}(R)\leq \frac{1}{10^{3}|\Pi|}\sum_{(\bdzeta,R)\in \calR^{\textup{leaf}}}\calD_{\no}(R)\cdot \phi(\bdzeta,R)\leq 3\cdot 10^{-3}.
\]
It suffices to show that if $\eta\in (0,1)$ is chosen to be small enough, we have
\begin{equation}\label{eq:few_are_not_almost_acyclic}
\sum_{(\bdzeta,R)\in \calR_{1}\setminus \calR_{2}}\calD_{\no}(R)\leq 10^{-3}.
\end{equation}
Once we have that, the collection $\calR=\calR^{\textup{leaf}}\setminus(\calR_{1}\cup\calR_{2})$ would satisfy the three conditions in this lemma.

For any $(\bdzeta,R)\in \calR^{\textup{leaf}}\setminus \calR_{2}$, since the second condition of $(10^{3}|\Pi|)$-fairness is satisfied, we know that $H_{\bdzeta}$ has at most $10^{3}|\Pi|$ edges. Furthermore, for any $\bfY\in R$, we know that $H_{\bdzeta}$ is a subgraph of $H_{\bfY}$. Therefore, if $\bdzeta$ is not almost-acyclic, then any $\bfY\in R$ is not $(10^{3}|\Pi|)$-locally-almost-acyclic. From Corollary~\ref{cor:random_graph_no_dense_sets} we know that if the constant $\eta$ is chosen to be small enough, then
\[
\Pru{\bfY\sim\calD_{\no}}{\bfY\text{ is not }(10^{3}|\Pi|)\text{-locally-almost-acyclic}}=o(1),
\]
and hence 
\[
\sum_{(\bdzeta,R)\in \calR^{\textup{leaf}}\setminus \calR_{2}}\ind{\bdzeta\text{ is not locally-acyclic}}\cdot \calD_{\no}(R)\leq o(1).
\]
This clearly implies the desired inequality \eqref{eq:few_are_not_almost_acyclic}.
\end{proof}


\subsection{Discrepancy of ``Fair'' Rectangles}\label{subsec:discrepancy_twowise}

The reason we were content with decomposing into ``fair'' structured rectangles (instead of requiring ``good'' ones) in Section~\ref{subsec:decomposition_twowise} is that when the distributions $\mu_{\sfe}$ in $G=(\calV,\calE,N,(\mu_{\sfe})_{\sfe\in \calE})$ are two-wise independent, we can prove a discrepancy bound even for fair rectangles. 

\begin{lemma}\label{lem:discrepancy_twowise}
    Fix a two-wise independent distribution-labeled $k$-graph $G=(\calV,\calE,N,(\mu_{\sfe})_{\sfe\in \calE})$, an integer $K>0$ and a parameter $\alpha\in \left(0,2^{-20}N^{-10k}|\calV|^{-2}\right]$. There exists a constant $\gamma\in (0,1)$ such that for any $\gamma n$-fair structured rectangle $(\boldsymbol{\zeta},R)$, we have
    \begin{align*}
        \calD_{\yes}(R)\geq (1-10^{-3})\cdot \calD_{\no}(R).
    \end{align*}
\end{lemma}

As is the case with Lemma~\ref{lem:discrepancy_bound}, the proof of Lemma~\ref{lem:discrepancy_twowise} proceeds by handling the ``structured information'' and ``pseudorandom noise'' in $R$ separately. It turns out that the psueodorandom part in Lemma~\ref{lem:discrepancy_twowise} can be handled in essentially the same way as in Section~\ref{subsec:discrepancy_bound}, while the structured part requires a different treatment from its counterpart in Section~\ref{subsec:structured}. In Section~\ref{subsubsec:structured_two_wise}, we prove a lemma (resembling Lemma~\ref{lem:structured_part}) that handles the structured part of Lemma~\ref{lem:discrepancy_twowise}; the proofs of Lemma~\ref{lem:discrepancy_twowise} and Theorem~\ref{thm:twowise} can then be finished rather quickly in Section~\ref{subsubsec:induction_twowise}.

\subsubsection{Fourier Growth of The Structured Part}\label{subsubsec:structured_two_wise}

Recall that in Section~\ref{subsec:structured}, in order to analyze the ``structured-part function'' $g_{\bdzeta}$ (defined in Definition~\ref{def:g_bdzeta}), a key step is to answer a purely combinatorial question: in a hypergraph, how many vertex subsets (of a given size) intersect each connected component in at least 2 vertices? It turns out that under the new set of assumptions stated in Lemma~\ref{lem:discrepancy_twowise}, the analysis of the same function $g_{\bdzeta}$ hinges on a different set of combinatorial questions about hypergraphs.

In a large hypergraph $H=(V,E)$, if a subset $E'\subseteq E$ of $r$ edges is chosen uniformly at random, most likely they will not be incident to each other. Our first question is, for how many size-$r$ subset $E'\subseteq E$ does the subgraph $(V,E')$ have only a small number of connected components? We answer this question in Lemma~\ref{lem:bound_for_frakE} by assuming a degree bound on the original hypergraph $H$.

\begin{definition}
Let $t,r$ be positive integers such that $r\geq t$, and let $H=(V,E)$ be a $k$-uniform hypergraph. We use $\mathfrak{E}(t,r,H)$ to denote the collection of subsets $E'\subseteq E$ such that $|E'|=r$ and the number of nontrivial connected components of the subgraph $(V,E')$ is $t$.
\end{definition}

\begin{lemma}\label{lem:bound_for_frakE}
Let $d,t,r$ be positive integers such that $r\geq t$. Let $H$ be a $k$-uniform hypergraph with $m$ edges, and suppose that every edge in $H$ is incident to at most $d$ other edges. Then we have $|\mathfrak{E}(t,r,H)|\leq (6d)^{r}(3m/t)^{t}$.
\end{lemma}
\begin{proof}
A sequence of $t$ connected subgraphs of $H$ with $r$ edges in total can be chosen by the following process:
\begin{enumerate}
\item First choose positive integers $c_{1},\dots,c_{t}$ such that $c_{1}+\dots+c_{t}=r$. The number of such tuples $(c_{1},\dots,c_{t})$ is $\binom{r-1}{t-1}\leq 2^{r}$.
\item Then for each $i\in [t]$, choose a connected subgraph of $H$ with $c_{i}$ edges. The number of such connected subgraphs is at most $m(3d)^{c_{i}}$, by standard results in enumerative combinatorics.\footnote{For example, one can first transit to the line graph of $H$ and then use \cite[Lemma 2.1(c)]{BCKL13}.} 
The total number of choices in this step is thus at most $m^{t}(3d)^{c_{1}+\dots+c_{t}}=m^{t}(3d)^{r}$.
\end{enumerate}
Since each member of the collection $\mathfrak{E}(t,r,H)$ is counted $t!$ times by this method, we have
\[
|\mathfrak{E}(t,r,H)|\leq \frac{2^{r}\cdot m^{t}(3d)^{r}}{t!}\leq (6d)^{r}\left(\frac{3m}{t}\right)^{t}.\qedhere
\]
\end{proof}

While in Section~\ref{subsec:structured}, the number of nonzero Fourier coefficients of the function $g_{\bdzeta}$ was controlled by the number of ways to choose at least 2 vertices from each connected component of a hypergraph, the latter quantity is unaffordable as an upper bound in the current context. In order to utilize the two-wise independence, we will associate the nonzero Fourier coefficients of $g_{\bdzeta}$ with a slightly more complicated combinatorial object: the collection $\calB(E)$ defined in the next definition.

\begin{definition}\label{def:calB}
Suppose $(V,E)$ is a $k$-uniform hypergraph. We define $\calB(E)$ to be the collection of maps $B:E\rightarrow\ZmodN^{V}$ such that for each $e=(v_{1},\dots,v_{k})\in E$, we have
\[B(e)\in \ZmodN^{V},\qquad \supp(B(e))\subseteq\{v_{1},\dots,v_{k}\},\qquad\text{and}\qquad\left|\supp(B(e))\right|=\|B(e)\|_{\sfH}\geq 3.\]
For a map $B\in \calB(E)$, we denote $\boldsymbol{\Sigma} B:=\sum_{e\in E}B(e)$. For any nonnegative integer $\ell$, we additionally define
\[
\calB_{\ell}(E)=\big\{B\in \calB(E)\,\big|\,\|\boldsymbol{\Sigma} B\|_{\sfH}=\ell\big\}\subseteq \ZmodN^{V}.
\]
\end{definition}

Before explaining how~\Cref{def:calB} is related to the nonzero Fourier coefficients of $g_{\bdzeta}$, we first prove some basic facts about the collection $\calB(E)$.

\begin{lemma}\label{lem:ell_at_least_t_and_r}
Suppose $(V,E)$ is an almost-acyclic $k$-uniform hypergraph. Let $t$ be the number of nontrivial connected components of the hypergraph $(V,E)$, and let $r=|E|$. Then for any $B\in \calB(E)$, we have
\[\|\boldsymbol{\Sigma}B\|_{\sfH}\geq \max\left\{2t,\frac{4}{5}r\right\}.\]
\end{lemma}
\begin{proof} Let $V_{1},\dots, V_t$ be the vertex sets of the nontrivial connected components of $(V,E)$, and let $E_{1},\dots,E_{t}$ be the edge sets they induce. For each $i\in [t]$, we write \(\boldsymbol{\Sigma}_i B : =\sum\nolimits_{e\in E_i } B(e).\) Since $V_{1},\dots,V_{t}$ are pairwise disjoint and $\supp(\boldsymbol{\Sigma}_i B) \subseteq V_i$, we have 
\[
    \| \boldsymbol{\Sigma} B\|_{\sfH} =\sum_{i=1}^t \left\lVert \boldsymbol{\Sigma}_i B\right \rVert_{\sfH}. 
\]
Therefore, it suffices to show for each $i\in [t]$ that
\begin{align}\label{ineq:singlecomponenthw}
\left\lVert \boldsymbol{\Sigma}_i B\right \rVert_{\sfH} \geq \max\left\{2,\frac{4}{5}|E_i|\right\}.
\end{align}
If $E_{i}$ consists of a single edge $e$, then by definition we have $\|\boldsymbol{\Sigma}_{i}B\|_{\sfH}=\|B(e)\|_{\sfH}\geq 3$. In the rest of the proof, we fix an $i\in [t]$ and assume $|E_{i}|\geq 2$. Our goal is to show that $\|\boldsymbol{\Sigma}_{i}B\|_{\sfH}\geq \frac{4}{5}|E_{i}|$. Once we have that, it also follows that $\|\boldsymbol{\Sigma}_{i}B\|_{\sfH}\geq \frac{4}{5}\cdot 2=1.6$ and hence $\|\boldsymbol{\Sigma}_{i}B\|_{\sfH}\geq 2$.

We define $V_i^{=1}\subseteq V_i$ to be the set of vertices with degree-1 in the graph $(V_i,E_i)$. For $v\in V_i$, we use $\mathrm{deg}_i (v)$ to denote the degree of vertex $v$ in the graph $(V_i,E_i)$. Since every vertex $v\in V_i \setminus V_i^{=1}$ must satisfy $\mathrm{deg}_i (v)\geq 2$,  
\begin{align*}
    k|E_i|=\sum_{v\in V_{i}}\deg_{i}(v)\geq \left|V_i^{=1}\right| + 2\left|V_i \setminus V_i^{=1}\right| = 2|V_i| - \left|V_i^{=1}\right|. 
\end{align*}
Using the condition $|V_i|\geq (k-1.1)|E_i|$ (which follows from the definition of almost-cyclicity) and rearranging, we get
\begin{align}\label{ineq:degree1vertex}
    \left|V_i^{=1}\right|\geq (k-2.2)|E_i|. 
\end{align}
For any $v\in V_{i}^{=1}$, there is a unique edge $e\in E_{i}$ that contains $v$, so we have $B(e)_{v}=(\boldsymbol{\Sigma}_{i}B)_{v}$. This means
\[
V_{i}^{=1}\setminus \supp(\boldsymbol{\Sigma}_{i}B)\subseteq \bigcup_{e=(v_{1},\dots,v_{k})\in E_{i}}\{v_{1},\dots,v_{k}\}\setminus \supp(B(e)).
\]
Using the condition that $\|B(e)\|_{\sfH}\geq 3$, we get
\[
\left|V^{=1}_{i}\setminus \supp(\boldsymbol{\Sigma}_{i}B)\right|\leq (k-3)|E_{i}|.
\]
Combining this with \eqref{ineq:degree1vertex} yields
\begin{align*}
\|\boldsymbol{\Sigma}_{i}B\|_{\sfH}=|\supp(\boldsymbol{\Sigma}_{i}B)|=\left|V_{i}^{=1}\right|-\left|V^{=1}_{i}\setminus \supp(\boldsymbol{\Sigma}_{i}B)\right|\geq (k-2.2)|E_{i}|-(k-3)|E_{i}|=\frac{4}{5}|E_{i}|,
\end{align*}
as desired.
\end{proof}

\begin{lemma}\label{lem:bound_on_calB_ell}
Let $d,\ell,m$ be positive integers such that $m\geq \ell$. Let $H$ be an almost-acyclic $k$-uniform hypergraph with $m$ edges, and suppose that every edge in $H$ is incident to at most $d$ other edges. Then we have
\[
\sum_{E'\subseteq E}\left|\calB_{\ell}(E')\right|\leq \left(\frac{(15dN^{k})^{4}m}{\ell}\right)^{\ell/2}.
\]
\end{lemma}
\begin{proof}
By Lemma~\ref{lem:ell_at_least_t_and_r}, we know that for any $E'\subseteq E$ such that $\calB_{\ell}(E')$ is nonempty, there exist positive integers $t\leq \ell/2$ and $r\leq 5\ell/4$ such that $E'\in \mathfrak{E}(t,r,H)$. Therefore, we have
\begin{align*}
\sum_{E'\subseteq E}\left|\calB_{\ell}(E')\right|&\leq \sum_{t=1}^{\lfloor \ell/2\rfloor}\sum_{r=t}^{\lfloor 5\ell/4\rfloor}|\mathfrak{E}(t,r,H)|\cdot \max_{E'\in \mathfrak{E}(t,r,H)}\left|\calB(E')\right|\\
&\leq \sum_{t=1}^{\lfloor \ell/2\rfloor}\sum_{r=t}^{2\ell}(6d)^{r}\left(\frac{3m}{t}\right)^{t}N^{kr}\tag{using Lemma~\ref{lem:bound_for_frakE}}\\
&\leq 2\cdot(6dN^{k})^{2\ell}\cdot\sum_{t=1}^{\lfloor \ell/2\rfloor} \left(\frac{3m}{t}\right)^{t}\\
&\leq 2\ell\cdot(6dN^{k})^{2\ell}\left(\frac{3m}{\ell/2}\right)^{\ell/2}\leq \left(\frac{(15dN^{k})^{4}m}{\ell}\right)^{\ell/2},\tag{using $m\geq \ell$}
\end{align*}
as desired.
\end{proof}

We can now state and prove the following analogue of Lemma~\ref{lem:structured_part}.

\begin{lemma}\label{lem:structured_part_2}
    Fix a $\DIHP(G,n,\alpha,K)$  game, where $G=(\calV,\calE,N,(\mu_{\sfe})_{\sfe\in \calE})$. Suppose $\mu_{\sfe}$ is two-wise independent for any $\sfe\in\calE$. If $\bdzeta$ is an almost-acyclic restriction sequence (as per Definition~\ref{def:almost_acyclic_labeled_matchings}) such that $H_{\bdzeta}$ has at most $\gamma n$ edges, then the function $g_{\bdzeta}$ (defined in Definition~\ref{def:g_bdzeta}) is $\big(n,(15kN^{k}K|\calE|)^{4},\gamma kn\log_{2}N,0\big)$-bounded. 
\end{lemma}

\begin{proof}
It suffices to show that $g_{\bdzeta}$ satisfies the three conditions of $\big(n,(15kN^{k}K|\calE|)^{4},\gamma kn\log_{2}N,0\big)$-boundedness required by Definition~\ref{def:Cs_star_bounded}.

\paragraph{Infinity-norm bound.} 
Let $E$ be the edge set of the hypergraph $H_{\bdzeta}$. For every $e\in E$, define a function $g_{e}\in L^2\big(\bZ_N^{\calV\times [n]}\big)$ by letting
\begin{align*}
    g_e (x) = N^{k}\mu_{\langle e\rangle}\big(x_{|e}-\widetilde{\bfz}(e)\big),\qquad\text{for all }x\in\ZmodN^{\calV\times [n]}.
\end{align*}
By \eqref{eq:g_full_expansion} we have $g_{\bdzeta}=\prod_{e\in E}g_{e}$, so
\[
\|g_{\bdzeta}\|_{\infty}\leq \prod_{e\in E}\|g_{e}\|_{\infty}\leq N^{k|E|}\leq N^{\gamma k n}.
\]

\paragraph{Expectation equals $1$.} 
For each edge $e=(v_{1},\dots,v_{k})\in E$, since $\mu_{\langle e\rangle}$ is two-wise independent, we can write 
\[
    g_{e}(x)=1+\sum_{b^{(e)}}\widehat{g_{e}}(b^{(e)})\cdot\chi_{b^{(e)}}(x),
\]
    where the sum is over all $b^{(e)}\in \ZmodN^{\calV\times [n]}$ such that $\supp(b^{(e)})\subseteq \{v_{1},\dots,v_{k}\}$ and $\left|\supp(b^{(e)})\right|\geq 3$.
We can then expand the product in the identity $g_{\bdzeta}=\prod_{e\in E}g_{e}$ and obtain
\[
g_{\bdzeta}(x)=\sum_{E'\subseteq E}\sum_{B\in \calB(E')}\left(\chi_{\boldsymbol{\Sigma} B}(x)\prod_{e\in E'}\widehat{g_{e}}(B(e))\right).
\]
Equivalently, for any $b\in \ZmodN^{\calV\times [n]}$ we have
\begin{equation}\label{eq:two_wise_structured_convolution}
\widehat{g_{\bdzeta}}(b)=\sum_{E'\subseteq E}\sum_{B\in \calB(E')}\left(\ind{\boldsymbol{\Sigma}B=b}\cdot\prod_{e\in E'}\widehat{g_{e}}(B(e))\right).
\end{equation}
For any $E'\subseteq E$, the hypergraph $(\calV\times [n],E')$ is almost-acyclic since it is a subgraph of $H_{\bdzeta}$. Therefore, by applying Lemma~\ref{lem:ell_at_least_t_and_r}, we know that whenever $\boldsymbol{\Sigma}B=\vec{0}$ for some $B\in \calB(E')$, we must have $E'=\emptyset$. It then follows from \eqref{eq:two_wise_structured_convolution} that $\widehat{g_{\bdzeta}}(\vec{0})=1$, or equivalently $\Ex{g_{\bdzeta}}=1$.

\paragraph{Fourier growth bound.} For each edge $e\in E$, since $\|g_{e}\|_{1}=1$, we know that $\left|\widehat{g_{e}}(b)\right|\leq 1$ for any $b\in \ZmodN^{\calV\times [n]}$. Therefore, it follows from \eqref{eq:two_wise_structured_convolution} that for any positive integer $\ell$, we have
\[
\left\|g_{\bdzeta}^{=\ell}\right\|_{\sfW}\leq \sum_{E'\subseteq E}\sum_{B\in \calB(E')}\ind{\|\boldsymbol{\Sigma}B\|_{\sfH}=\ell\big.}=\sum_{E'\subseteq E}\left|\calB_{\ell}(E')\right|
\]
Since $E$ is the union of $K|\calE|$ partial matchings, in the hypergraph $H_{\bdzeta}=(\calV\times [n],E)$, every vertex is incident to at most $K|\calE|$ edges. Consequently, every edge is incident to at most $kK|\calE|$ other edges. We may now apply Lemma~\ref{lem:bound_on_calB_ell} and obtain
\[
\left\|g_{\bdzeta}^{=\ell}\right\|_{\sfW}\leq \left(\frac{(15kN^{k}K|\calE|)^{4}\cdot \gamma n}{\ell}\right)^{\ell/2}\leq F_{(15kN^{k}K|\calE|)^{4}}\bigl(n,\ell,\gamma^{2}n\bigr)\leq F_{(15kN^{k}K|\calE|)^{4}}\bigl(n,\ell,\gamma n\log_{2}N\bigr),
\]
where we used Proposition~\ref{prop:Fourier_bound_non_decreasing} in the last transition.
\end{proof}

\subsubsection{Finishing up the Proofs}\label{subsubsec:induction_twowise}

As noted earlier, Lemma~\ref{lem:discrepancy_twowise} and Theorem~\ref{thm:twowise} can now be proved using essentially the same arguments as their counterparts, Lemma~\ref{lem:discrepancy_bound} and Theorem~\ref{thm:sqrt_lower_bound}. 
To avoid unnecessary repetition, we provide only brief proof sketches.

\begin{proof}[Proof Sketch of Lemma~\ref{lem:discrepancy_twowise}]
By the first and second condition of $\gamma n$-fairness, the restriction sequence $\bdzeta$ satisfies the assumptions of Lemma~\ref{lem:structured_part_2}. The conclusion of Lemma~\ref{lem:structured_part_2} can then be used as the base case for an induction argument similar to the one in the proof of Lemma~\ref{lem:discrepancy_bound}. 

More specifically, one can use Lemmas~\ref{lem:relating_yes_no} and~\ref{lem:separating_structure_pseudorandom} to express the ratio $\calD_{\yes}(R)/\calD_{\no}(R)$ as an expectation of the form (see \eqref{eq:relating_yesno})
\begin{equation}\label{eq:yes_no_ratio_two_wise}
\Exu{(M_{i})_{i=1}^{K|\calE|}}{\Exu{x}{g_{\bdzeta}(x)\prod_{i=1}^{K|\calE|}h_{i,M_{i}}(x)}},
\end{equation}
where each $h_{i,M_{i}}$ represents the ``pseudorandom part'' of the $i$-th player's messages. Next use Lemma~\ref{lem:induction} to inductively prove the following statement: if $n\gg \gamma^{-1}\gg 1$, then for any $r\in \{0,1,\dots,K|\calE|\}$, the function $g_{\bdzeta}\prod_{i=1}^{r}h_{i,M_{i}}$ is
\[
\Big(n,\,\bigl(15kN^{k}K|\calE|\bigr)^{4}\bigl(10^{16}N^{8k}K^{2}|\calE|^{2}\bigr)^{r},\,2^{r+1}\gamma n\log_{2}N,\,10^{-4}K^{-1}|\calE|^{-1}r\Big)\text{-bounded}
\]
with probability at least $1-4r/10^{4}$ over the randomness in $M_{1},\dots,M_{r}$. The $r=K|\calE|$ case of this statement then yields a lower bound for \eqref{eq:yes_no_ratio_two_wise}. 
\end{proof}

\begin{proof}[Proof Sketch of Theorem~\ref{thm:twowise}]
We apply Lemmas~\ref{lem:decomposition_twowise} and~\ref{lem:discrepancy_twowise} to obtain constants $\eta$ and $\gamma$, respectively. We show that any protocol $\Pi$ for $\DIHP(G,n,\alpha,K)$ with $|\Pi|\leq \min\{\eta,10^{-3}\gamma\}\cdot n$ must have $\adv(\Pi)<0.1$.

We apply Lemma~\ref{lem:decomposition_twowise} to $\Pi$, and let $\calR$ be the collection of structured rectangles obtained. By the upper bound on $|\Pi|$ and conclusion (2) of Lemma~\ref{lem:decomposition_twowise}, we know that each structured rectangle in $\calR$ is $\gamma n$-fair. It then follows from Lemma~\ref{lem:discrepancy_bound} that $\calD_{\yes}(R)\geq (1-10^{-3})\cdot \calD_{\no}(R)$ for each $(\bdzeta,R)\in \calR$. The rest of the proof of identical to the proof of Theorem~\ref{thm:sqrt_lower_bound}.
\end{proof}

\subsection{Implications of Two-Wise Independent DHIP}\label{subsec:implication_twowise}

In this subsection, we sketch the proofs of Theorems~\ref{thm:multi_pass_lin} and~\ref{thm:unbounded_width}. For the proof of Theorem~\ref{thm:multi_pass_lin}, we need to open up the proof of Lemma~\ref{lem:communication_reduction} (given in \cite[Proof of Lemma 5.14]{FMW25b}) as a white box.

\begin{proof}[Proof Sketch of Theorem~\ref{thm:multi_pass_lin}]
Given a constant $\varepsilon'\in(0,1)$ and predicate family $\calF\subseteq \{f:\Sigma^{k}\rightarrow\{0,1\}\}$ such that every predicate $f\in \calF$ supports two-wise independence, the proof proceeds as follows:

\begin{enumerate}
\item We can pick an instance $\calI\in\cspF$ with $\val_{\calI}< \rho(\calF)+\varepsilon'$ (by the definition~\eqref{eq:def_rho_F}). By Proposition~\ref{prop:two_wise_LP_val_1} we know that there is a solution to $\lp_{\calI}$ achieving objective value 1 in which the distribution over $\Sigma^{k}$ associated to every constraint is two-wise independent.

\item Then we use Lemma~\ref{lem:communication_reduction} (more precisely, the proof of this lemma given in \cite{FMW25b}) with \[\varepsilon=\rho(\calF)+\varepsilon'-\val_{\calI}\]
to convert $\calI$ into a $\DIHP(G,n,\alpha,K)$ communication game. Because of the two-wise independence of the distributions in the $\lp$ solution, the distribution-labeled $k$-graph $G=(\calV,\calE,N,(\mu_{\sfe}))_{\sfe\in \calE}$ provided by the proof of Lemma~\ref{lem:communication_reduction} satisfies $N=|\Sigma|$ and, importantly, that $\mu_{\sfe}$ is \emph{two-wise} independent for any $\sfe\in \calE$.
\end{enumerate}

By the conclusion of Lemma~\ref{lem:communication_reduction}, any $p$-pass streaming algorithm for $\McspF{1}{\rho_{\calF}+\varepsilon'}$ must use $(pK|\calE|)^{-1}\cdot \CC(G,n,\alpha,K)$ bits of memory on instances with $O(n)$ variables and $O(n)$ constraints. By the two-wise independence of the distributions $\mu_{\sfe}$ (for $\sfe\in\calE$), it follows from Theorem~\ref{thm:twowise} that $\CC(G,n,\alpha,K)\geq \Omega(n)$. Therefore, any $p$-pass streaming algorithm for $\McspF{1}{\rho(\calF)+\varepsilon'}$ must use $\Omega(n/p)$ bits of memory.
\end{proof}

For any predicate family $\calF \subseteq \{f : \Sigma^{k} \to \{0,1\}\}$ of unbounded width, the $\Omega(n)$ query lower bound for $\McspF{1}{1-\varepsilon}$ established in \cite{fei2025unbounded} relies on a known polynomial-time reduction (due to \cite{dalmau2013robust}) from a problem of the form $\mathsf{MaxCSP}(\calF')[1,\rho(\calF')+\varepsilon]$ to $\McspF{1}{1-\varepsilon}$, where $\calF'$ consists of predicates supporting two-wise independence. 

As shown in \cite[Section~5]{fei2025unbounded}, this reduction can be adapted to the bounded-degree query model. Moreover, it is straightforward to verify that the same reduction extends to the streaming setting, thereby deriving Theorem~\ref{thm:unbounded_width} from Theorem~\ref{thm:twowise}. 

Due to the technical length of the reduction, we omit the full details and provide only a proof sketch of Theorem~\ref{thm:unbounded_width}.

\begin{definition}
Let $\mathsf{G}$ be an Abelian group. For any element $b\in \mathsf{G}$, we let $S^{\mathsf{G}}_{b}:\mathsf{G}^{3}\rightarrow\{0,1\}$ be the predicate defined by
\[S^{\mathsf{G}}_{b}(x_{1},x_{2},x_{3})=1\quad\text{if and only if}\quad
x_{1}+x_{2}+x_{3}=b.
\]
We use \(\textsc{3Sum}_{
\mathsf{G}}\) to denote the predicate family $\left\{S_{b}^{\mathsf{G}}\,\middle|\,b\in \mathsf{G}\right\}$.
\end{definition}

\begin{proof}[Proof Sketch of Theorem~\ref{thm:unbounded_width}]
Given a predicate family $\calF\subseteq \{f:\Sigma^{k}\rightarrow\{0,1\}\}$ of unbounded width, the proof proceeds as follows:
\begin{enumerate}
\item We first use the result of \cite{dalmau2013robust} (see also the exposition in \cite{fei2025unbounded}) to obtain an Abelian group $\mathsf{G}$ such that the predicate family $\textsc{3Sum}_{\mathsf{G}}$ can be ``simulated'' by $\calF$.\footnote{For the precise notion of ``simulated,'' see \cite[Definition 2.1]{fei2025unbounded}. Technically, the relational structure $(\mathsf{G},\textsc{3Sum}_{\mathsf{G}})$ should be simulated by ``the core of $(\Sigma,\mathcal{F})$ appended with constant relations.'' We refer to \cite{fei2025unbounded} for the details.}
\item Then pick an instance $\calI\in\mathrm{CSP}(\textsc{3Sum}_{\mathsf{G}})$ with $\val_{\calI}<2/3$.\footnote{It is easy to see that such an instance $\calI$ exist (see e.g. \cite[Lemma 4.5]{fei2025unbounded}).} Since every predicate in the family $\textsc{3Sum}_{\mathsf{G}}$ supports two-wise independence, by Proposition~\ref{prop:two_wise_LP_val_1} we know that there is a solution to $\lp_{\calI}$ achieving objective value 1 in which the distribution over $\Sigma^{k}$ associated to every constraint is two-wise independent.

\item Use Lemma~\ref{lem:communication_reduction} (more precisely, the proof of this lemma given in \cite{FMW25b}) with $\varepsilon=2/3-\val_{\calI}$ to convert $\calI$ into a $\DIHP(G,n,\alpha,K)$ communication game. Because of the fact that $f^{-1}(1)$ is two-wise independent for any $f\in \textsc{3Sum}_{\mathsf{G}}$, the distribution-labeled $k$-graph $G=(\calV,\calE,N,(\mu_{\sfe}))_{\sfe\in \calE}$ provided by the proof of Lemma~\ref{lem:communication_reduction} satisfies $N=|\mathsf{G}|$ and, importantly, that $\mu_{\sfe}$ is \emph{two-wise} independent for any $\sfe\in \calE$.

\item The proof of Lemma~\ref{lem:communication_reduction} also provides a reduction map from the joint input space of $\DIHP(G,n,\alpha,K)$ to $\mathrm{CSP(\textsc{3Sum}_{\mathsf{G}})}$:
\[
\varphi_{n}:\prod_{(\sfe,j)\in \calE\times [K]}\Omega^{\calU_{\sfe},\alpha n}\longrightarrow\mathrm{CSP(\textsc{3Sum}_{\mathsf{G}})},
\] such that
\[
\Pru{\bfY\sim \calD_{\yes}}{\val_{\varphi_{n}(\bfY)}=1}=1\qquad\text{and}\qquad \Pru{\bfY\sim \calD_{\yes}}{\val_{\varphi_{n}(\bfY)}\leq \frac{2}{3}}\geq 1-o(1),
\]
where $o(1)$ denotes a term tending to 0 as $n\rightarrow +\infty$.

\item Use \cite[Section 5]{fei2025unbounded} to construct a reduction map $\psi_{n}$ from the image of $\varphi_{n}$ to $\cspF$, such that
\[
\Pru{\bfY\sim \calD_{\yes}}{\val_{\psi_{n}(\varphi_{n}(\bfY))}=1}=1\qquad\text{and}\qquad \Pru{\bfY\sim \calD_{\yes}}{\val_{\psi_{n}(\varphi_{n}(\bfY))}\leq 1-\varepsilon'}\geq 1-o(1),
\]
for some constant $\varepsilon'\in (0,1)$ that does not depend on $n$. The composed reduction map $\psi_{n}\circ \varphi_{n}$ satisfies the following:
\begin{itemize}
\item There is a fixed variable set $V_{n}$ with size $O(n)$, and a fixed collection $\calC_{n}^{(0)}$ of constraints on $V_{n}$. 
\item For each player $(\sfe,j)\in\calE\times [K]$, there is a deterministic map that generates from any input $\bfy\in\Omega^{\calU_{\sfe},\alpha n}$ a collection of $O(n)$ constraints on $V_{n}$.
\item The constraint sequence of the instance $\psi_{n}(\varphi_{n}(\bfY))$ is the concatenation of $\calC^{(0)}_{n}$ and the $K|\calE|$ constraint sequences generated by individual players. 
\end{itemize}
\end{enumerate}
From the last two steps above, we can conclude that for sufficiently large $n$, any $p$-pass streaming algorithm for $\McspF{1}{1-\varepsilon'}$ must use $(pK|\calE|)^{-1}\cdot \CC(G,n,\alpha,K)$ bits of memory on input instances with $O(n)$ variables and $O(n)$ constraints. By the two-wise independence of the distributions $\mu_{\sfe}$ (for $\sfe\in\calE$), it follows from Theorem~\ref{thm:twowise} that $\CC(G,n,\alpha,K)\geq \Omega(n)$. Therefore, any $p$-pass streaming algorithm for $\McspF{1}{1-\varepsilon'}$ must use $\Omega(n/p)$ bits of memory.
\end{proof}

%% file: appendices.tex
\section{Some Simple Inequalities}

\begin{proposition}\label{prop:basic_calculus}
For any real numbers $x,\ell$ such that $x\geq 0$ and $\ell\geq 2$, we have \[(1+x)^{\ell}-\ell x\leq (1+2\ell x^{2})^{\ell/2}.\]
\end{proposition}
\begin{proof}
If $x\geq 2/(2\ell-1)$, then $(1+x)^{2}\leq 1+2\ell x^{2}$ and hence $(1+x)^{\ell}\leq (1+2\ell x^{2})^{\ell/2}$.

If $x\leq 2/(2\ell -1)$, then by taking second derivatives one can show that $(1+x)^{\ell}\leq 1+\ell x+\ell^{2}x^{2}$, and hence $(1+x)^{\ell}-\ell x\leq 1+\ell^{2}x^{2}\leq (1+2\ell x^{2})^{\ell/2}$.
\end{proof}

\begin{proposition}\label{prop:calculus}
For real numbers $a,b,c,d$ such that $0<a\leq b\leq c\leq d$ and $a+d=b+c$, we have
\[
1\leq \frac{a^{a}d^{d}}{b^{b}c^{c}}\leq 2^{d}.
\]
\end{proposition}
\begin{proof}
Let $r=b-a=d-c$. We have
\begin{equation}\label{eq:logarithm_integral}
\ln\frac{a^{a}d^{d}}{b^{b}c^{c}}=\ln\frac{d^{d}}{c^{c}}-\ln\frac{b^{b}}{a^{a}}=\int_{a}^{c}\left(\frac{\d}{\d x}\ln\frac{(x+r)^{x+r}}{x^{x}}\right)\d x=\int_{a}^{c}\ln\frac{x+r}{x}\d x\geq 0,
\end{equation}
and therefore $a^{a}d^{d}b^{-b}c^{-c}\geq 1$. To obtain an upper bound on $a^{a}d^{d}b^{-b}c^{-c}$, we extend the interval of integral in \eqref{eq:logarithm_integral} from $[a,c]$ to $[0,c]$, and then divide it at the point $r$:
\begin{align*}
\ln\frac{a^{a}d^{d}}{b^{b}c^{c}}&\leq 
\int_{0}^{r}\ln\frac{x+r}{x}\d x+\int_{r}^{c}\ln\frac{x+r}{x}\d x
\\
&\leq \int_{0}^{r}\ln\frac{x+r}{x}\d x+\int_{r}^{c}\ln2\d x\\
&= r\cdot \int_{0}^{1}\ln\left(\frac{x+1}{x}\right)\d x+(c-r)\ln 2\\
&=r\ln 4+(c-r)\ln 2=(c+r)\ln 2=d\ln 2.
\end{align*}
Therefore, we have $a^{a}d^{d}b^{-b}c^{-c}\leq 2^{d}$, as desired.
\end{proof}